\documentclass{article}

\usepackage[english]{babel}
\usepackage[latin1]{inputenc}
\usepackage{amsmath,amssymb}
\usepackage{mathrsfs}

\usepackage[]{todonotes}

\usepackage[final]{hyperref}

\newtheorem{theorem}{Theorem}[section]
\newtheorem{corollary}[theorem]{Corollary}
\newtheorem{lemma}[theorem]{Lemma}

\newtheorem{fact}[theorem]{Fact}
\newtheorem{definition}[theorem]{Definition}
\newtheorem{example}[theorem]{Example}
\newtheorem{remark}[theorem]{Remark}
\newenvironment{proof}{\noindent {\bf Proof}.\ }{\ \\}
\newenvironment{myitemize}{\begin{list}{$\bullet$}{\setlength{\leftmargin}{10pt}
\setlength{\itemindent}{0.1\labelwidth}}} {\end{list}}

\newcommand{\spath}[1]{\mbox{[$#1$]-}path}
\newcommand{\adj}[1]{\mbox{[$#1$]-}adjacent}
\newcommand{\component}[1]{\mbox{[$#1$]-}component}
\newcommand{\components}[1]{\mbox{[$#1$]-}components}

\newcommand{\connected}[1]{\mbox{[$#1$]-}connected}
\newcommand{\touches}[1]{\mbox{[$#1$]-}touches}

\newcommand{\nodes}{\mathit{nodes}}
\newcommand{\edges}{\mathit{edges}}
\newcommand{\HD}{H\!D}
\newcommand{\HG}{{\cal H}}
\newcommand{\JT}{J\!T}

\newcommand{\node}{{\mathcal N}}
\renewcommand{\root}{\mathit{root}}
\newcommand{\vertices}{\mathit{vertices}}

\newcommand{\Pol}{\mbox{\rm P}}
\newcommand{\NP}{\mbox{\rm NP}}

\newcommand{\SL}{\mbox{\rm SL}}

\newcommand{\A}{\mathcal{A}}
\newcommand{\B}{\mathcal{B}}

\newcommand{\fontOn}{\tiny}
\newcommand{\U}{\mathcal{U}}

\newcommand{\DB}{{\rm \mbox{\rm DB}}}
\newcommand{\barDB}{\overline{\rm \mbox{\rm DB}}}
\newcommand{\DS}{{\mbox{$\mathcal DS$}}}
\newcommand{\WDS}{{\mbox{$\mathcal A(DS)$}}}
\newcommand{\onDB}{{\mbox{\rm \tiny DB}}}
\newcommand{\onbarDB}{{\overline{\mbox{\rm \tiny DB}}}}
\newcommand{\onDBp}{{\mbox{\rm \fontOn DB}'}}
\newcommand{\onDBpp}{{\mbox{\rm \fontOn DB}''}}
\newcommand{\onDBpiu}{{\mbox{\rm \fontOn DB}^+}}
\newcommand{\onDBred}{{\fontOn{\it \fontOn red}({\mbox{\rm \fontOn DB}},\W)}}

\newcommand{\red}{\mathsf{red}}
\newcommand{\vars}{\mathit{vars}}
\newcommand{\atoms}{\mathit{atoms}}
\newcommand{\atom}{\mathit{atom}}

\newcommand{\views}{\mathit{views}}

\newcommand{\chase}{{\tt chase}}
\newcommand{\fchase}{{\tt fchase}}

\newcommand{\cores}{{\tt cores}}
\newcommand{\homEquiv}{\approx_{\tt hom}}
\newcommand{\V}{\mathcal{V}}
\newcommand{\W}{\mathcal{V}}

\newcommand{\HGUnoSet}{\HG_{Q'}}

\newcommand{\HGDue}{\HG_\mathcal{V}}
\newcommand{\lc}{\mathsf{lc}}

\newcommand{\gc}{\mathsf{gc}}

\renewcommand{\A}{\mathbb{A}}
\renewcommand{\B}{\mathbb{B}}

\newcommand{\lDM}{v\mbox{-\tt DM}}
\newcommand{\rDM}{d\mbox{-\tt DM}}

\newcommand{\CR}{\mbox{\rm R\&C}}
\newcommand{\C}{\mathcal{C}}
\newcommand{\E}{\mathrm{ED}}
\newcommand{\F}{\mathrm{Fr}}
\newcommand{\ecomponent}[1]{\mbox{[$#1$]-}option}
\newcommand{\ecomponents}[1]{\mbox{[$#1$]-}options}

\newcommand{\iecomponent}[1]{\mbox{\em [$#1$]-}option}

\newcommand{\tuple}[1]{\langle#1\rangle}

\newcommand{\tpCovered}{\mbox{\it tp-covered}}

\newcommand{\vcovers}{\mbox{\it covers}}

\newcommand{\FPT}{\mbox{\rm FPT}}

\newcommand{\nop}[1]{}

\newcommand{\longv}[1]{}

\begin{document}

%

\title{Tree Projections and Structural Decomposition Methods: The Power of Local Consistency and
Larger Islands of Tractability}

\author{Gianluigi Greco and Francesco Scarcello\\
\ \\
\small
  Dept. of Mathematics and DEIS,\\
\small  University of Calabria, Via P. Bucci 30B, 87036, Rende, Italy\\
\small  {\tt \{ggreco\}@mat.unical.it}, {\tt \{scarcello\}@deis.unical.it} }

\date{}

\maketitle

\begin{abstract}
Evaluating conjunctive queries and solving constraint satisfaction problems are fundamental problems in database theory and artificial
intelligence, respectively. These problems are {\rm NP}-hard, so that several research efforts have been made in the literature for identifying
tractable classes, known as islands of tractability, as well as for devising clever heuristics for solving efficiently real-world instances.

Many heuristic approaches are based on enforcing on the given instance a property called {\em local consistency}, where (in database terms)
each tuple in every query atom matches at least one tuple in every other query atom. Interestingly, it turns out that, for many well-known
classes of queries, such as for the acyclic queries, enforcing local consistency is even sufficient to solve the given instance correctly.
However, the precise power of such a procedure was unclear, but for some very restricted cases.

The paper provides full answers to the long-standing questions about the precise power of algorithms based on enforcing local consistency. In
particular, the paper deals with both the general framework of {\em tree projections}, where local consistency is enforced among arbitrary
views defined over the given database instance, and the specific cases where such views are computed according to so-called {\em structural
decomposition methods}, such as generalized hypertree width, component hypertree decompositions, and so on.

The classes of instances where enforcing local consistency turns out to be a correct query-answering procedure are however not efficiently
recognizable. In fact, the paper finally focuses on certain subclasses defined in terms of the novel notion of {\em greedy tree projections}.
These latter classes are shown to be efficiently recognizable and strictly larger than most islands of tractability known so far, both in the
general case of tree projections and for specific structural decomposition methods.
\end{abstract}

\raggedbottom

\section{Introduction}

\subsection{Acyclic Conjunctive Queries}\label{sec:acyclic}

Answering conjunctive queries to relational databases is a basic problem in database theory, and it is equivalent to many other fundamental
problems, such as conjunctive query containment and constraint satisfaction. Recall that conjunctive queries are defined through conjunctions
of atoms (without negation), and are known to be equivalent to Select-Project-Join queries. The problem of evaluating such queries is
$\NP$-hard in general, but it is feasible in polynomial time on the class of acyclic queries (we omit ``conjunctive,'' hereafter), which was
the subject of many seminal research works since the early ages of database theory (see, e.g.,~\cite{BFMY83}). This class contains all queries
$Q$ whose associated query hypergraph $\HG_Q$ is acyclic,\footnote{For completeness, observe that different notions of hypergraph acyclicity
have been proposed in the literature. This paper follows the standard definition of acyclic conjunctive queries, so that hypergraph acyclicity
always refers to the most liberal notion, known as $\alpha$-acyclicity~\cite{fagi-83}.} where $\HG_Q$ is a hypergraph having the variables of
$Q$ as its nodes, and the (sets of variables occurring in the) atoms of $Q$ as its hyperedges.
It is well known that acyclic queries enjoy a number of highly desirable properties, recalled next.

First, {\em acyclic queries can be efficiently solved}. From any acyclic query, we can build (in linear time) a \emph{join
tree}~\cite{bern-good-81}, which is a tree whose vertices correspond to the various atoms and where the subgraph induced by vertices containing
any given variable is a tree. According to Yannakakis's algorithm~\cite{yann-81}, {Boolean} acyclic queries can be evaluated by processing any
of their join trees bottom-up, by performing upward semijoins between the relations associated with the query atoms,
thus keeping the size of the intermediate relations small. At the end, if the relation associated with the root of the join tree is not empty,
then the answer of the query is not empty. For non-Boolean queries, after the bottom-up step described above, one can perform the opposite
top-down step by filtering each child vertex from those tuples that do not match with its parent tuples. The filtered database, called
\emph{full reducer}, then enjoys the {\em global consistency} property:
every tuple in every relation participates in some solution. By exploiting this property, all solutions can be computed with a backtrack-free
procedure (i.e., with backtracks used to look for further solutions, and never caused by wrong choices).

Second, \emph{the class of acyclic instances coincides with the class of queries where local consistency entails global consistency}. We say
that local (also, pairwise) consistency holds if the relations associated with the query atoms are not empty and we do not miss any tuple by
taking semijoins between any pair of them.
The acyclic instances that fulfil this property also fulfil the global consistency property~\cite{BFMY83}. Note that local consistency may
easily be enforced by taking the semijoins between all pairs of atoms until a fixpoint is reached. Therefore, in abstract terms, any acyclic
query can be answered by means of ``local'' computations only, without any additional knowledge about the whole structure, in particular
without computing any join tree of the query. In addition, and more surprisingly, if a class of instances can be answered by means of this
approach, then it only contains acyclic instances~\cite{BFMY83}.\footnote{Actually, this classical result holds only for queries where every
relation symbol is used at most once. The precise power of local computations in the general case is identified in this paper (for acyclic
queries too).}

Finally, \emph{acyclicity is efficiently recognizable}. Deciding whether a hypergraph is acyclic is feasible in linear
time~\cite{tarj-yann-84}, and also in deterministic logspace. In fact, this latter property follows from the fact that hypergraph acyclicity
belongs to $\SL$~\cite{gott-etal-01}, and that $\SL$ is equal to deterministic logspace~\cite{rein-04}. Note that, in the light of this
property and the first one above, these queries identify a so-called (accessible) ``island of tractability'' for the query answering
problem~\cite{K03}.

\subsection{Generalization of Acyclicity}

Queries arising from real applications are hardly precisely acyclic. Yet, they are often not very intricate and, in fact, tend to exhibit some
limited degree of cyclicity, which suffices to retain most of the nice properties of acyclic ones.

Several efforts have been spent to investigate invariants that are best suited to identify nearly-acyclic hypergraphs, leading to the
definition of a number of so-called {\em (purely) structural decomposition-methods}, such as the \emph{(generalized)
hypertree}~\cite{gott-etal-99}, \emph{fractional hypertree}~\cite{GM06}, \emph{spread-cut}~\cite{CJG08}, and \emph{component
hypertree}~\cite{GMS07} decompositions. These methods aim at transforming a given cyclic hypergraph into an acyclic one, by organizing its
edges (or its nodes) into a polynomial number of clusters, and by suitably arranging these clusters as a tree, called decomposition tree. The
original problem instance can then be evaluated over such a tree of subproblems, with a cost that is exponential in the cardinality of the
largest cluster, also called {\em width} of the decomposition, and polynomial if this width is bounded by some constant.

Despite their different technical definitions, there is a simple mathematical framework that encompasses all the above decomposition methods,
which is the framework of the \emph{tree projections}~\cite{GS84}. In this setting, a query $Q$ is given together with a set $\V$ of atoms,
called views, which are defined over the variables in $Q$.
The question is whether (parts of) the views can be arranged as to form a tree projection (playing the role of a decomposition tree), i.e., a
novel acyclic query that still ``covers'' $Q$.
By representing $Q$ and $\V$ via the hypergraphs $\HG_Q$ and $\HG_\V$, where hyperedges one-to-one correspond with query atoms and views,
respectively, the tree projection problem reveals its graph-theoretic nature. For a pair of hypergraphs $\HG_1,\HG_2$, let $\HG_1\leq \HG_2$
denote that each hyperedge of $\HG_1$ is contained in some hyperedge of $\HG_2$. Then, a tree projection of $\HG_Q$ w.r.t.~$\HG_\V$ is any
acyclic hypergraph $\HG_a$ such that $\HG_Q\leq \HG_a\leq \HG_\V$. If such a hypergraph exists, then we say that the pair of hypergraphs
$(\HG_Q,\HG_\V)$ has a tree projection.

\begin{figure}[t]
  \centering
  \includegraphics[width=0.99\textwidth]{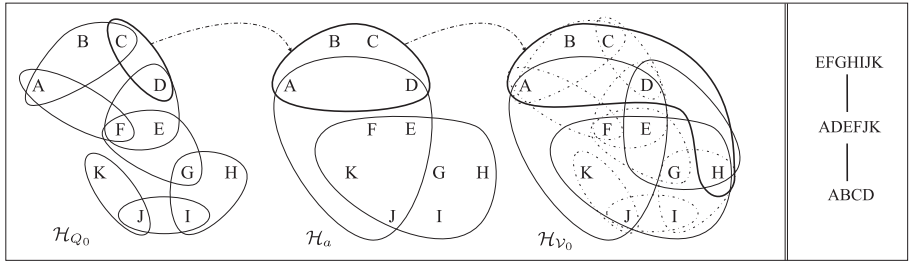}
  \caption{A tree projection $\HG_a$ of $\HG_{Q_0}$ w.r.t.~$\HG_{\V_0}$; 
  On the right: A join tree $\JT_a$ for $\HG_a$.
  }\label{fig:hypergraph}
\end{figure}

\begin{example}\em Consider the conjunctive query
\vspace{-1mm}
\[
\begin{array}{ll}
   Q_0:  & r_1(A,B,C)\wedge r_2(A,F)\wedge r_3(C,D)\wedge r_4(D,E,F)\wedge\\
            & r_5(E,F,G) \wedge r_6(G,H,I)\wedge r_7(I,J)\wedge r_8(J,K),\\
\end{array}
\]

\vspace{-1mm}\noindent whose associated hypergraph $\HG_{Q_0}$ is depicted in Figure~\ref{fig:hypergraph}, together with other hypergraphs that
are discussed next.

To answer $Q_0$, assume that a set $\V_0$ of views is available comprising some views, called {\em query views}, playing the role of query
atoms, plus four additional views. The set of variables of each view is a hyperedge in the hypergraph $\HG_{\V_0}$ (query views are depicted as
dashed hyperedges).
In the middle between $\HG_{Q_0}$ and $\HG_{\V_0}$, Figure~\ref{fig:hypergraph} reports the hypergraph $\HG_a$ which covers $\HG_{Q_0}$, and
which is in its turn covered by $\HG_{\V_0}$---e.g., $\{C,D\}\subseteq \{A,B,C,D\}\subseteq\{A,B,C,D,H\}$. Since $\HG_a$ is in addition acyclic
(just check the join tree $\JT_a$ in the figure), $\HG_a$ is a tree projection of $\HG_{Q_0}$ w.r.t.~$\HG_{\V_0}$.~\hfill~$\lhd$
\end{example}

Observe that, in the tree projection framework, views can be arbitrary, i.e, they do not depend on the specific conjunctive query $Q$, and can
be reused to answer different queries. In particular, views may be the materialized output of any procedure over the database, possibly much
more powerful than conjunctive queries.
Moreover, it is known and easy to see that any decomposition method based on clustering subproblems can be viewed as an instance of this
general setting, identifying  a specific set of views to answer a given query $Q$ efficiently (see Section~\ref{sec:framework} and
Section~\ref{SDM}).

For example (see, e.g., \cite{adler08,GS08,GS10}), for any fixed natural number $k$, the {generalized hypertree decomposition method}
associates with any query $Q$ a set $\it v\mbox{-}hw_k(Q)$ of views, containing one distinct view over each set of variables that can be
covered by at most $k$ query-atoms.~
For any hypergraph $\HG$, let $\HG^k$ be the hypergraph whose hyperedges are all possible sets obtained by the union of at most $k$ hyperedges
of $\HG$, and notice that $\HG_Q^k$ is precisely the hypergraph associated with $\it v\mbox{-}hw_k(Q)$. A query $Q$ has \emph{generalized
hypertree width} bounded by $k$ {if, and only if,} there is a tree projection of $\HG_Q$ w.r.t.~$\HG_Q^k$.

For another example, we recall the \emph{tree decomposition} method~\cite{DP89,Fre90}, based on the notion of treewidth~\cite{RS84}, which is
the most general decomposition method over classes of \emph{bounded-arity} queries (see, e.g, \cite{gott-etal-00,G07}). For any fixed natural
number $k$, the method defines the set $\it v\mbox{-}tw_k(Q)$ of views containing one distinct view over each set of at most $k+1$ variables
occurring in $Q$. Let $\HG_Q^{tk}$ be the hypergraph associated with $\it v\mbox{-}tw_k(Q)$, i.e., the hypergraph whose hyperedges are all
possible sets of at most $k+1$ variables. Then, a query $Q$ has \emph{treewidth} bounded by $k$ {if, and only if,} there is a tree projection
of $\HG_Q$ w.r.t.~$\HG_Q^{tk}$ (see, e.g., \cite{GS08,GS10}).

In fact, the notion of tree projection is quite natural and may be exploited in different applications where hypergraphs naturally represent
structural properties of input instances. For example, Adler~\cite{adler-thesis} pointed out that the notion of acyclicity for a conjunctive
query with negation $Q$, as defined in~\cite{FFG02}, can be immediately recast as the existence of a tree projection of $\HG_Q$ w.r.t.
$\HG_{Q^+}$, where the hyperedges of $\HG_{Q^+}$ are the sets of variables occurring in the positive atoms of $Q$ only, while the hyperedges of
$\HG_{Q}$ correspond to all atoms, including the negative ones. Then, we can generalize this notion to obtain larger classes of tractable
instances, by saying that a query with negation $Q$ has {\em generalized hypertree width} at most $k$ if the pair $(\HG_Q,\HG_{Q^+}^k)$ has a
tree projection. Indeed, following the same reasoning as in~\cite{FFG02}, it is easy to see that, given such a tree projection, the query $Q$
can be evaluated in polynomial time.

\subsection{Open Questions About Tree Projections and Structural Decomposition Methods}

The interest on the tree projection framework goes back to the eighties, when it was noticed that queries that admit a tree projection can be
evaluated in polynomial time~\cite{GS84} (see, also, \cite{SS93}). Thus, tree projections smoothly preserve the first crucial property of
acyclic queries discussed in Section~\ref{sec:acyclic}.
Our knowledge on the preservation of the other properties of acyclic queries was less clear, instead.
In fact, the following two questions have been posed in the literature for the general tree projection framework as well as for structural
decomposition methods specifically tailored to deal with classes of queries without a fixed arity bound. Such questions were in particular open
for the generalized hypertree decomposition method, which on classes of unbounded-arity queries is a natural counterpart of the tree
decomposition method.

\medskip

\noindent \textbf{(Q1) What is the precise power of local-consistency based algorithms?} This question was firstly raised in~\cite{BFMY83} and
specifically for the general case of tree projections in~\cite{SS93}, and remained open so far, despite it was attacked via different
approaches and proof techniques, which gave some partial results, reported below.

Let $\W$ be an arbitrary set of views, which also contains the query views representing the atoms of a given query $Q$. Let $\lc(\W,\DB)$
denote that the views in $\V$ evaluated over a database $\DB$ enjoy the local consistency property, i.e., they are non-empty and we do not miss
any tuple by taking the semijoin between any pair of views.
Let $\red(\W,\DB)$ be the \emph{reduct} of $\DB$ according to $\W$, computed by taking all possible semijoins until a fixpoint is reached. More
precisely, $\red(\W,\DB)$ is the (set-inclusion) maximal subset of $\DB$ such that $\lc(\W,\DB)$ holds, or $\red(\W,\DB)=\emptyset$, whenever
such a maximal subset does not exist.
Let $\gc(\W,\DB,Q)$ denote that the global consistency property holds, i.e., every tuple in every query view (evaluated over $\DB$)
participates in the query answer. Let $Q^\onDB\neq \emptyset$ denote that the answer of $Q$ on $\DB$ is not empty. Then, the picture emerging
from the literature is as follows:

\begin{myitemize}
\item[--] The existence of a tree projection of $\HG_Q$ w.r.t.~$\HG_\V$ entails that, $\forall \DB$, $\lc(\W,\DB)\Rightarrow
    \gc(\W,\DB,Q)$~\cite{SS93}. In words, the existence of a tree projection is a sufficient condition for the global
    consistency property to hold, whenever the database is local consistent. Thus, if a tree projection exists, then both
    deciding whether the query is not empty and computing a query answer (if any) are feasible in polynomial time, by enforcing local
    consistency. Observe that such a procedure is based on local computations only, and hence there is no need to actually compute a tree
    projection. This is a remarkable result, since computing a tree projection is instead not feasible in polynomial time, unless
    $\Pol=\NP$~\cite{GMS07}.
    \emph{It was conjectured that the existence of a tree projection is also a necessary condition for
    having this property~{\em \cite{GS84,SS93}}}.

\smallskip

\item[--] Consider classes of bounded-arity queries $Q$, and the tree decomposition method, hence the view set $\it
    v\mbox{-}tw_k(Q)$ with its associated hypergraph $\HG_Q^{tk}$. For any database $\DB$, let ${\it d\mbox{-}tw_k(Q,\DB)}$ be the database obtained
    by associating each view in $\it v\mbox{-}tw_k(Q)$ with the cartesian product of the set of constants that variables occurring in it may
    take. It is known that $\forall \DB$, $(\red(\it v\mbox{-}tw_k(Q),{\it d\mbox{-}tw_k(Q,\DB)})\neq\emptyset) \Rightarrow$ ($Q^\onDB\neq \emptyset)$
    if~\cite{DKV02}, and only if~\cite{ABD07}, there is a tree projection of $\HGUnoSet$ w.r.t.~$\HG_Q^{tk}$, for some {\em core} $Q'$~of~$Q$.
    In fact, the result holds for any query $Q'$ that is homomorphically equivalent to $Q$, denoted by $Q' \homEquiv Q$ (instead of just for a core,
    which is any smallest one). This result provides a necessary and sufficient condition for query answering via local consistency,
    without computing any tree-decomposition of such a subquery $Q'$, which would be an $\NP$-hard task~\cite{DKV02}. Observe that \emph{the necessary
    condition holds only for structures of bounded arity, and the result provides only information about the decision problem (i.e., checking whether the answer is empty or not)}.

\smallskip
\item[--] For the general case of queries $Q$ with unbounded arity, consider the generalized hypertree decomposition method and hence the
    view set $\it v\mbox{-}hw_k(Q)$, containing one distinct view over each set of variables that can be covered by at most $k$ query-atoms, and its associated hypergraph $\HG_Q^{k}$.
    Moreover, for any database $\DB$, let ${\it d\mbox{-}hw_k(Q,\DB)}$ be the database obtained by associating each view in $\it v\mbox{-}tw_k(Q)$ with the (natural) join of all
    query-views over which it is defined. It is known that $\forall \DB$, $(\red(\it v\mbox{-}hw_k(Q),{\it d\mbox{-}hw_k(Q,\DB)})\neq\emptyset)\Rightarrow (Q^\onDB\neq
    \emptyset)$ if there exists a tree projection of $\HGUnoSet$ w.r.t.~$\HG_{Q'}^k$, where $Q'$ is any query such that $Q' \homEquiv
    Q$~\cite{CD05}. Note that, when we focus on generalized hypertree decompositions, instead of looking at views in $\it v\mbox{-}hw_k(Q)$ and tree
    projections, we may directly look at the consistency between every pair of sets of $k$ atoms, also called \emph{$k$-local consistency}.
    Hence, the result states a sufficient condition for deciding whether the answer is empty or not by enforcing $k$-local consistency,
    (again) without actually identifying such a subquery $Q'$ and without computing a generalized hypertree decomposition of $Q'$, which
    are both $\NP$-hard tasks. It was open whether the condition is also necessary~\cite{CD05}. Moreover, \emph{as in the above point about tree
    decompositions, the relationship  with global consistency and hence with the related problem of computing solutions was missing}.
\end{myitemize}

\noindent From these results, it emerges that the precise power of local-consistency based computations and of their relationships with tree
projections and with the other structural decomposition methods (in particular, tree decompositions and generalized hypertree decompositions)
was far from being clear: Is it possible that there are queries where such local computations do work even if no decomposition (or tree
projection) exists?

For instance, from the above recent results based on homomorphically equivalent subqueries for tree decompositions and generalized hypertree
decompositions, one may deduce that the mentioned conjecture in~\cite{GS84,SS93} (i.e., that local consistency implies global consistency if,
and only if, a tree projection of the query hypergraph exists) may not hold, in general. This is because in the case of queries with multiple
occurrences of the same relation symbol, the concept of {core} of the query plays a crucial role~\cite{DKV02}, as it should be clear from the
next example.

\begin{example}\label{ex:NeedOfCore}\em
Consider the following queries:
\[
\begin{array}{ll}
Q_1:\ r(A,B)\wedge r(B,C) \wedge r(C,D)\wedge r(D,A)\\
Q_2:\ r(A,B)\wedge r(B,C) \wedge r(D,C)\wedge r(A,D)\\
Q_3:\ r(B,A)\wedge r(C,B) \wedge r(C,D)\wedge r(D,A)\\
\end{array}
\]
These queries are completely equivalent as far as their hypergraphs are concerned, since $\HG_{Q_1}=\HG_{Q_2}=\HG_{Q_3}$. However, $Q_1$ is
already a core, while a core of $Q_2$ (resp., $Q_3$) is the acyclic sub-query $r(A,B)\wedge r(B,C)$ (resp., $r(C,D)\wedge r(D,A)$).
Thus, by focusing on $Q_2$ and $Q_3$ rather than on their cores, we could overestimate their intricacy.~\hfill~$\lhd$
\end{example}

However, the above conjecture might still hold in the original setting considered in~\cite{GS84}, where all relation symbols in a query are
distinct.

\medskip

\noindent \textbf{(Q2) Are there unexplored islands of tractability based on tree projections?}
An \emph{island of tractability} 
in the tree projection framework is a class $\C$ of pairs $(Q,\V)$ that can be efficiently recognized, and such that $Q$ can be efficiently
evaluated on every database, by possibly exploiting the views that are available in~$\V$.

Many specializations of tree projections, such as \emph{tree decompositions}~\cite{RS84}, \emph{hypertree decompositions}~\cite{gott-etal-99},
\emph{component decompositions}~\cite{GMS07}, and \emph{spread-cuts decompositions}~\cite{CJG08}, define islands of tractability whenever some
fixed bound is imposed on their widths. This is also the case for \emph{fractional hypertree decompositions}~\cite{GM06}, whenever the
resources sufficient for computing their $O(w^3)$ approximation~\cite{M09} are used as available  views.
However, this is not the case for general tree projections. Indeed, while Goodman and Shmueli~\cite{GS84} observed that queries that admit a
tree projection can be evaluated in polynomial time, Gottlob et al.~\cite{GMS07} proved that checking whether a tree projection exists or not
is an $\NP$-hard problem. Hence, the class $\C_{tp}=\{(Q,\V) \mid \HG_Q \mbox{ has a tree projection w.r.t.~} \HG_\V  \}$, which includes all
the above mentioned islands of tractability, is not an island of tractability in its turn.
In fact, in addition to the above result, we also know that:
\begin{myitemize}
\item[--] Deciding whether a tree projection of $\HG_Q$ w.r.t.~$\HG_Q^{tk}$\vspace{-1mm} (corresponding to a {tree decomposition}) exist is
    feasible in
    time $O({\small 2^{ck^2}}\times n)$, where $n$ is the size of $\HG_Q$, $k$ is the treewidth, and $c$ is a
    constant~\cite{bodl-96}, hence in linear time for a fixed width~$k$.

\smallskip
\item[--] The problem remains $\NP$-hard for the case of generalized hypertree decompositions, that is, when we have to decide the existence of
    a tree projection of $\HG_Q$ w.r.t.~$\HG_Q^k$, even if $k$ is a fixed number (greater than 2)~\cite{GMS07}.
\end{myitemize}

Moreover, recall that the sufficient conditions we have discussed in the previous point \textbf{(Q1)} do not identify (accessible) islands of
tractability, because their recognition problems are $\NP$-hard, too.
Such conditions are particularly useful in those settings where it is intractable to compute any tree projection, so that answers are computed
via procedures enforcing local consistency. However, having a tree projection at hands allows queries to be evaluated more efficiently
w.r.t.~techniques based on ``blind'' local-consistency enforcing. Intuitively, by having such a projection $\HG_a$ and hence a join tree for
$\HG_a$, we are able to exploit all the well known algorithms developed for acyclic queries. In particular, in this approach, only the views
occurring in the join tree are involved in the query evaluation,  while all available views should be used if no tree projection is available.
Furthermore, the number of semijoin operations to be performed having the join tree is at most the number of nodes in such a tree and does not
depend on the database, as it happens instead while enforcing local consistency.
Therefore, a natural question is whether there is any subclass of $\C_{tp}$, at least including all the tractable classes mentioned above,
which identifies an actual island of tractability where tree projections can be computed efficiently.

\subsection{Contribution}

In this paper, we provide a clear picture of the power of tree projections and structural decomposition methods, by answering the two questions
illustrated above.

It is worthwhile noting that our answers, summarized below, find applications in all those problems that can be solved efficiently on acyclic
and quasi-acyclic instances, even outside the Database area. In particular, our results can be exploited immediately for solving Constraint
Satisfaction Problems (CSPs) where constraints are represented as finite relations encoding allowed tuples of values (see,
e.g.,~\cite{gott-etal-00}).

\medskip

\noindent \textbf{(Q1)}
The first achievement of this paper is to solve the long-standing question about the power of local-consistency based computations, by
addressing in the analysis both the decision problem of checking whether the query is not empty, and the problem of characterizing a necessary
and sufficient condition guaranteeing that local consistency entails global consistency, which is useful from the query answering perspective.

Concerning the decision problem of checking whether the query has a solution, we show that the sufficient conditions identified for some
specializations of tree decompositions are also necessary, even in the most general framework. However, the technical machinery needed for
obtaining our results is quite different from the one used in~\cite{ABD07} for tree decompositions, which does not work when we have arbitrary
signatures or arbitrary views. Our first contribution is to show that:

\begin{center} \hspace{-5mm}{ \centering \fbox{\parbox{0.88\textwidth}{ \em \vspace{1mm}
The following are equivalent:

\vspace{-1mm}
\begin{enumerate}
\item[(1)] For every database \emph{$\DB$}, \emph{$\lc(\W,\DB)$} entails \emph{$Q^\onDB \neq \emptyset$}.

\item[(2)] There is a subquery $Q' \homEquiv Q$ for which $(\HGUnoSet,\HGDue)$ has a tree projection.
\end{enumerate}\vspace{-1mm}}} }
\end{center}

\smallskip

Our second contribution is then to single out the (stronger) conditions under which local consistency entails global consistency. We show that
finding a necessary and sufficient condition requires to exploit possible endomorphisms of the query. It emerged that to characterize when, at
local consistency, an atom  $p$ contains \emph{all, and only,} the correct tuples of the query $Q$ projected over the variables
$\vars(p)=\{X_1,...,X_n\}$ of $p$, we must look for tree projections of some ``output-aware'' substructures of $Q$. We say that
$\{X_1,...,X_n\}$ is $\tpCovered$ in $Q$ (w.r.t.~$\V$) if there is a tree projection of $(\HG_{Q_p},\HGDue)$, where $Q_p$ is a core of the
novel query $Q\wedge r(X_1,...,X_n)$, in which $r$ is a fresh relation symbol. Intuitively, $r$ is used to force any such a core to contain the
desired variables $\{X_1,...,X_n\}$. It turns out that, for having global consistency guaranteed by local consistency,
 for each query atom $p$, a tree projection $\HG_p$ of such a $Q_p$ must exist.

\begin{center} \hspace{-5mm}{ \centering \fbox{\parbox{0.88\textwidth}{ \em \vspace{1mm}The following are equivalent:

\vspace{-1mm}
\begin{enumerate} \item[(1)] For every database \emph{$\DB$}, \emph{$\lc(\W,\DB)$} entails \emph{$\gc(\W,\DB,Q)$}.

\item[(2)] For each query atom $q$, $\vars(q)$ is $\tpCovered$ in $Q$.
\end{enumerate}\vspace{-1mm}}} }
\end{center}

Thus, if \emph{(2)} holds and one is interested in computing query answers over output variables included in some query atom, then all
solutions are immediately available. In fact, the above result comes in the paper as a specialization of a more general result dealing with
those cases where one is interested in computing answers over an arbitrary subset of variables covered by some available view.

Moreover, observe that in the above condition different tree projections for different query atoms are allowed. That is, global consistency can
hold even if there is no tree projection that is able to cover all query atoms at once. However, if every relation symbol is used at most once
in the query, it is easy to see that \emph{(2)} is equivalent to requiring that a tree projection of the whole query exists. Hence, the
conjecture of~\cite{GS84} about the necessity of having a tree projection of the query does not hold in general, but it does hold for such a
restricted setting (in fact, the one considered in \cite{GS84}).

Actually, in this informal statement we have implicitly assumed databases where views are not more restrictive than the query; otherwise, using
such views may clearly lead to missing some tuple in the query answer. Note that this condition trivially holds whenever views are computed
from parts of the query (i.e., they are in fact subqueries), which happens in structural decomposition methods. However, this is not
necessarily true if one would like to exploit existing materialized views. Anyway, we show that soundness of query answers is always
guaranteed. If views are too restrictive w.r.t.~$Q$, then we may just miss completeness.

\medskip

\noindent \textbf{(Q1: Application to Decomposition Methods)}
As a direct consequence of our contribution w.r.t.~question \textbf{(Q1)}, we get in a unique result the generalization of all tractability
results known for purely structural decompositions methods (because all of them are specializations of the notion of tree projections).
Moreover, we provide the precise characterization of the power of \emph{$k$-local consistency}  for classes of queries without a fixed bound on
the arity, which was missing in~\cite{ABD07} and~\cite{CD05}.

In particular, we provide a necessary and sufficient condition such that $k$-local consistency entails global consistency, which is useful for
computing solutions.
Furthermore, concerning the decision problem (query non-emptiness), we show that the sufficient condition identified in~\cite{CD05} is in fact
necessary, too:

\begin{center} \hspace{-5mm}{ \centering \fbox{\parbox{0.88\textwidth}{ \em \vspace{1mm}
The following are equivalent:

\vspace{-1mm}\begin{enumerate} \item[(1)] For every database \emph{$\DB$}, \emph{$\red({\it v\mbox{-}hw_k(Q)},{\it d\mbox{-}hw_k(Q,\DB)})\neq
\emptyset$} entails $Q^\onDB\neq\emptyset$.

\item[(2)] $Q$ has a core having generalized hypertree width at most $k$.
\end{enumerate}\vspace{-1mm}}} }
\end{center}

We point out that the result is not an immediate corollary of the previous one about tree projections (by setting $\HGDue=\HG_Q^k$, where
$\HG_Q^k$ is the hypergraph where each hyperedge is the set of variables occurring in some group of at most $k$ query-atoms). Indeed, let $Q'$
be any core of $Q$, and recall that $Q'$ may be much smaller than $Q$. Thus, the set of views that can be used to form a $k$-width generalized
hypertree decomposition of $Q'$ only come from groups of at most $k$ atoms occurring in $Q'$. It follows that this set can be much smaller than
$\W_k$, which is built from the full query $Q$.
For another difference between our general result and the above one, note that the database ${\it d\mbox{-}hw_k(Q,\DB)}$ for the available
views, over which local consistency is considered, is functionally determined by the relations of query atoms in $\DB$ (instead of being almost
arbitrary).

Note that, for $k=1$, local consistency is required to hold only on the query views playing the role of the original query atoms. We thus
obtain {\em the precise characterization of the power of local consistency in acyclic queries}, generalizing the classical result
in~\cite{BFMY83} given for queries without multiple occurrences of the same relation symbol: for every database \emph{$\DB$}, {\em local
consistency (of query views) entails  $Q^\onDB\neq\emptyset$ if, and only if, $Q$ has an acyclic core}.

\medskip

\noindent \textbf{(Q2)} As discussed above, the classes of instances where enforcing local consistency is a correct query-answering procedure
are not efficiently recognizable. Therefore, it is natural to look for subclasses that are efficiently recognizable and that are strictly
larger than the islands of tractability known so far.
Addressing this issue is the second main achievement of the paper.
To this end, we exploit the game-theoretic characterization of tree projections in terms of the \emph{Captain and Robber} game~\cite{GS08}. The
game is played on a pair of hypergraphs $(\HG_1,\HG_2)$ by a Captain controlling, at each move, a squads of cops encoded as the nodes in a
hyperedge $h\in \edges(\HG_2)$, and by a Robber who stands on a node and can run at great speed along the edges of $\HG_1$, while being not
permitted to run trough a node that is controlled by a cop. In particular, the Captain may ask any cops in the squad $h$ to run in action, as
long as they occupy nodes that are currently reachable by the Robber, thereby blocking an escape path for the Robber. While cops move, the
Robber may run trough those positions that are left by cops or not yet occupied. The goal of the Captain is to place a cop on the node occupied
by the Robber, while the Robber tries to avoid her capture. The Captain has a winning strategy if, and only if, there is a tree projection of
$\HG_1$ w.r.t.~$\HG_2$. Then,

\begin{itemize}
\item[$\blacktriangleright$] We define the notion of {\em greedy strategies}, which are winning strategies for the Captain, possibly
    non-monotone, where it is required that \emph{all} cops available at the current squad $h$ and reachable by the Robber enter in action.
    If all of them are in action, then a new squad $h'$ is selected, again requiring that all the active cops, i.e., those in the frontier,
    enter in action. In the Captain and Robber game, it is known that there is no incentive for the Captain to play a strategy that is not
    monotone~\cite{GS08}. Instead, by focusing on greedy strategies, we can exhibit examples where there exists non-monotone winning
    strategies but no monotone winning one.

\item[$\blacktriangleright$] We show that {\em greedy strategies} can be computed in polynomial time, and that based on them (even on
    non-monotone ones) it is possible to construct, again in polynomial time, tree projections, which are called \emph{greedy}. Therefore,
    the class $\C_{gtp}\subset \C_{tp}$ of all {greedy tree projections} turns out to be an island of tractability.

\item[$\blacktriangleright$] Finally, we show that $\C_{gtp}$ properly includes most previously known islands of tractability (based on
    structural properties), precisely because of the power of non-monotonic strategies. In particular, the novel notion of greedy tree
    projections allows us to define new islands of tractability from any known structural decomposition method, such as the {\em greedy
    (generalized) hypertree decomposition} or the {\em greedy component decomposition}, which are tractable and strictly more powerful than
    their original versions.
\end{itemize}

\subsection{Organization}

The paper is organized as follows.
Section~\ref{sec:framework} illustrates some basic notions and concepts. The characterization of the power of local consistency is given in
Section~\ref{CP}, while its application to structural decomposition methods is reported in Section~\ref{SDM}. Islands of tractability for tree
projections are singled out in Section~\ref{SECGAMES}, and an application of the results to structures having ``small'' arities is presented in
Section~\ref{smallArities}. A few further remarks and open issues are discussed in Section~\ref{sec:conclusion}.

\section{Preliminaries}\label{sec:framework}

\noindent \textbf{Hypergraphs and Acyclicity.} A \emph{hypergraph} $\HG$ is a pair $(V,H)$, where $V$ is a finite set of nodes and $H$ is a set
of hyperedges such that, for each $h\in H$, $h\subseteq V$.
If $|h|=2$ for each (hyper)edge $h\in H$, then $\HG$ is a {\em graph}.
For the sake of simplicity, we always denote $V$ and $H$ by $\nodes(\HG)$ and $\edges(\HG)$, respectively.

A hypergraph $\HG$ is {\em acyclic} (more precisely, $\alpha$-acyclic~\cite{fagi-83}) if, and only if, it has a join tree~\cite{bern-good-81}.
A {\em join tree} $\JT$ for a hypergraph $\HG$ is a tree whose vertices are the hyperedges of $\HG$ such that, whenever a node $X\in V$ occurs
in two hyperedges $h_1$ and $h_2$ of $\HG$, then $h_1$ and $h_2$ are connected in $\JT$, and $X$ occurs in each vertex on the unique path
linking $h_1$ and $h_2$. In words, the set of vertices in which $X$ occurs induces a (connected) subtree of $\JT$. We will refer to this
condition as the {\em connectedness condition} of join trees.

\medskip \noindent \textbf{Tree Decompositions.}
A \emph{tree decomposition}~\cite{RS84} of a graph $G$ is a pair $\tuple{T,\chi}$, where $T=(N,E)$ is a tree, and $\chi$ is a labeling function
assigning to each vertex $v\in N$ a set of vertices $\chi(v)\subseteq \nodes(G)$, such that the following conditions are satisfied: (1) for
each node $Y\in \nodes(G)$, there exists $p\in N$ such that $Y\in \chi(p)$; (2) for each edge  $\{X,Y\}\in \edges(G)$, there exists $p\in N$
such that $\{X,Y\}\subseteq \chi(p)$; and (3) for each node $Y\in \nodes(G)$, the set $\{p\in N \mid  Y\in \chi(p)\}$ induces a (connected)
subtree of $T$. The \emph{width} of $\tuple{T,\chi}$ is the number $\max_{p\in N}(|\chi(p)|-1)$.

The \emph{Gaifman graph} of a hypergraph $\HG$ is defined over the set $\nodes(\HG)$ of the nodes of $\HG$, and contains an edge $\{X,Y\}$ if,
and only if, $\{X,Y\}\subseteq h$ holds, for some hyperedge $h\in\edges(\HG)$. The {\em treewidth} of $\HG$ is the minimum width over all the
tree decompositions of its Gaifman graph. Deciding whether a given hypergraph has treewidth bounded by a fixed natural number $k$ is known to
be feasible in linear time~\cite{bodl-96}.

\medskip \noindent \textbf{(Generalized) Hypertree Decompositions.}
A {\em hypertree for a hypergraph $\HG$} is a triple $\tuple{T,\chi,\lambda}$, where $T=(N,E)$ is a rooted tree, and $\chi$ and $\lambda$ are
labeling functions which associate each vertex $p\in N$ with two sets $\chi(p)\subseteq \nodes(\HG)$ and $\lambda(p)\subseteq \edges(\HG)$. If
$T'=(N',E')$ is a subtree of $T$, we define $\chi(T')= \bigcup_{v\in N'} \chi(v)$.
In the following, for any rooted tree $T$, we denote the set of vertices $N$ of $T$ by $\vertices(T)$, and the root of $T$ by $\root(T)$.
Moreover, for any $p\in N$, $T_p$ denotes the subtree of $T$ rooted at $p$.

A {\em generalized hypertree decomposition}~\cite{gott-etal-03} of a hypergraph $\HG$ is a hypertree $\HD=\tuple{T,\chi,\lambda}$ for $\HG$
such that: (1) for each hyperedge $h\in \edges(\HG)$, there exists $p\in \vertices(T)$ such that $h\subseteq \chi(p)$; (2) for each node $Y\in
\nodes(\HG)$, the set $\{ p\in \vertices(T) \mid  Y \in \chi(p) \}$ induces a (connected) subtree of $T$; and (3) for each $p\in \vertices(T)$,
$\chi(p)\subseteq \nodes(\lambda(p))$.
The {\em width} of a generalized hypertree decomposition $\tuple{T,\chi,\lambda}$ is $max_{p\in \vertices(T)} |\lambda(p)|$. The {\em
generalized hypertree width} $ghw(\HG)$ of $\HG$ is the minimum width over all its
generalized hypertree decompositions. 

A \emph{hypertree decomposition}~\cite{gott-etal-99} of $\HG$ is a generalized hypertree decomposition $\HD=\tuple{T,\chi,\lambda}$ where: (4)
for each $p\in \vertices(T)$, $\nodes(\lambda(p)) \cap \chi(T_p) \;\subseteq\; \chi(p)$. Note that the inclusion in the above condition is
actually an equality, because Condition~(3) implies the reverse inclusion. The {\em hypertree width} $hw(\HG)$ of $\HG$ is the minimum width
over all its hypertree decompositions.
Note that, for any hypergraph $\HG$, it is the case that $ghw(\HG)\leq hw(\HG)\leq 3\times ghw(\HG)+1$~\cite{AGG07}. Moreover, for any fixed
natural number $k>0$, deciding whether $hw(\HG)\leq k$ is feasible in polynomial time (and, actually, it is
highly-parallelizable)~\cite{gott-etal-99}, while deciding whether $ghw(\HG)\leq k$ is $\NP$-complete~\cite{GMS07}.

\medskip \noindent\textbf{Tree Projections.}
For two hypergraphs $\HG_1$ and $\HG_2$, we write $\HG_1\leq \HG_2$ if, and only if, each hyperedge of $\HG_1$ is contained in at least one
hyperedge of $\HG_2$. Let $\HG_1\leq \HG_2$; then, a \emph{tree projection} of $\HG_1$ with respect to $\HG_2$ is an acyclic hypergraph $\HG_a$
such that $\HG_1\leq \HG_a \leq \HG_2$. Whenever such a hypergraph $\HG_a$ exists, we say that the pair of hypergraphs $(\HG_1,\HG_2)$ has a
tree projection.

Note that the notion of tree projection is more general than the above mentioned (hyper)graph based notions. For instance, consider the
generalized hypertree decomposition approach. Given a hypergraph $\HG$ and a natural number $k>0$, let $\HG^k$ denote the hypergraph over the
same set of nodes as $\HG$, and whose set of hyperedges is given by all possible unions of $k$ edges in $\HG$, i.e., $\edges(\HG^k)= \{ h_1
\cup h_2 \cup \cdots \cup h_k \mid \{h_1,h_2,\ldots,h_k\}\subseteq \edges(\HG)\}$. Then, it is well known and easy to see that $\HG$ has
generalized hypertree width at most $k$ if, and only if, there is a tree projection for $(\HG,\HG^k)$.

Similarly, for tree decompositions, let $\HG^{tk}$ be the hypergraph over the same set of nodes as $\HG$, and whose set of hyperedges is given
by all possible clusters $B\subseteq\nodes(\HG)$ of nodes such that $|B| \leq k+1$. Then, $\HG$ has treewidth at most $k$ if, and only if,
there is a tree projection for $(\HG,\HG^{tk})$.

\medskip \noindent\textbf{Relational Structures and Homomorphisms.}
Let $\U$ and $\mathcal{X}$ be disjoint infinite sets that we call the {\em universe of constants} and the \emph{universe of variables},
respectively. A (relational) vocabulary $\tau$ is a finite set of relation symbols of specified (finite) arities. A {\em relational structure}
$\A$ over $\tau$ (short: $\tau$-structure) consists of a universe $A\subseteq \U\cup \mathcal{X}$ and, for each relation symbol $r$ in $\tau$,
of a relation $r^\A\subseteq A^\rho$, where $\rho$ is the arity of $r$.

Let $\A$ and $\B$ be two $\tau$-structures with universes $A$ and $B$, respectively.
A {\em homomorphism} from $\A$ to $\B$ is a mapping $h: A \mapsto B$ such that $h(c)=c$ for each constant $c$ in $A\cap \U$, and such that, for
each relation symbol $r$ in $\tau$ and for each tuple $\tuple{a_1,\ldots,a_\rho}\in r^\A$, it holds that $\tuple{h(a_1),\ldots,h(a_\rho)}\in
r^\B$. For any mapping $h$ (not necessarily a homomorphism), $h(\tuple{a_1,\ldots,a_\rho})$ is used, as usual, as a shorthand for
$\tuple{h(a_1),\ldots,h(a_\rho)}$.

A $\tau$-structure $\A$ is a substructure of a $\tau$-structure $\B$ if $A \subseteq B$ and $r^\A \subseteq r^\B$, for each relation symbol $r$
in $\tau$.

\medskip \noindent\textbf{Relational Databases.}
Let $\tau$ be a given vocabulary. A {\em database instance} (or, simply, a database) $\DB$ over $D\subseteq \U$ is a $\tau$-structure $\DB$
whose universe is the set $D$ of constants. For each relation symbol $r$ in $\tau$, $r^\onDB$ is a \emph{relation instance} (or, simply,
relation) of $\DB$. Sometimes, we adopt the logical representation of a database~\cite{ullm-89,abit-etal-95}, where a tuple
$\tuple{a_1,...,a_\rho}$ of values from $D$ belonging to the $\rho$-ary relation (over symbol) $r$ is identified with the {\em ground atom}
$r(a_1,...,a_\rho)$. Accordingly, a database $\DB$ can be viewed as a set of ground atoms. Unless otherwise stated, we implicitly assume that
databases are finite.

\medskip\noindent\textbf{Conjunctive Queries.} A {\em conjunctive query} $Q$ consists of a finite conjunction of atoms of the form $r_1({\bf
u_1})\wedge\cdots\wedge r_m({\bf u_m})$, where $r_1,...,r_m$ (with $m > 0$) are relation symbols (not necessarily distinct), and ${\bf
u_1},...,{\bf u_m}$ are lists of terms (i.e., variables or constants). The set of all atoms occurring in $Q$ is denoted by $\atoms(Q)$. For a
set of atoms $A$, $\vars(A)$ is the set of variables occurring in the atoms in $A$. For short, $\vars(Q)$ denotes $\vars(\atoms(Q))$. We say
that $Q$ is a {\em simple query} if every atom is over a distinct relation symbol.
Given a database $\DB$ over $D$, $Q^\onDB$ denotes the set of all \emph{answers} of $Q$ on $\DB$, that is, all substitutions
$\theta:\vars(Q) \mapsto D$ such that for each $1\leq i\leq m$, $\theta'(r_{\alpha_i}({\bf u_i}))\in \DB$, where $\theta'(t)=\theta(t)$ if
$t\in \vars(Q)$ and $\theta'(t)=t$ otherwise (i.e., if the term $t$ is a constant).

Note that any conjunctive query $Q$ can be viewed as a relational structure $\mathcal{Q}$, whose vocabulary $\tau_Q$ and universe $U_Q$ are the
set of relation symbols and the set of terms occurring in its atoms, respectively. For each symbol $r_i\in\tau_Q$,  the relation
$r_i^{\mathcal{Q}}$ contains a tuple of terms ${\bf u}$, for any atom of the form $r_i({\bf u})\in\atoms(Q)$ defined over $r_i$. In the special
case of simple queries, every relation $r_i^{\mathcal{Q}}$ of $\mathcal{Q}$ contains just one tuple of terms.
According to this view, elements in $Q^\onDB$ are in a one-to-one correspondence with homomorphisms from $\mathcal{Q}$ to $\DB_Q$, where the
latter is the (maximal) substructure of $\DB$ over the (sub)vocabulary $\tau_Q$. Hereafter, we adopt this view but, for the sake of
presentation, we identify queries and databases with their relational structures, i.e., we use directly $Q$ and $\DB$ in place of $\mathcal{Q}$
and $\DB_Q$.

For any given set $S$ of variables, we denote by $Q^\onDB[S]$ the restriction of the (substitutions/)homomorphisms in $Q^\onDB$ over the
variables in $S$. For the extreme case where $S=\emptyset$, define $h_{\it true}$ to be the restriction of any homomorphism over the empty set.
Then, $Q^\onDB[\emptyset]=\{h_{\it true}\}$ if $Q^\onDB\neq \emptyset$, and $Q^\onDB[\emptyset]=\emptyset$ if $Q^\onDB= \emptyset$.
If $a$ is an atom, then $Q^\onDB[a]$ denotes $Q^\onDB[\vars(a)]$.

Note that any atom $a$ can be viewed as a one-atom query, so that $a^\onDB$ is the set of all the homomorphisms from $a$ to $\DB$, restricted
to $\vars(a)$ (i.e., projecting out possible constants occurring in $a$). For a set $A$ of atoms, we denote by $A^\onDB$ the set $\{ a^\onDB
\mid a\in A \}$.

A \emph{core} of $Q$ is a query $Q'$ such that: \emph{(1)} $\atoms(Q')\subseteq \atoms(Q)$; \emph{(2)} there is a homomorphism from $Q$ to
$Q'$; and \emph{(3)} there is no query $Q''$ satisfying \emph{(1)} and \emph{(2)} such that $\atoms(Q'')\subset \atoms(Q')$. Equivalently, in
terms of relational structures, $Q'$ is a minimal substructure of $Q$ such that (2) holds. The set of all the cores of $Q$ is denoted by
$\cores(Q)$. Elements in $\cores(Q)$ are \emph{isomorphic}.

\medskip\noindent \textbf{Hypergraphs and atoms.} There is a very natural way to associate a hypergraph $\HG_\V=(N,H)$ with any set $\V$ of
atoms: the set $N$ of nodes consists of all variables occurring in $\V$; for each atom in $\V$, the set $H$ of hyperedges contains a hyperedge
including all its variables; and no other hyperedge is in $H$.

For a query $Q$, the hypergraph associated with $\atoms(Q)$ is briefly denoted by $\HG_Q$. If $\HG_Q$ is a connected hypergraph, we say that
$Q$ is a {\em connected} query.

\section{The Power of Local Consistency}\label{CP}

Throughout the paper, we assume that $Q$ is a conjunctive query and that $\V$ is a non-empty set of atoms, which we call \emph{views}, such
that $\vars(\V) = \vars(Q)$. Moreover, $\DB$ is a database over the vocabulary $\DS$ containing the relation symbols of query atoms and views.
We require w.l.o.g. that every available view is over a specific relation symbol, which does not occur in the given query, and that the list of
terms of every view does not contain any constant or repeated variables (in fact, observe that from any given set of available views, one may
immediately get a new set of views where these assumptions hold).
Note that, within this setting, each view $w\in \V$ is univocally associated with a relation instance in $\DB$, whose tuples are in a
one-to-one correspondence with the homomorphisms in $w^\onDB$. Therefore, this relation instance will be simply denoted by $w^\onDB$, and we
freely use the term tuples interchangeably with homomorphisms, when we refer to its elements.

Our first goal is to characterize the relationships between tree projections and certain consistency properties that hold for $Q$ and $\W$ over
some (or all) given databases. To this end, we need to state some preliminary notions and definitions, which will be illustrated by referring
to the following running example.

\begin{figure}[t]
  \centering
  \includegraphics[width=0.8\textwidth]{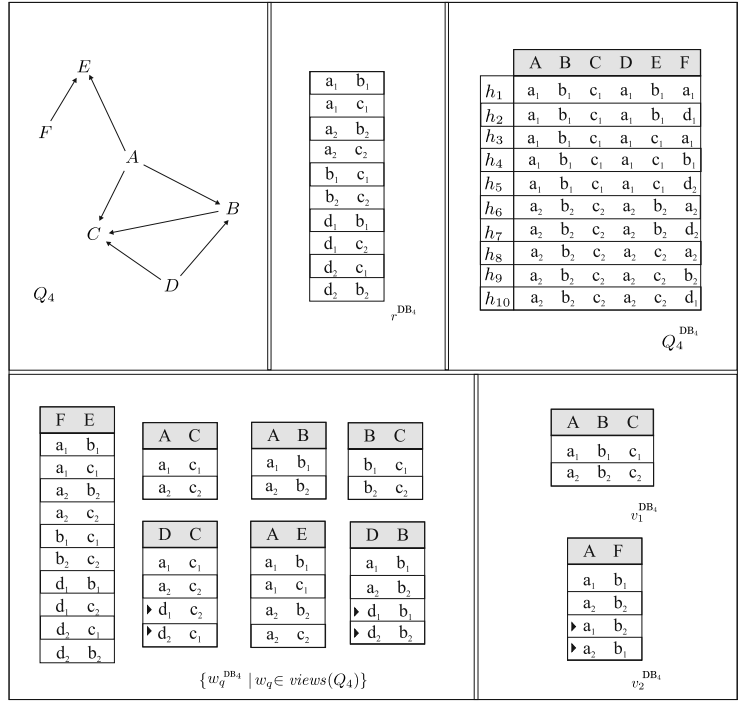}
  \caption{The (hypergraph of the) query $Q_4$, the tuples in the database $\DB_4$, and the answers in $Q_4^{\onDB_4}$, in Example~\ref{ex:MultipleCores}.
 }\label{fig:triangle}
\end{figure}

\begin{example}\label{ex:MultipleCores}\em
Consider the following query $Q_4$, where all atoms are over the same binary relation symbol $r$:
\[
\begin{array}{ll}
Q_4: &
r(A,B)\wedge r(B,C)\wedge r(A,C)\wedge r(D,C)\wedge r(D,B)\wedge r(A,E)\wedge r(F,E).\\
\end{array}
\]
A graphical representation of this query is reported in Figure~\ref{fig:triangle}, where edge orientation just reflects the position of the
variables in query atoms. Moreover, consider the database $\DB_4$ shown in Figure~\ref{fig:triangle}, by focusing on the relation instance
$r^{\onDB_4}$.
Then, it can be checked that the answers of $Q_4$ on $\DB_4$ are the homomorphisms $h_1,...,h_{10}$, which are also reported, in tabular form,
in Figure~\ref{fig:triangle}.

In this example, in order to  answers $Q_4$, we assume the availability of the set of views $\V_4=\{v_1(A,B,C),$ $v_2(A,F),$ $v_3(A,B),$
$v_4(A,C),$ $v_5(A,E),$ $v_6(B,C),$ $v_7(D,B),$ $v_8(D,C),$ $v_9(F,E)\}$, and that the database $\DB_4$ includes a relation instance
$w^{\onDB_4}$, for each view $w\in\V_4$. Note that, in the figure, such relation instances are identified by the list of variables on which the
views are defined. \hfill $\lhd$
\end{example}

\subsection{Consistency Properties and Views}

\medskip\noindent \textbf{View Consistency.} For a view $w\in\W$, we say that $w^\onDB$ is \emph{view consistent} w.r.t.~$Q$ if
$w^\onDB\supseteq Q^\onDB[w]$.
For the set of views $\W$, we say that $\W^\onDB$ is \emph{view consistent} w.r.t.~$Q$, if the property holds for each $w\in \W$. That is,
views are not more restrictive than the query.

Note that view consistency holds in general for all views initialized from subsets of query atoms, such as those employed in all known
decomposition methods, such as (hyper)tree decompositions.
However, we are also interested in a wider framework where views are completely arbitrary and may be available from previous computations,
possibly unrelated with the present query $Q$. Accordingly, we do not require that view consistency holds for such views, and we shall look for
general results, which will be then smoothly inherited by more specific settings.

\begin{example}\label{ex:vc}\em
Consider again the setting of Example~\ref{ex:MultipleCores}, and in particular the views $v_1(A,B,C)$ and $v_2(A,B,C)$. Note that
$v_1(A,B,C)^{\onDB_4}$ is a set of two homomorphisms, which are precisely those in the set $Q_4^{\onDB_4}[\{A,B,C\}]$ of the answers of $Q_4$
on $\DB_4$ projected over the variables in $\{A,B,C\}$. Therefore, $v_1(A,B,C)$ is view consistent w.r.t.~$\DB_4$. Similarly, it can be checked
that the views $v_3(A,B),$ $v_4(A,C),$ $v_5(A,E),$ $v_6(B,C),$ $v_7(D,B),$ $v_8(D,C),$ and $v_9(F,E)$ are all view consistent w.r.t.~$\DB_4$.

Instead, $v_2(A,F)$ is not view consistent w.r.t.~$\DB_4$, since $v_2(A,F)^{\onDB_4}\supseteq Q_4^{\onDB_4}[\{A,F\}]$ does not hold. For
instance, $v_2(A,F)^{\onDB_4}$ does not include the homomorphism mapping both $A$ and $F$ to the constant $a_1$. Hence, $\V_4^{\onDB_4}$ is not
view consistent w.r.t.~$Q_4$. \hfill $\lhd$
\end{example}

\medskip\noindent \textbf{Local Consistency.}
We say that $\V^\onDB$ is \emph{locally (also, pairwise) consistent}, denoted by $\lc(\W,\DB)$, if $w^\onDB\neq\emptyset$ and $w^\onDB=(w\wedge
w')^\onDB[w]$, for each $\{w,w'\}\subseteq \V$.

From any set of views and any instance $\DB$, we may compute a subset of $\DB$ that is locally consistent. Let the \emph{reduct} of $\DB$
according to $\W$, denoted by $\red(\W,\DB)$, be the (set-inclusion) maximal subset  of $\DB$ such that $\W^\onDBred$ is locally consistent; or
$\red(\W,\DB)=\emptyset$, whenever such a maximal subset does not exist. It is well known that the reduct can be computed as the unique
fixpoint of a procedure consisting of semijoin operations over~$\DB$, which runs in polynomial time.
It is easy to see that such a reducing procedure preserves the given query, unless the used views are more restrictive than the query, of
course. In fact, computing a reduct is often used as a useful heuristic procedure in different areas of computer science, where the
homomorphism problem underlying conjunctive query evaluation comes out---e.g., in {\em constraint satisfaction problems (CSP)}, where such a
procedure is known as {\em generalized arc consistency}~\cite{dech-03}. Indeed, if the reduct is empty, we may safely conclude that there are
no solutions; otherwise, we got anyway a smaller instance of the problem to deal with.

\begin{example}\label{ex:local}\em
In the running example depicted in Figure~\ref{fig:triangle}, the set $\V_4$ of views and the database $\DB_4$ are such that $\V_4^{\onDB_4}$
is locally consistent. Consider for instance the views $v_1(A,B,C)$ and $v_3(A,B)$, and observe that both $(v_1(A,B,C)\wedge
v_3(A,B))^{\onDB_4}[\{A,B,C\}]= v_1(A,B,C)^{\onDB_4}$ and $(v_3(A,B)\wedge v_1(A,B,C))^{\onDB_4}[\{A,B\}]= v_3(A,B)^{\onDB_4}$. Indeed, every
tuple in the relation associated with either view matches with some tuple in the other view on the variables they have in common, so that no
tuple is missed by performing such semijoin operations. This is easily seen because
$v_1(A,B,C)^{\onDB_4}[\{A,B\}]=v_3(A,B)^{\onDB_4}=\{\tuple{a_1,b_1},\tuple{a_2,b_2}\}$ (where these two tuples also identify the homomorphisms
mapping $(A,B)$ to $(a_1,b_1)$ and to $(a_2,b_2)$, respectively). \hfill $\lhd$
\end{example}

\medskip\noindent \textbf{Query Views.}
In the seminal paper about local and global consistency in acyclic queries~\cite{BFMY83}, local consistency is enforced directly on the
relations of query atoms, while we only consider (and possibly enforce) this property on views, in this paper. This is because that paper, as
well as other related papers such as~\cite{GS84}, uses a slightly different formal framework where every relation symbol may occur just once in
a query, i.e., where only simple queries are considered. In contrast with these classical papers, we do not assume anything about the query,
which may contain multiple occurrences of the same relation symbol. This means that the same relation instance may be shared by different query
atoms, and this feature plays a very relevant role, as it was first pointed out in~\cite{DKV02}. In this case, a tuple may be useful for some
atom and useless for another one defined over the same relation symbol. It follows that local consistency  cannot be enforced on the relations
of the query atoms, because such a filtering procedure would lead to undesirable side effects (possibly deleting all tuples in the database,
including the useful ones).

Therefore, we always keep the ``original'' database relations untouched and we rather use suitable views, each one with its own database
relation, to play the role of query atoms in the definition of consistency properties in general queries and in consistency enforcing
procedures. Formally, we say that $\V$ is a {\em view system} (for $Q$) if it contains, for each atom $q\in\atoms(Q)$, a view $w_q$ (over a
distinct relation symbol) with the same set of variables as $q$. These special views in $\V$ are called hereafter  {\em query views}, and are
denoted by $\views(Q)$. If $Q'$ is a subquery of $Q$, $\views(Q')$ denotes the set of query views associated with its atoms.
In the following, the set of available views $\V$ is assumed to be a view system for the given query $Q$, unless otherwise specified.

\begin{example}\em
Consider again the setting of Example~\ref{ex:MultipleCores}, and note that $\V_4$ is in fact a view system for $Q_4$. Indeed, the views in the
set $\{v_3(A,B),$ $v_4(A,C),$ $v_5(A,E),$ $v_6(B,C), v_7(D,B), v_8(D,C), v_9(F,E)\}$ are in a one-to-one correspondence with the query atoms of
$Q_4$. For instance, $v_3(A,C)$ is the query view $w_{r(A,C)}$, with $r(A,C)$ being a query atom of $Q_4$.
Hence, $\views(Q_4)=\{ v_3(A,B),$ $v_4(A,C),$ $v_5(A,E),$ $v_6(B,C),$ $v_7(D,B),$ $v_8(D,C), v_9(F,E)\}$, and $\V_4=\views(Q_4)\cup
\{v_1(A,B,C), v_2(A,F)\}$.\hfill $\lhd$
\end{example}

Observe that working with view systems instead that with arbitrary set of views is not a restrictive assumption, for our purposes. On the
practical side, if some atom misses its associated query view $w_q$ in the available views, one may just add a fresh view $w_q$ to the views,
with a corresponding relation in the database such that  $w_q^\onDB=q^\onDB$.
On the theoretical side, recall that we are dealing with consistency properties of $Q$ and $\V$, and with tree projections of $(\HG_Q,\HG_\V)$.
In fact, such a tree projection exists only if the set of variables of every atom $q$ in $Q$ is covered by some view $w\in\V$, i.e.,
$\vars(q)\subseteq \vars(w)$. Therefore, whenever $\V$ is a set of ``useful views,'' for each query atom $q$ there must exist some view in $\V$
that may play the role of the query view $w_q$ (after projecting it on $\vars(q)$). However, requiring that query views belong to $\V$
simplifies the presentation and allows us to define consistency properties in a clean way. In particular, the role of query views is crucial in
the following definition.

\medskip\noindent \textbf{Global Consistency.}
Informally, this is a highly desirable state of the database where query views contain all and only those tuples that can be returned by query
answers. In this case, an answer of the query can be computed in polynomial time: for each query view $w_q$, select one tuple $h$ in the
relation $w_q^\onDB$ that is univocally associated with $w_q$ in $\DB$, modify this relation so that  $w_q^\onDB=\{h\}$, and propagate this
choice by enforcing again local consistency (see Section~\ref{SDM} for more results and discussions about the problem of computing answers).

Observe that the classical definition, which states the above property for the relations of query atoms, is not useful whenever any relation
symbol $r$ is shared by some query atoms (because we miss the information relating any tuple in $r^\onDB$ with those atoms where the tuple
participates in some answer). By using query views instead of query atoms, no confusion may arise, and we get the desired extension of the
classical definition given (in the literature discussed above) for simple queries.

We say that a database $\DB$ is \emph{globally consistent} with respect to $Q$ and $\W$, denoted by $\gc(\W,\DB,Q)$, if $w_q^\onDB=Q^\onDB[q]$
(which is also equal to $Q^\onDB[w_q]$), for each $q\in \atoms(Q)$, where $w_q$ is the query view associated with $q$.

\begin{example}\em
Let us focus on the query views in $\views(Q_4)$. Consider for instance the view $v_3(A,B)\in\views(Q_4)$ (associated with the query atom
$r(A,B)$), and note that $v_3(A,B)^{\onDB_4}=Q_4^{\onDB_4}[\{A,B\}]$. That is, the answers of $Q_4$ on $\DB_4$ projected over the set $\{A,B\}$
are immediately available by looking at the relation $v_3(A,B)^{\onDB_4}$.

On the other hand, for the view $v_8(D,C)\in \views(Q_4)$, the set $v_8(D,B)^{\onDB_4}$ contains two homomorphisms that do not belong to the
set $Q_4^{\onDB_4}[\{D,C\}]$ (identified by the two tuples marked with the symbol ``{\small $\blacktriangleright$}'' in
Figure~\ref{fig:triangle}). Therefore, $\DB_4$ is not globally consistent w.r.t.~$Q_4$ and $\V_4$. \hfill $\lhd$
\end{example}

\medskip\noindent \textbf{Legal Database.}
While no special requirement is assumed for the database relations of the available views in $\V$, the relations associated with the query
views cannot be arbitrary, otherwise we would lose any connection with the query $Q$ to be solved using the view system $\V$. In fact, these
relations should reflect the intended initialization with the tuples contained in the relations associated with their corresponding query atoms
(possibly filtered by eliminating tuples that are irrelevant w.r.t.~query answers).

We say that $\DB$ is a \emph{legal} database instance (w.r.t.~$Q$ and $\V$)
  if (i)  $w_q^\onDB \subseteq q^\onDB$ holds, for
each query view $w_q\in\views(Q)$;
 and (ii) $\views(Q)^\onDB$ is \emph{view consistent}.
All other view instances may be arbitrary. Then, the following is immediate.

\begin{fact}
For every legal database \emph{$\DB$, $$ Q^\onDB = (\bigwedge_{q\in \atoms(Q)} q)^\onDB = (\bigwedge_{w_q\in \views(Q)} w_q)^\onDB.$$}
\end{fact}

\begin{example}\em
The database $\DB_4$ is legal w.r.t.~$Q_4$ and $\V_4$. Indeed, condition (i) is seen to hold by comparing the relations associated with the
query views with the relation instance $r^{\onDB_4}$. Moreover, in Example~\ref{ex:vc}, we have observed that $\views(Q_4)^{\onDB_4}$ is view
consistent, i.e., condition (ii) holds as well. Then, because of the above fact, the answers of $Q_4$ on $\DB_4$ are also given by the
expression $(\bigwedge_{w_q\in \views(Q_4)} w_q)^{\onDB_4}$.~\hfill~$\lhd$
\end{example}

\begin{remark}
Only legal databases over $Q$ and $\V$ are meaningful for the purpose of this paper. Therefore, unless otherwise stated, we always implicitly
assume hereafter this requirement for any database instance. In particular, whenever we say ``for every database'', we actually mean ``for
every {\em legal} database''. Of course, whenever we define some database instance in proofs of our results, we deal with this requirement, and
we explicitly prove that such a database is actually legal.
\end{remark}

Now that the setting is clarified, our next task is to provide sufficient and necessary conditions to evaluate queries via local consistency.
For the sake of presentation and without loss of generality, we assume that the given query $Q$ is connected and that $\vars(Q)=\vars(\V)$.
Note that, under these assumptions, whenever $\W^\onDB$ is locally consistent, requiring that every relation associated with some view in $\W$
is non-empty is equivalent to requiring that there is at least one $w$ in $\W$ with $w^\onDB\neq\emptyset$. Indeed, the query views in the view
system $\V$ makes $\HG_{\V}$ connected, and thus any empty relation in the database would entail that all relations must be empty, at local
consistency.

\subsection{From Tree Projections to Consistency...}

The fact that local consistency holds for $\V$ and $\DB$ is of course unrelated with the fact that global consistency holds for $\V$ and $\DB$
with respect to $Q$, in general. In this section, we show how the existence of tree projections of some parts of the query is a sufficient
condition to get the implication $\lc(\V,\DB) \Rightarrow \gc(\V,\DB,Q)$. Our analysis will consider arbitrary conjunctive queries, with any
desired set $O$ of output variables, and tree projections w.r.t.~arbitrary view systems.

We start by observing that, when arbitrary view systems are considered, it suddenly emerges that it does not make sense to talk about ``the''
core of a query, because different isomorphic cores may differently behave with respect to the available views. In fact, this phenomenon does
not occur, e.g., for generalized hypertree decompositions (resp., tree decompositions) where all combinations of $k$ atoms (resp., $k+1$
variables) are available as views (see Section~\ref{SDM}).

\begin{figure}[t]
  \centering
  \includegraphics[width=0.9\textwidth]{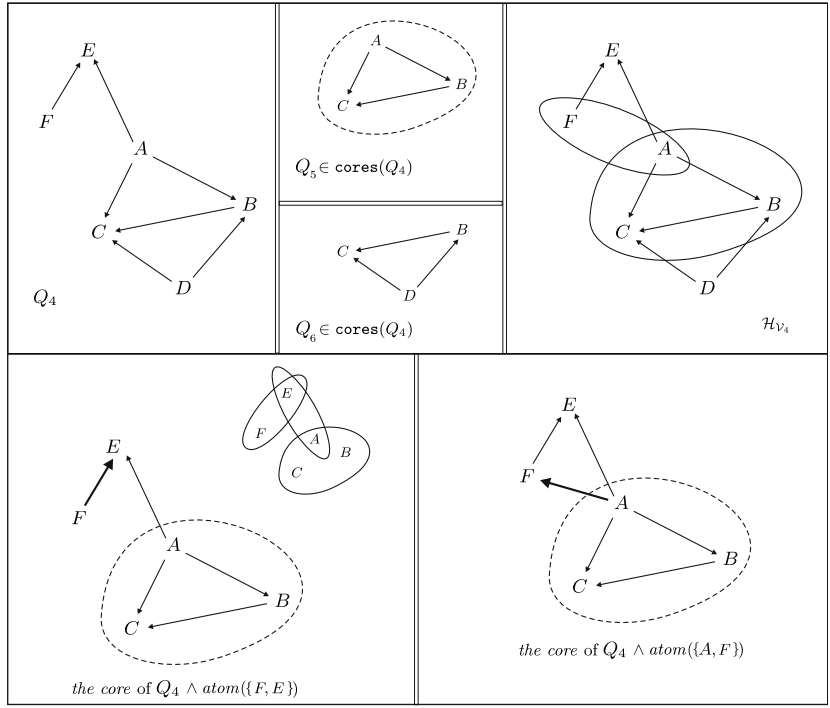}
  \caption{
  The (hypergraph of the) query $Q_4$, the cores $Q_5$ and $Q_6$, the hypergraph $\HG_{\V_4}$, and the cores of the queries $Q_4\wedge \mathit{atom}(\{F,E\})$ (with its tree projection) and $Q_4\wedge \mathit{atom}(\{A,F\})$,
  in Example~\ref{ex:MultipleCores-bis}.
 }\label{fig:triangle-bis}
\end{figure}

\begin{example}\label{ex:MultipleCores-bis}\em
Consider again the query
\[
\begin{array}{ll}
Q_4: & r(A,B)\wedge r(B,C)\wedge r(A,C)\wedge r(D,C)\wedge r(D,B)\wedge r(A,E)\wedge r(F,E),\\
\end{array}
\]
which has been discussed in Example~\ref{ex:MultipleCores}, and which is graphically reported again in Figure~\ref{fig:triangle-bis}, for the
sake of presentation. The figure also reports the hypergraph $\HG_{\V_4}$ associated with the views in $\V_4$ (where, e.g., the hyperedges
$\{A,B,C\}$ and $\{A,F\}$ are those corresponding to the views $v_1(A,B,C)$ and $v_2(A,F)$, and where (hyper)edges associated with the query
views are still depicted with their original orientation in $Q_4$, as to make the correspondence clearer). Moreover, the figure reports the two
queries
\[
\begin{array}{ll}
Q_5: & r(A,B)\wedge r(B,C)\wedge r(A,C)\\
Q_6: & r(D,B)\wedge r(B,C)\wedge r(D,C).
\end{array}
\]
Note that $Q_5$ and $Q_6$ are two (isomorphic) cores of $Q_4$, but they have different structural properties. Indeed, $(\HG_{Q_5},\HG_{\V_4})$
admits a tree projection (note in the figure that the view over $\{A,B,C\}$ ``absorbs'' the cycle), while $(\HG_{Q_6},\HG_{\V_4})$ does not.
\hfill $\lhd$
\end{example}

\medskip\noindent \textbf{Computation Problem.} Armed with the observation exemplified above, the relationship between consistency and
structural properties will be next stated by considering the existence of a tree projection for \emph{some} core of the query $Q$.

In addition, to properly deal with arbitrary sets of output variables (which may be not included in any core of $Q$), we need to define an
``output-aware'' notion of covering by tree projections, where cores are forced to contain the desired output variables.

\begin{definition}\label{def:tpCovered}
For any set of variables $O$ occurring in some atom $w\in\W$, define $\atom(O)$ to be a fresh atom (with a fresh relation symbol) over these
variables, i.e., such that $O=\vars(\atom(O))$.
Then, we say that $O$ is $\tpCovered$ in $Q$ (w.r.t.~$\W$) if there exists some core $Q'$ of $Q\wedge \atom(O)$ such that $(\HGUnoSet,\HGDue)$
has a tree projection. \hfill $\Box$
\end{definition}

A first easy observation is that the $\tpCovered$ property holds for every set of variables occurring in every query atom, whenever
$(\HG_Q,\HGDue)$ has a tree projection.

\begin{fact}\label{fact:tpcoveredAndTP}
Assume that $(\HG_Q,\HGDue)$ has a tree projection. Then, for every $q\in\atoms(Q)$ and every $O\subseteq\vars(q)$, $O$ is $\tpCovered$ in $Q$
(w.r.t.~$\V$).
\end{fact}
\begin{proof}
Let $q$ be any atom occurring in $Q$ and take any $O\subseteq \vars(q)$. Let $Q'$ be any core of $Q\wedge\atom(O)$. Since $Q'$ is a subquery of
$Q\wedge\atom(O)$ and  $O\subseteq \vars(q)$, $\HGUnoSet\leq \HG_Q$. Thus, $\HGUnoSet\leq \HG_Q\leq\HG_a\leq \HGDue$, where $\HG_a$ is any tree
projection of $(\HG_Q,\HGDue)$, which exists by hypothesis.~\hfill~$\Box$
\end{proof}

We next show that the above fact may be extended to those atoms occuring in some core of $Q$ having a tree projection.

\begin{lemma}\label{lem:coreMembershipEntailsTPcovered}
Let $q\in\atoms(Q')$ be an atom occurring in some core $Q'$ of $Q$ for which $(\HGUnoSet,\HGDue)$ has a tree projection. Then, $\forall
O\subseteq\vars(q)$, $O$ is $\tpCovered$ in $Q$ (w.r.t.~$\V$).
\end{lemma}
\begin{proof}
Let $O\subseteq\vars(q)$, and consider the query $Q\wedge\atom(O)$. We first claim that there is a homomorphism from $Q\wedge\atom(O)$ to
$Q'\wedge\atom(O)$.
Indeed, since $Q'\in\cores(Q)$, it is also a \emph{retract} of $Q$ (see, e.g., \cite{GN08}); that is, there is a homomorphism $f$ from $Q$ to
$Q'$ which is the identity on its range (i.e., $f(X)=X$, for every term $X$ occurring in $Q'$). Moreover, $O\subseteq \vars(Q')$, because
$q\in\atoms(Q')$. It follows that $f$ is also a homomorphism from $Q\wedge\atom(O)$ to $Q'\wedge\atom(O)$. In particular, note that $f$ maps
the atom $\atom(O)$ to itself. We thus conclude that $Q'\wedge\atom(O)$ is also a core of $Q\wedge\atom(O)$, because $\atom(O)$ is over a fresh
relation symbol and hence must belong to any core, and dropping atoms from $Q'$ would contradict the minimality of $Q'$ as a core of $Q$.
Finally, since $\vars(\atom(O))= O \subseteq \vars(q)$ and $q\in\atoms(Q')$, the hypergraph associated with $Q'\wedge\atom(O)$, say $\HG'$, is
such that  $\HG'\leq\HGUnoSet$. Hence, any tree projection of $\HGUnoSet$ w.r.t. $\HGDue$, which exists by hypothesis, is a tree projection of
$\HG'$ w.r.t. $\HGDue$. That is, $O$ is $\tpCovered$ in $Q$ (w.r.t. $\V$).~\hfill~$\Box$
\end{proof}

\begin{example}\label{ex:notTP}\em
Consider again the setting of Example~\ref{ex:MultipleCores-bis}. The core $Q_5$ contains the atoms $r(A,B)$, $r(B,C)$, and $r(A,C)$, and we
have noticed that $Q_5$ admits a tree projection. Therefore, we can apply Lemma~\ref{lem:coreMembershipEntailsTPcovered} to conclude that the
sets of variables $\{A,B\}$, $\{B,C\}$, and $\{A,C\}$ are $\tpCovered$ in $Q_4$.

Consider now the set of variables $\{F,E\}$, which does not occur in any core of the query, and the novel query $Q_4 \wedge
\mathit{atom}(\{F,E\})$. This query has a unique core, which is again depicted in Figure~\ref{fig:triangle-bis}. Notice that this core does not
coincide with any of the two cores of the original query. Yet, it admits a tree projection, consisting of the hyperedges $\{F,E\}$, $\{A,E\}$,
and $\{A,B,C\}$, as shown in the figure. Thus, $\{F,E\}$ is $\tpCovered$ in $Q_4$.

On the other hand, the hypergraphs associated with the cores of $Q_4 \wedge \mathit{atom}(\{D,C\})$ and $Q_4 \wedge \mathit{atom}(\{D,B\})$ are
precisely the same as the hypergraph $\HG_{Q_6}$ associated with the core $Q_6$, that is, the triangle with vertices $D,B$, and $C$, having no
tree-projection w.r.t.~$\HG_{V_4}$. Hence, $\{D,C\}$ and $\{D,B\}$ are not $\tpCovered$ in $Q_4$.

Finally, for an example application of Definition~\ref{def:tpCovered} with arbitrary set of variables (i.e., not just contained in query
atoms), consider the set $\{A,F\}$. Consider then the query $Q_4 \wedge \mathit{atom}(\{A,F\})$ and note that its core does not have a tree
projection. Thus, $\{A,F\}$ is not $\tpCovered$ in $Q_4$.~\hfill~$\lhd$
\end{example}

The notion of \emph{tp-covering} plays a crucial role in establishing consistency properties. To help the intuition, this role is next
exemplified.

\begin{example}\em
Consider again the setting of Example~\ref{ex:MultipleCores} (and Example~\ref{ex:MultipleCores-bis}) and the database $\DB_4$ shown in
Figure~\ref{fig:triangle} over the relation symbol $r$ (in $Q_4$) and the symbols for the views in $\V_4=\views(Q_4)\cup \{v_1(A,B,C),
v_2(A,F)\}$. Recall from Example~\ref{ex:local} that $\V_4^{\onDB_4}$ is locally consistent.

Observe that for the query view $v_4(A,C)$, $v_4(A,C)^{\onDB_4}$  consists of the two tuples/homomorphisms $\tuple{a_1,c_1}$ and
$\tuple{a_2,c_2}$. That is, this query view provides exactly the two homomorphisms in $Q_4^{\onDB_4}[\{A,C\}]$, i.e., the answers of $Q_4$
projected over the variables ($A$ and $C$) of the view $w_{r(A,C)}$. Note that the same property holds for the views over the set of variables
$\{A,C\}$, $\{A,B\}$, $\{B,C\}$, $\{F,E\}$, $\{A,E\}$, and $\{A,B,C\}$. Interestingly, each one of this set is $\tpCovered$ in $Q_4$ (see also
Example~\ref{ex:notTP}).

On the other hand, each one of the sets $v_7(D,B)^{\onDB_4}$, $v_8(D,C)^{\onDB_4}$, and $v_2(A,F)^{\onDB_4}$ contains two homomorphisms that do
not correspond to any answer of the query (suitably projected over the variables of interest), which are those identified by the tuples marked
with the symbol ``{\small $\blacktriangleright$}'' in Figure~\ref{fig:triangle}. In fact, we observe that, in this case, $\{D,C\}$, $\{D,B\}$,
and $\{A,F\}$ are not $\tpCovered$ in $Q_4$.~\hfill~$\lhd$
\end{example}

In the above example, the fact that homomorphisms that are not correct answers are associated with views whose variables are not $\tpCovered$
is not by chance. Indeed, the intuition is now that to guarantee global consistency by just enforcing local consistency, all the variables
contained in query atoms must be $\tpCovered$.

Next, we establish a lemma that actually proves a slightly more general result dealing with any set of output variables covered by some view.
For a set of variables $O$, let $\vcovers(O)$ denote the set of all views $w\in\W$ such that $O\subseteq \vars(w)$.

\begin{lemma}\label{lem:TPimpliesLCtoGC-specific}
Assume that \emph{$\W^\onDB$} is locally consistent. For any set of variables $O$ that is $\tpCovered$ in $Q$, \emph{$w^\onDB[O] \subseteq
Q^\onDB[O]$} holds, for every $w\in\vcovers(O)$. Moreover, if $\bar w^\onDB$ is view consistent w.r.t.~$Q$, for some $\bar w\in\vcovers(O)$,
then we actually get all the right homomorphisms for all of them, i.e., \emph{$w^\onDB[O] = Q^\onDB[O]$} holds, for every $w\in\vcovers(O)$.
\end{lemma}
\begin{proof}
Let $Q_e= Q\wedge\atom(O)$. Assume that $O$ is $\tpCovered$ in $Q$, that is, there exists $Q'\in\cores(Q_e)$ for which $(\HG_{Q'},\HGDue)$ has
a tree projection. Since $Q'$ is a core, it is also a \emph{retract} of $Q_e$; that is, there is a homomorphism $f$ from $Q_e$ to $Q'$ such
that $f(X)=X$, for every term $X$ occurring in $Q'$. Clearly, $f$ is a homomorphism from $Q$ to $Q'$, too.
Then, for every (legal) database $\barDB$, $Q'^\onbarDB \subseteq Q^\onbarDB[\vars(Q')]$. Moreover, consider the query $WQ'$ where we have
query views in place of the original query atoms, that is,
$$WQ'= \atom(O) \wedge \bigwedge_{q\in \atoms(Q')\setminus \{\atom(O)\}} w_q.$$
Because $\barDB$ is a legal database, we immediately get $WQ'^\onbarDB = Q'^\onbarDB \subseteq Q^\onbarDB[\vars(Q')]$ and, hence,
$WQ'^\onbarDB[\bar X] \subseteq Q^\onbarDB[\bar X]$ holds as well, for any $\bar X\subseteq \vars(Q')$.

Now consider any (legal) database $\DB$ such that \emph{$\W^\onDB$} is locally consistent, and any tree projection $\HG_a$ of
$(\HGUnoSet,\HGDue)$. Assume w.l.o.g. that $\nodes(\HG_a)=\nodes(\HG_{Q'})$ (otherwise, just drop possible additional variables, and you still
get a tree projection of $(\HGUnoSet,\HGDue)$). Observe that $O\subseteq h_O$, for some hyperedge $h_O$ of $\HG_a$. Indeed,
$\atom(O)\in\atoms(Q')$, since $\atom(O)$ is defined on a fresh relation symbol, and thus this atom must occur in every core of $Q_e$, i.e.,
$O\in \edges(\HGUnoSet)$. Let us associate with $\HG_a$ the following query:
$$Q_a = WQ' \wedge \bigwedge_{h\in \edges(\HG_a)} \atom(h).$$
For any fresh atom $\atom(h)\in\atoms(Q_a)$ (including $\atom(O)$), let $\atom(h)^\onDB= v^\onDB[h]$, where $v\in\V$ is any view satisfying
$h\subseteq \vars(v)$, chosen according to some fixed (arbitrary) criterium. Such a view always exists because $\HG_a$ is a tree projection of
$(\HGUnoSet,\HGDue)$.

Note that $Q_a^\onDB \subseteq WQ'^\onDB$, because $WQ'$ is a subquery of $Q_a$. By construction $Q_a$ is a simple acyclic query, and
$\atoms(Q_a)^\onDB$ is locally consistent because all these relations are projections of views in the locally consistent set $\W^\onDB$. Thus,
by the results in~\cite{BFMY83}, $Q_a$ is globally consistent and we get, for the atom $\atom(O)$, $\atom(O)^\onDB = Q_a^\onDB[O]\subseteq
WQ'^\onDB[O]\subseteq Q^\onDB[O]$. Moreover, since $\W^\onDB$ is locally consistent, this property must hold for every $w^\onDB$, with $w\in\W$
and $O\subseteq \vars(w)$.
That is, {$w^\onDB[O] \subseteq Q^\onDB[O]$} holds, for every $w\in\vcovers(O)$.

Assume now that the output variables $O$ are covered by some view consistent atom, i.e., $O\subseteq \vars(\bar w)$ for some $\bar w\in\W$ such
that $Q^\onDB[\vars(\bar w)]\subseteq \bar w^\onDB$ and thus $Q^\onDB[O]\subseteq \bar w^\onDB[O]$. Since $\W^\onDB$ is locally consistent, it
follows that $\bar w^\onDB[O]=\atom(O)^\onDB$ and thus $Q^\onDB[O]\subseteq \atom(O)^\onDB$. Combined with the above relationship, we get the
desired equality $Q^\onDB[O] = \atom(O)^\onDB$. Again, since $\W^\onDB$ is locally consistent, this property must hold for every $w^\onDB$,
with $w\in\W$ and $O\subseteq \vars(w)$. That is, $Q^\onDB[O] = w^\onDB[O]$, for every $w\in\vcovers(O)$.~\hfill~$\Box$
\end{proof}

Since query views are always view consistent (over legal databases), we immediately get the following sufficient condition for the global
consistency, which clearly also holds for restricted tree projections corresponding to 
decomposition methods.

\begin{theorem}\label{thm:TPimpliesLCtoGC}
Assume that, for every $q\in\atoms(Q)$, $\vars(q)$ is $\tpCovered$ in $Q$ (w.r.t.~$\V$). Then, for every database $\DB$, \emph{$\lc(\W,\DB)$}
entails \emph{$\gc(\W,\DB,Q)$}.
\end{theorem}

\begin{figure}[t]
  \centering
  \includegraphics[width=0.32\textwidth]{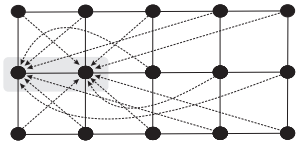}
  \caption{Mapping an undirected grid into an edge.}\label{fig:grid}
\end{figure}

Having a tree projection of the full query is therefore not necessary for getting global consistency through local consistency. For instance,
an unsuspectedly easy class of queries consists of the grid queries of the form $GQ_n = \bigwedge_{{X,Y}\in E_n} (e(X,Y)\wedge e(Y,X))$, where
$E_n$ is the edge set of an $n\times n$ grid. Indeed, while such grids are well known obstructions to the existence of tree decompositions, any
of their edges is a core (and, thus, trivially acyclic)---see Figure~\ref{fig:grid}. Therefore, even the smallest possible set of views
$\V=\views(GQ_n)$ is sufficient to obtain global consistency by enforcing local consistency.

As we shall prove in Section~\ref{sec:back}, Theorem~\ref{thm:TPimpliesLCtoGC} defines the most general possible condition to guarantee global
consistency, which is what we need to answer the query by exploiting local consistency if the output variables are included in some query atom.

\medskip

\noindent \textbf{Decision Problem.}
The situation is rather different if we just look for the most general sufficient conditions to solve the decision problem
$Q^\onDB\neq\emptyset$. In this case, it is sufficient the existence of a tree projection  of any structure for which there is an endomorphism
of the query. Of course, any such a subquery $Q'$ is homomorphically equivalent to $Q$,  denoted by $Q' \homEquiv Q$ in the following.
In fact, the concept of \emph{tp-covering} is immaterial here, given that we are not interested in output variables (i.e., $O=\emptyset$).
Thus, as a special case of our analysis on the computation problem, we get the following result, which generalizes to tree projections (where
cores may behave differently) a similar sufficient condition known for the special cases of tree decompositions~\cite{DKV02}, and generalized
hypertree decompositions~\cite{CD05}.

\begin{theorem}\label{thm:TPimpliesLCtoGC-decision}
Assume there is a subquery $Q'\homEquiv Q$ for which $(\HGUnoSet,\HGDue)$ has a tree projection. Then, for every 
 database $\DB$, \emph{$\lc(\W,\DB)$} entails \emph{$Q^\onDB\neq\emptyset$}.
\end{theorem}
\begin{proof}
Let $\HG_a$ be a tree projection of $(\HGUnoSet,\HGDue)$, for some $Q'\homEquiv Q$. Then, it is also a tree projection of $(\HG_{Q''},\HGDue)$,
for any $Q''\in\cores(Q')\subseteq\cores(Q)$, because $\HG_{Q''}\leq \HGUnoSet$. From Lemma~\ref{lem:coreMembershipEntailsTPcovered}, for any
(query atom) $q\in\atoms(Q'')$, $\vars(q)$ is $\tpCovered$ in $Q$ and thus, from Lemma~\ref{lem:TPimpliesLCtoGC-specific},
$Q^\onDB[\vars(q)]=w_q^\onDB$. Then, whenever $\lc(\W,\DB)$ holds, $w_q^\onDB\neq\emptyset$ and hence $Q^\onDB\neq\emptyset$.~\hfill~$\Box$
\end{proof}

Note that the above condition is more liberal than what we need for having the global consistency. In the next section we prove that it is in
fact also a necessary condition as far as the decision problem is concerned.

Moreover, we point out that, from an application perspective, either results above may be useful only if we have some guarantee (or some
efficient way to check) that the required conditions are met. Otherwise, as it happens for the decision problems in the special cases of
(generalized) (hyper)tree decompositions~\cite{CD05,G07}, we are in a {\em promise} setting where, in general, we are not able to actually
compute any full (and thus polynomial-time checkable) query answer (or disprove the ``promise''). In particular, it has been observed in a
slightly different setting by~\cite{GS10} (see, also, \cite{SGG08,BDGM09}) that, rather surprisingly, the global consistency property (and
hence having a full reducer) is not sufficient to actually compute a full query answer (unless $\Pol=\NP$). Intuitively this is due to the fact
that, as soon as we fix some tuple in a relation in order to extend it to a full solution, we are changing the set of available query
endomorphisms and thus we may loose the property of some variables to be $\tpCovered$. As a consequence, subsequent propagations are not
guaranteed to maintain the global consistency.

\smallskip

\subsection{...and Back to Tree Projections}\label{sec:back}

The question of whether the cases in which local consistency implies global consistency precisely coincide with the cases in which there is a
tree projection of the query with respect to a set of views was a long-standing open problem in the literature~\cite{GS84,SS93}. We next answer
this question, both in the setting considered in those papers (where all relation symbols in the query are distinct), with the answer being
positive there, and in the unrestricted setting where the answer is instead negative. In fact, we precisely characterize the relationships
between local and global consistency and tree projections in the general setting too, by showing that tree projections are still necessary, but
not necessarily involving the query as a whole.

\medskip

\noindent\textbf{Decision Problem.} We start with the problem of checking whether the given query is not empty. Theorem~\ref{thm:svuotamento}
below provides the counterpart of Theorem~\ref{thm:TPimpliesLCtoGC-decision}. The proof requires some preparation.

Let $\DB$ be a database over the vocabulary $\DS$. For the following results, we assume that each relation symbol $r\in\DS$ of arity $\rho$ is
associated with a set of $\rho$ (distinct) \emph{attributes} that identify the $\rho$ positions available in $r$. In this context, $r$ is also
called \emph{relational schema}, and $\DS$ is called \emph{database schema}.
An \emph{inclusion dependency} is an expression of the form $r_1[S]\subseteq r_2[S]$, where $r_1$ and $r_2$ are two relational schemas in $\DS$
and $S$ is a set of attributes that $r_1$ and $r_2$ have in common. A database $\DB$ over $\DS$ satisfies this inclusion dependency if, for
each tuple $t_1\in r_1^\onDB$, there is a tuple $t_2\in r_2^\onDB$ with $t_1[S]=t_2[S]$ (where $[\cdot]$ is here the classical projection
relational operator applied to a set of attributes). Moreover, if $\DB$ satisfies each inclusion dependency in a given set $I$, then we simply
say that $\DB$ satisfies $I$.

Define $\WDS$ as the set of \emph{canonical atoms} associated with the schema $\DS$, that is, the set containing, for each relation $r$ of
$\DS$, the atom $r({\bf u})$ having as its variables the attributes of $r$. A conjunctive query $Q$ is said to be a {\em canonical query} for
$\DS$ whenever it consists of atoms from $\WDS$, i.e., $\atoms(Q)\subseteq \WDS$ holds.

We are now ready to state a fundamental lemma on \emph{union} of conjunctive queries, i.e., on queries of the form $\overline
Q=Q_1\vee\cdots\vee Q_n$, where $Q_i$ is a conjunctive query $\forall i\in\{1,..,n\}$. We are interested in unions of Boolean queries, so that
$\overline Q^\onDB\neq\emptyset$ if (and only if) $Q_i^\onDB\neq\emptyset$ for some query $Q_i$ in the union $\overline Q$.
The ingredients in the lemma are a recent result on the  {\em finite controllability} of unions of conjunctive queries in the framework of
databases under the open-world assumption~\cite{R06}, and a connection between tree projections and the \emph{chase} procedure firstly observed
in \cite{SS93}.

\begin{lemma}\label{lem:claim:ID}
Let $\DS$ be a database schema equipped with a set $I$ of inclusion dependencies. Let $\overline Q$ be a union of canonical queries for $\DS$
such that,$\forall$ (finite) $\DB\neq\emptyset$ over $\DS$, $\DB$ satisfies $I$ $\Rightarrow$ $\overline Q^\onDB \neq \emptyset$. Then, there
exists a conjunctive query $Q'$ in the union $\overline Q$ such that $(\HG_{Q'} ,\HG_{\mathcal A(DS)})$ has a tree projection.
\end{lemma}
\begin{proof}
Unlike all other proofs in the paper, we next deal both with finite and infinite databases, and thus we always point out whether a database is
(or may be) infinite. All databases are implicitly assumed to be over the database schema $\DS$. From the hypothesis, the following property
holds for $\overline Q$:
\begin{description}
\item[$\rm P_1$]  \ $\forall$ finite $\DB\neq\emptyset$, $\DB$ satisfies $I$ $\Rightarrow$ $\overline Q^\onDB \neq \emptyset$.
\end{description}

Let us start by taking an arbitrary atom $r_w(X_1,\ldots,X_m)$ in $\overline Q$, and let $\DB_0 = \ \{\ r_w(c_{X_1},\ldots,c_{X_m})\ \}$, where
$c_{X_1},\ldots,c_{X_m}$ are fresh (distinct) constants. Trivially, $\rm P_1$ entails the following property:
\begin{description}
\item[$\rm P_2$] \ $\forall$ finite $\DB \supseteq \DB_0$, $\DB$ satisfies $I$ $\Rightarrow$  $\overline Q^\onDB \neq \emptyset$.
\end{description}

Recall that the possibly infinite database $\chase(I,\DB_0)$ is built from $\DB_0$ by adding iteratively new tuples to satisfy inclusion
dependencies in $I$, until no dependency is violated by the current database (see for instance \cite{abit-etal-95}). In the following, it is
convenient to represent $\chase(I,\DB_0)$ as a tree $T$ of tuples rooted at $r_w(c_{X_1},\ldots,c_{X_m})$, and where edges are built as
follows. Let $\DB_i$ denote the set of all the tuples in $\chase(I,\DB_0)$ associated with nodes in the first $i$ levels of $T$ (the root is
level 0).
Let $r({\bf t})$ be a node of $T$ at level $i$. For each inclusion dependency $r[A]\subseteq r'[A]\in I$ such that there is no tuple $r'({\bf
t}')\in \DB_{i}$ that matches with $r({\bf t})$ over the attributes in $A$, a node $r'({\bf t}'')$ is added as a child of $r({\bf t})$, where
$r'({\bf t}'')$ is a fresh tuple that matches with $r({\bf t})$ over the attributes in $A$ and contains fresh constants of the form $c_Y$, for
any (other) attribute $Y\notin A$ in the schema of relation $r'$.

A well known property of $\chase(I,\DB_0)$ is that it maps via homomorphism to any other (possibly infinite) database that satisfies $I$ and
includes the non-empty database $\DB_0$. Therefore, whenever $\overline Q^{{\fontOn\chase(I,\DB_0)}} \neq \emptyset$, the same holds for every
database that satisfies $I$ and includes $\DB_0$.

We now use the {\em finite controllability} result by Rosati~\cite{R06} which, applied to our $\overline Q$, $I$, and $\DB_0$, reads as
follows: the answer of $\bar Q$ is not empty on every (possibly infinite) database that satisfies $I$ and includes $\DB_0$ if, and only if, the
answer of $\bar Q$ is not empty on every finite database that satisfies $I$ and includes $\DB_0$ (by Theorem~2 in~\cite{R06}).\footnote{In
particular, it is shown that this is equivalent to the condition $\overline Q^{{\fontOn\fchase(I,\DB_0,m)}} \neq \emptyset$, where $m$ is a
finite natural number that depends on the given instance (including the query) and $\fchase(I,\DB_0,m)$ is the so-called {\em finite chase},
that is, a non-empty finite database playing the same role of  the (possibly) infinite chase, as far as the evaluation of $\bar Q$ is
concerned.}
Therefore, $\rm P_2$ implies the following property:
\begin{description}
\item[$\rm P_3$] \  $\overline Q^{{\fontOn\chase(I,\DB_0)}} \neq \emptyset$.
\end{description}

Because $\overline Q$ is a union of conjunctive queries, this means that there is a query $Q'$ in $\overline Q$ having a homomorphism $h:
\vars(Q')\mapsto \U_c$ from $Q'$ to $\chase(I,\DB_0)$, where $\U_c$ is the universe of $\chase(I,\DB_0)$. In particular, from a well known
result of Johnson and Klug~\cite{JK84}, we may assume, w.l.o.g., that $h$ maps $Q'$ to a finite subtree $T_f$ of $T$.

Observe now that $h$ is a bijection. Indeed, $\DB_0$ contains the one tuple $r_w(c_{X_1},\ldots,c_{X_m})$ with a distinct constant for each
attribute of $r_w$ and, by definition of $\chase(I,\DB_0)$, any constant $c_{Y}$ can never be used for an attribute different from $Y$. In
fact, either $c_{Y}$ belongs to the starting tuple and it is then propagated to fresh tuples by the chase generating-rule, or it is a fresh
constant belonging to a tuple created to satisfy some inclusion dependency (which does not involve attribute $Y$).
Moreover, recall that attributes in $\DS$ are in fact variables in $Q'$, because the latter is a canonical query. Then, since $h$ is a
homomorphism, for each variable (attribute) $Y$, $h(Y)$ has the form $c_Y$ for some constant $c_Y$ occurring in tuples of $\chase(I,\DB_0)$.

We now define a labeling $\lambda$, associating each node of $T_f$ with a set of variables in $\vars(Q')$. Let $V=\{h(X) \mid X\in
\vars(Q')\}$.
For each vertex $p=r(c_{Y_1},...,c_{Y_n})$ in $T_f$, define $\lambda(p)$ as the set $\{ h^{-1}(c_{Y_i}) \mid c_{Y_i}\in V \}$.
Let $p_1$ and $p_2$ be two vertices of $T_f$ such that $X\in \lambda(p_1)\cap \lambda(p_2)$ is a variable in $\vars(Q')$. Consider the chase
constant $h(X)$, which occurs in $p_1$ and $p_2$ in $T_f$. Let $p_X$ be the top-most vertex of $T_f$ where $h(X)$ occurs. Because of the chase
generating-rule, each node in the path from $p_X$ to $p_1$ (resp., $p_2$) contains the constant $h(X)$. Thus, since $T_f$ is a tree, $h(X)$
occurs in the path between $p_1$ and $p_2$. Therefore, $X$ occurs in $\lambda$-labeling of each vertex in this path, too.

Now consider the hypergraph $\HG_a$ containing exactly one hyperedge $\lambda(p)$, for each vertex $p$ of $T_f$, and note that $\HG_a$ is
acyclic, because we have actually just shown that the $\lambda$-labeling on $T_f$ defines a join tree of $\HG_a$. Moreover, since $h$ is a
homomorphism from $Q'$ to $\chase(I,\DB_0)$, for each atom $q\in \atoms(Q')$ there exists a vertex $p=h(q)$ in $T_f$ for which
$\lambda(p)=\vars(q)$; thus, $\HG_{Q'}\leq \HG_a$. Finally, by construction, each hyperedge $\lambda(p)$ in $\HG_a$ is built from a tuple
$p=r(c_{Y_1},...,c_{Y_n})$ of $\chase(I,\DB_0)$, hence a tuple of (the relation of) some canonical atom $a_r$ in $\WDS$. Moreover, we observed
that, for each variable $Y_i\in \lambda(p)$, $h^{-1}(c_{Y_i})=Y_i\in \vars(a_r)$. Then, $\lambda(p)\subseteq \vars(a_r)$, and hence $\HG_a\leq
\HG_{\mathcal A(DS)}$. All in all, we have shown that, for the query $Q'$ in $\overline Q$, there is a tree projection of $\HG_{Q'}$ w.r.t.
$\HG_{\mathcal A(DS)}$.~\hfill~$\Box$
\end{proof}

\begin{theorem}\label{thm:svuotamento}
Assume there is no tree projection of $(\HGUnoSet,\HGDue)$, for each core $Q'\in \cores(Q)$. Then, local consistency does not entail global
consistency. In particular, there exists a (legal) database $\DB$ such that \emph{$\lc(\W,\DB)$} holds but \emph{$Q^\onDB = \emptyset$}.
\end{theorem}
\begin{proof}
Recall that we assumed w.l.o.g. that no constants or repeated variables occur in the views in $\V$, while the query $Q$ has no restriction.
Moreover, each view $w\in \V$ is over a distinct relation symbol (let us denote it by $r_w$, in the following), so that there is a one-to-one
correspondence between relations and views. Therefore, $\V$ identifies a database schema $\DS$ consisting of such a relation $r_w$, for each
$w\in\V$, whose list of attributes is precisely the list of variables of the view $w$. Thus, $\V$ is by construction the set of canonical atoms
associated with $\DS$.\footnote{We remark that the assumption that no constant or repeated variables occur in views is just for the sake of
presentation. If this assumption does not hold, it is sufficient to define instead a database schema $\DS'$ obtained from $\V$ by removing such
useless occurrences, to use its canonical atoms, and to manage, after the described construction, the correspondence between relations in
$\DS'$ and views in $\V$.}

Let us equip $\DS$ with the following set $I$ of inclusion dependencies: For each pair of views $w,w'\in \W$ such that $S=\vars(w)\cap
\vars(w')\neq\emptyset$, $I$ contains the two inclusion dependencies $r_w[S]\subseteq r_{w'}[S]$ and $r_{w'}[S]\subseteq r_{w}[S]$.

Observe that, by the construction of $I$, for each database $\DB$ over $\DS$, $\lc(\W,\DB)$ holds if, and only if, $\DB$ satisfies $I$ and
$\DB\neq\emptyset$ (recall that $Q$ is connected and $\vars(Q)=\vars(\W)$, hence $\HG_\V$ is also connected because $\views(Q)\subseteq \V$).

For any set of atoms $D$, let us denote by $\bigwedge D$ the Boolean conjunctive query defined as the conjunction of all atoms in $D$. Let
$\overline Q = \bigvee_{Q'\in \cores(Q)}  (\bigwedge \views(Q'))$ be the union of (Boolean) canonical queries for $\DS$ obtained by considering
the cores of $Q$, and assume that there is no tree projection of $(\HGUnoSet,\HGDue)$, and hence of $(\HG_{\views(Q')},\HGDue)$, for each core
$Q'\in \cores(Q)$. Then, by Lemma~\ref{lem:claim:ID}, there exists a (finite) database $\DB_f\neq\emptyset$ that satisfies $I$ and such that
$\forall Q'\in \cores(Q)$, $(\bigwedge\views(Q'))^{\onDB_f} = \emptyset$. In particular, because this database satisfies $I$, $\lc(\W,\DB_f)$
holds.

From $\DB_f$, let us now build a new legal database instance $\DB_f'$  over the  vocabulary including both views and query atoms. This database
is obtained by slightly changing the relations in $\DB_f$ in order to keep the information about the (active) domains of the variables, and by
adding the relation instances for the query atoms in $Q$. Recall that more query atoms may share the same database relation.

Let $q\in\atoms(Q)$ be any query atom defined over a relation symbol $r$ of arity $\rho$, and let $r_{w_q}(X_1,\dots,X_n)\in \views(Q)$ be the
query view $w_q$ associated with $q$. Recall that both constants and repeated variables may occur in $q$, so that $\rho\geq n$. Let
$r_{w_q}(c_1,...,c_n)$ be any tuple in $\DB_f$. Then, $\DB_f'$ contains a tuple $r_{w_q}(\tuple{X_1,c_1},\dots,\tuple{X_n,c_n})$ in the
relation instance for the query view $w_q\in \V$. Moreover, for the relation $r$, $\DB_f'$ contains a tuple $r(v_1,\dots,v_\rho)$ defined as
follows. For each $i\in \{1,\dots,\rho\}$: if some constant term $u_i$ occurs in $q$ at position $i$, then $v_i=u_i$; if some variable $X_j$
occurs in $q$ at position $i$, then  $v_i=\tuple{X_j,c_j}$. Note that this value may occur in $r(v_1,\dots,v_\rho)$ at different positions, if
$X_j$ occurs more than once in $q$. Moreover, if the relation $r$ is shared by different query atoms, such a tuple $r(v_1,\dots,v_\rho)$ will
be available to every atom defined over $r$, besides $q$.
Finally, for any (non-query view) $w$ over a relation $r_w$ and any tuple $r_{w}(c_1,...,c_n)\in\DB_f$, $\DB_f'$ contains a tuple
$r_{w}(\tuple{X_1,c_1},\dots,\tuple{X_n,c_n})$. No further tuples belong to $\DB_f'$.

As $\lc(\W,\DB_f)$ holds, we immediately have that $\lc(\W,\DB_f')$ holds, too. We now claim that $Q'^{\fontOn \DB_f'}=\emptyset$, for each
subquery $Q'\in \cores(Q)$, which entails $Q^{\fontOn \DB_f'}=\emptyset$. Indeed, assume for the sake of contradiction that there is a core
$Q'$ such that $Q'^{\fontOn\DB_f'}\neq \emptyset$, and let $h'$ be a homomorphism from $Q'$ to $\DB_f'$. Define $\pi_1$ and $\pi_2$ to be the
projections mapping a binary tuple $\tuple{u,c}$ to its first element $u$ and to its second element $c$, respectively; moreover, for a plain
(term) element $u$, $\pi_1(u)=\pi_2(u)=u$. In particular, for any tuple $r(v_1,\dots,v_\rho)$ in $\DB_f'$, where any value $v_i$ is either of
the form $\tuple{u_i,c_i}$ or of the form $u_i$  with $u_i$ being a constant term, we have $\pi_1(r(v_1,\dots,v_\rho)) =
r(\pi_1(v_1),...,\pi_1(v_\rho)) = r(u_1,...,u_\rho)$. By construction of the tuples in $\DB_f'$, the composition $h'\circ \pi_1$ is a
homomorphism from $Q'$ to $Q$ (if we obtain a certain tuple of terms after applying $\pi_1$, there must exist some query atom with that tuple
of terms). But, since $Q'$ is a core, we have that the image $Q''=(h'\circ \pi_1)(Q')$ is also a core in $\cores(Q)$, and thus $h'\circ \pi_1$
is actually an isomorphism. In particular, $h''=((h'\circ \pi_1)^{-1}\circ h')$ is now such that $h''(u_i)=\tuple{u_i,c_i}$. In particular,
whenever $u_i=X$, for some variable $X\in\vars(Q'')$,
$h''(X)=\tuple{X,c_i}$. 
It follows that $h''$ is a homomorphism from $Q''$ to $\DB_f'$. Then, we immediately get that $h''\circ \pi_2$ is a homomorphism from
$\bigwedge \views(Q'')$ to $\DB_f$. Indeed, by construction, for each atom $q\in\atoms(Q'')$ defined on a relation $r$, if
$r(u_1,\dots,u_\rho)\in \DB_f'$, then $\pi_2(r_{w_q}(\bar u_1,\dots,\bar u_n))\in \DB_f$ (with $\tuple{\bar u_1,\dots,\bar u_n}$ being the
tuple derived from $\tuple{u_1,\dots,u_\rho}$ by inverting the above construction, i.e., by eliminating constants and repeated variables).
However, the existence of this homomorphism  contradicts the fact that $(\bigwedge \views(Q''))^{\fontOn \DB_f}=\emptyset$ holds by the
construction of $\DB_f$.

Finally, note that $\DB_f'$ is legal. Indeed, for each query view $w_q$, by construction  $w_q^{\fontOn \DB_f'} \subseteq q^{\fontOn \DB_f'}$,
and $w_q$ is trivially view consistent because $Q^{\fontOn \DB_f'}=\emptyset$.~\hfill~$\Box$
\end{proof}

A consequence of the above result and Theorem~\ref{thm:TPimpliesLCtoGC-decision} is the precise characterization of the power of local
consistency, as far as the decision problem is concerned. This characterization was so far only known for the special case of treewidth and for
structures of fixed arity~\cite{ABD07}, where, however, all the cores enjoy the same structural properties (and hence such results are defined
in terms of ``the core'' of the query).

\begin{corollary}\label{cor:Decision}
The following are equivalent:

\vspace{-1mm}
\begin{enumerate}
\item[\em (1)] For every database $\DB$, \emph{$\lc(\W,\DB)$} entails \emph{$Q^\onDB \neq \emptyset$}.

\item[\em (2)] There is a subquery $Q' \homEquiv Q$ for which $(\HGUnoSet,\HGDue)$ has a tree projection.

\item[\em (3)] There is a core $Q''$ of $Q$ for which $(\HG_{Q''},\HGDue)$ has a tree projection.
\end{enumerate}
\end{corollary}
\begin{proof}
From Theorem~\ref{thm:TPimpliesLCtoGC-decision}, we know that (2) implies (1). Theorem~\ref{thm:svuotamento} entails that (1) implies (3).
Finally, (3) implies (2) because any core of $Q$ is homomorphically equivalent to $Q$.~\hfill~$\Box$
\end{proof}

\begin{figure*}[t]
  \centering
  \includegraphics[width=0.7\textwidth]{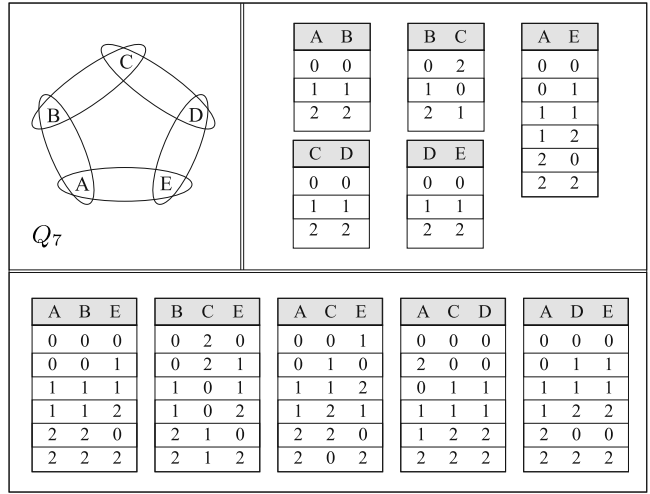}
  \caption{The (hypergraph of the) query $Q_7$, the non-query views in $\V_7$, and their respective tuples in the database $\DB_7$ of Example~\ref{ex:NonMonotone}.}\label{fig:ACE}
\end{figure*}

Eventually, we can specialize Corollary~\ref{cor:Decision} to the setting of simple queries (considered in many seminal papers about tree
projections, as \cite{GS84}), where every relation symbol occurs at most once in the query and thus the whole query is its (unique) core.

\begin{corollary}
Let $Q$ be a simple query. Then, the following are equivalent:

\vspace{-1mm}
\begin{enumerate}
\item[\em (1)] For every database $\DB$, \emph{$\lc(\W,\DB)$} entails \emph{$Q^\onDB \neq \emptyset$}.

\item[\em (2)] $(\HG_Q,\HGDue)$ has a tree projection.
\end{enumerate}
\end{corollary}

\begin{example}\label{ex:NonMonotone}\em
Consider the query
\[
\begin{array}{ll}
Q_7: & r_1(A,B)\wedge r_2(B,C) \wedge r_3(C,D)\wedge r_4(D,E)\wedge r_5(A,E),
\end{array}
\]
the set of views $\V_7=\{ v_1(A,B,E), v_2(B,C,E),$ $v_3(A,C,E),$ $v_4(A,C,D), v_5(A,D,E) \}$, and the database instance $\DB_7$ depicted in
Figure~\ref{fig:ACE}. It is easy to check that $(\V_7\cup \atoms(Q_7))^{\mbox{\rm \fontOn DB}_7}$ is local consistent but $Q_7^{\mbox{\rm
\fontOn DB}_7} =\emptyset$. Indeed, it can be checked that $(\HG_{Q_7},\HG_{\V_7})$ does not have a tree projection.
 \hfill $\lhd$
\end{example}

\medskip

\noindent\textbf{Computation Problem.} We next complete the picture and give the conditions that precisely characterize those cases where
answers of the query over output variables covered by some view may be immediately obtained by enforcing local consistency.
Again, we start with the problem where we are interested in query answers over some arbitrary set of output variables. In this case, requiring
that just some view covering $O$ is trustable is sufficient to allow all such answers to be immediately obtained.

\begin{theorem}\label{thm:correttezzaViews}
Let $O$ be any set of variables occurring in some view in $\W$. Then, the following are equivalent:

\vspace{-1mm}
\begin{enumerate}
\item[\em (1)] For each database $\DB$ such that \emph{$\lc(\W,\DB$)} holds, \emph{$w^\onDB[O]\subseteq Q^\onDB[O],$} for every
    $w\in\vcovers(O)$. If there is a view consistent $\bar w^\onDB$ with $\bar w\in\vcovers(O)$, then
         \emph{$w^\onDB[O] = Q^\onDB[O],$} for every $w\in\vcovers(O)$.

 \item[\em (2)] The set of variables $O$ is $\tpCovered$ in $Q$ (w.r.t.~$\V$).
\end{enumerate}
\end{theorem}
\begin{proof}
First observe that (2) entails (1), by Lemma~\ref{lem:TPimpliesLCtoGC-specific}. Then, in order show that (1) entails (2), it suffices to
consider the case where there exists $Q''\in\cores(Q)$ for which $(\HG_{Q''},\HGDue)$ has a tree projection. Otherwise, we immediately get the
contradiction that all views are incorrect for some database, from Theorem~\ref{thm:svuotamento}.
Consider the new query $Q_e=Q\wedge \atom(O)$, and assume by contradiction that $O$ is not $\tpCovered$ in $Q$. That is, for every
$Q'\in\cores(Q_e)$, $(\HGUnoSet,\HGDue)$ has no tree projections. We show that there exists a database $\DB$ such that $\lc(\W,\DB)$ but
$Q^\onDB[O] \subset a^\onDB[O]$, for every $a\in \vcovers(O)$, where  $\vcovers(O)\neq\emptyset$, by hypothesis.

Let $\W_e = \W \cup \{\atom(O)\}$. Since no core of $Q_e$ has tree projections, by Theorem~\ref{thm:svuotamento} it follows that there is a
(non-empty legal) database $\DB'$ such that $\lc(\W_e,\DB')$, but $Q_e^\onDBp = \emptyset$.
Now define a new database $\DB$ such that, for every $a\in \W_e$, $a^\onDB = a^\onDBp \cup Q^\onDBp[a]$, and where the relations in $\DB'$ over
which the original query atoms are defined are just copied into $\DB$.
By construction, $\lc(\W_e,\DB)$ holds, because $\lc(\W_e,\DB')$ holds and the tuples possibly added to any view are projections of mappings
over the full set of variables, as they are obtained from the total homomorphisms in $Q^\onDBp$. Moreover, note that only views are modified,
as no tuple is added to the relations over which the original atoms in the query are defined. Thus, $Q^\onDB=Q^\onDBp$ holds.

Observe that $\DB$ is a legal database instance w.r.t.~$Q$. Indeed, the relations for query views are still subsets of the relations of the
original query atoms (as in $\DB'$). Moreover, by construction, they include all tuples that are part of some query answer, and thus all query
views are view consistent w.r.t.~$Q$.

Recall now that we are considering the case where some cores of $Q$ have tree projections, and $\lc(\W_e,\DB)$ and hence $\lc(\W,\DB)$ hold.
From Theorem~\ref{thm:TPimpliesLCtoGC-decision}, it follows that $Q^\onDB = Q^\onDBp \neq\emptyset$. However, $(Q\wedge \atom(O))^\onDBp =
\emptyset$. It follows that all homomorphisms that are answers of $Q$ over $\DB'$ does not satisfy $\atom(O)$, that is, $Q^\onDBp[O]\cap
\atom(O)^\onDBp =\emptyset$, and recall that $\atom(O)^\onDBp \neq\emptyset$, because $\lc(\W_e,\DB')$ holds.

Therefore, we get the proper inclusion  $Q^\onDB[O] \subset \atom(O)^\onDB$. Indeed, $\atom(O)^\onDBp$ is not empty and all its tuples, which
do not belong to $Q^\onDBp[O]=Q^\onDB[O]\neq\emptyset$, are kept in $\atom(O)^\onDB$. Finally, since $\W_e^\onDB$ and hence $\W^\onDB$ are
locally consistent, this also entails $\atom(O)^\onDB= a^\onDB[O]$ and thus $Q^\onDB[O] \subset a[O]^\onDB$, for each view $a\in
\vcovers(O)$.~\hfill~$\Box$
\end{proof}

The following corollary is the specialization to the case where we are interested in output variables covered by some query atom.

\begin{corollary}\label{cor:General}
The following are equivalent:

\vspace{-1mm}
\begin{enumerate}
\item[\em (1)] For every database $\DB$, \emph{$\lc(\W,\DB)$} entails \emph{$\gc(\W,\DB,Q)$}.

\item[\em (2)] For each $q\in \atoms(Q)$, $\vars(q)$ is $\tpCovered$ in $Q$ (w.r.t.~$\V$).
\end{enumerate}
\end{corollary}
\begin{proof}
Since query views covers the variables of query atoms and are always view consistent w.r.t.~$Q$ in any legal database, the statement
immediately follows from Theorem~\ref{thm:correttezzaViews} and Theorem~\ref{thm:TPimpliesLCtoGC}.~\hfill~$\Box$
\end{proof}

The specialization of Corollary~\ref{cor:General} to the setting where every relation symbol occurs at most once in the query provides the
answer to the question posed by \cite{GS84}.

\begin{corollary}
Let $Q$ be a simple query. Then, the following are equivalent:

\vspace{-1mm}
\begin{enumerate}
\item[\em (1)] For every database $\DB$, \emph{$\lc(\W,\DB)$} entails \emph{$\gc(\W,\DB,Q)$}.

\item[\em (2)] $(\HG_Q,\HGDue)$ has a tree projection.
\end{enumerate}
\end{corollary}

Finally, we point out that Theorem~\ref{thm:correttezzaViews} may be equivalently stated in terms of any arbitrary (legal) database $\DB$, by
considering its reduct $\red(\V,\DB)$ obtained enforcing local consistency.

\begin{corollary}\label{cor:correttezzaViews}
Let $O$ be any set of variables occurring in some view in $\W$. Then, the following are equivalent:

\vspace{-1mm}
\begin{enumerate}
\item[\em (1)] For each database $\DB$, \emph{$w^\onDBp[O]\subseteq Q^\onDB[O],$} for every $w\in\vcovers(O)$, where $\DB'=\red(\V,\DB)$.
    If there is a view consistent $\bar w^\onDB$ with $\bar w\in\vcovers(O)$, then
         \emph{$w^\onDBp[O] = Q^\onDB[O],$} for every $w\in\vcovers(O)$.

 \item[\em (2)] The set of variables $O$ is $\tpCovered$ in $Q$ (w.r.t.~$\V$).
\end{enumerate}
\end{corollary}
\begin{proof}
 $(1) \Rightarrow (2)$ follows from the corresponding implication $(1) \Rightarrow (2)$ in Theorem~\ref{thm:correttezzaViews} which entails that, whenever $O$ is not $\tpCovered$ in $Q$ (w.r.t.~$\V$), there exists a locally-consistent legal database $\DB$ and a view $w\in\vcovers(O)$ such that $w^\onDB[O]\supset Q^\onDB[O]$.
  In fact, because it is locally consistent, $\DB=\red(\V,\DB)$ holds.

 $(2) \Rightarrow (1)$ follows from the corresponding implication in Theorem~\ref{thm:correttezzaViews}, and from the fact that the only tuples occurring in $\DB$ and deleted in its reduct $\DB'$ do not participate in any query answer. Therefore $\DB'$ is a legal locally consistent database.~\hfill~$\Box$
\end{proof}

\section{Application to Structural Decomposition Methods}\label{SDM}

In this section, we specialize our results about consistency properties and tree projections to the purely structural decomposition methods
described in the literature (both in the database and in the constraint satisfaction area), because all of them can be recast in terms of tree
projections. In fact, each of them can be seen as a method to define suitable set of views to be exploited for solving the given query
answering instance. Here, views represent subproblems over subsets of variables, whose solutions can be computed efficiently.

We also provide further results that hold on such special cases only, such as the positive answer to the question in~\cite{CD05} about
$k$-local consistency and generalized hypertree decomposition, and the precise relationship between acyclic queries and local consistency,
solved in~\cite{BFMY83} for the simple queries.

\subsection{Decomposition Methods and Views}

We start by formalizing the concept of structural decomposition method in our framework. Let the pair $(Q,\DB)$ be any query answering problem
instance. For any subset of variables $S\subseteq\vars(Q)$, let $(Q_{|S},\DB_{|S})$ be the {\em subproblem of} $(Q,\DB)$ {\em induced} by $S$
defined as follows: for each atom $a\in\atoms(Q)$ with $\vars(a)\cap S\neq\emptyset$, $Q_{|S}$ contains an atom $a'$ over a fresh relation
symbol $r_{a'}$ having $\vars(a)\cap S$ as its set of variables, and whose database relation is such that $a'^{\onDB_{|S}}=a^\onDB[S]$. No
further atom belongs to $Q_{|S}$, and no further relation belongs to $\DB_{|S}$.
Intuitively,  $(Q_{|S},\DB_{|S})$ is the most constrained subproblem of $(Q,\DB)$ where only variables from $S$ occur, because all atoms
involving (even partially) those variables are considered. In particular, for each subquery $Q'$ whose set of variables is $S$, we have
$Q'^\onDB \supseteq  Q_{|S}^{\onDB_{|S}} \supseteq Q^\onDB[S]$.

\begin{definition}\label{def:dm}
A {\em structural decomposition method} {\tt DM} is a pair of polynomial-time computable functions $\lDM$ and $\rDM$ that, given a conjunctive
query $Q$ and a database $\DB'$, compute, respectively, a view system  $\V=\lDM(Q)$ and a database $\DB''=\rDM(Q,\DB')$ over the vocabulary of
$\V$ such that:\footnote{For the sake of presentation, we do not consider $\FPT$ decomposition methods (where functions $\lDM$ and $\rDM$ are
computable in fixed-parameter polynomial-time), but our results can be extended easily to them.}

\vspace{-1mm}
\begin{itemize}
\item[--]  the database $\DB=\DB'\cup\DB''$ over the (disjoint) vocabularies of $Q$ and $\V$ is legal;

\item[--]  for each $w\in\V$, $w^\onDB \supseteq  Q_{|\vars(w)}^{\onDB_{|\vars(w)}}$. That is, any view $w$ contains at least the solutions
    of the subproblem of $(Q,\DB)$ induced by its variables  ({\em subproblem completeness}). \hfill $\Box$
\end{itemize}
\end{definition}

Note that the above completeness property is a local property, and clearly entails the (global) view consistency property for $\V^\onDB$.

Every known \emph{purely-structural} decomposition method {\tt DM}, where views (subproblems) are only determined by the query and do not
depend on the database instance, can be recast this way, with decompositions of $Q$ according to {\tt DM} being tree projections of
$(\HG_Q,\HG_\V)$. Indeed, all such methods are in fact subproblem-based, because any view relation $w^\onDB$ is instantiated with the solutions
$Q'^\onDB$ of some subquery $Q'$ (depending on the specific method),
which is not necessarily an induced subproblem. 
Some exemplifications of the above definition are discussed below.

\medskip \noindent \textbf{Tree Decompositions.} For any fixed natural number $k$, the {tree decomposition method}~\cite{DP89,Fre90} ($tw_k$) is
characterized by the functions $\it v\mbox{-}tw_k$ and $\it d\mbox{-}tw_k$ that, given a query $Q$ and a database $\DB$, build the view system
${\it v\mbox{-}tw_k(Q)}$ and the database ${\it d\mbox{-}tw_k(Q,\DB)}$. In particular, for each subset $S$ of at most $k+1$ variables, there is
a view $w_{S}$ over the variables in $S$ (i.e., $\vars(w_{S})=S$) whose tuples are the solutions of the subproblem induced by $S$ (or, more
liberally, the cartesian product of the set of constants that variables in $S$ may take).
An illustration of the view set characterizing treewidth is reported below.

\begin{figure*}[t]
  \centering
  \includegraphics[width=\textwidth]{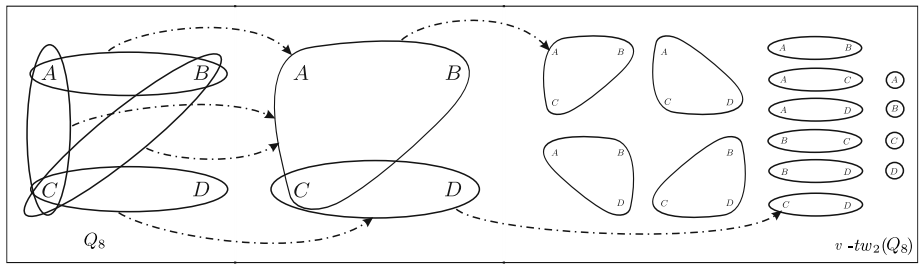}
  \caption{Structures discussed in Example~\ref{ex:treeDecomposition}.\vspace{-3mm}
  }\label{fig:Example-TW}
\end{figure*}

\begin{example}\label{ex:treeDecomposition}\em
Consider the query
\[
\begin{array}{ll}
Q_8: & r_1(A,B)\wedge r_2(B,C) \wedge r_3(A,C)\wedge r_4(C,D),
\end{array}
\]
whose associated hypergraph is depicted on the left of Figure~\ref{fig:Example-TW}. Consider the application on $Q_8$ of the tree decomposition
method. The set of views $\it v\mbox{-}{\it tw}_2(Q_8)$ defined by this method for $k=2$ is graphically illustrated on the right of
Figure~\ref{fig:Example-TW}. In fact, the figure shows how $Q_8$ can be covered via an acyclic hypergraph that consists of two hyperedges
covered by two available views, the largest of which includes three variables. In fact, the treewidth of $Q_8$ is 2.\hfill $\lhd$
\end{example}

\noindent \textbf{Generalized Hypertree Decompositions.} For any fixed natural number $k$, the {generalized hypertree decomposition
method}~\cite{gott-etal-99} (short: $hw_k$) is characterized by the functions $\it v\mbox{-}hw_k$ and $\it d\mbox{-}hw_k$ that, given a query
$Q$ and a database $\DB$, build the view system ${\it v\mbox{-}hw_k(Q)}$ and the database ${\it d\mbox{-}hw_k(Q,\DB)}$ where, for each subquery
$Q'$ of $Q$ such that $|\atoms(Q')|\leq k$, there is a view $w_{Q'}$ over all variables in $Q'$ (i.e., $\vars(w_{Q'})=\vars(Q')$)
and whose tuples are the answers of $Q'$. Note that $hw_k$ satisfies the subproblem completeness property too, because $Q'$ is in general more
liberal than the subproblem induced by $\vars(Q')$. Indeed, the latter also deals with further atoms where such variables occur (possibly
together with other variables not occurring in $Q'$).

\medskip \noindent \textbf{Acyclicity.}
Recall that a hypergraph is acyclic if, and only if, it has (generalized) hypertree width~1~\cite{gott-etal-99}. Therefore, the
\emph{acyclicity method} (short: ${\it acyc}$) is just the specialization of the above method for the case of $k=1$. In particular, $\it
v\mbox{-}acyc(Q)$ is precisely the set of query views $\views(Q)$.

\medskip \noindent \textbf{Fractional Hypertree Decompositions.}
For any fixed natural number $k$, consider the subqueries characterizing the {fractional hypertree decomposition method}~\cite{GM06}: they are
defined precisely as in the case of the generalized hypertree decomposition method, except that a view $w_{Q'}$ is built more generally if, for
a subquery $Q'$, its hypergraph $\HG_{Q'}$ has \emph{fractional edge-cover number}~\cite{GM06} at most $k$. Unfortunately, these views may be
exponentially-many even if $k$ is a fixed constant, and in fact there is no known polynomial time algorithm to decide whether the fractional
hypertree-width of a hypergraph is at most $k$. However, we may still define the required pair of polynomial-time functions $\it v\mbox{-}fw_k$
and $\it d\mbox{-}fw_k$ for this decomposition method, by actually exploiting for their computation the subproblems identified by Marx in his
$O(k^3)$ polynomial-time approximation of the fractional hypertree-width~\cite{M09}. Moreover, following the same kind of arguments used for
the generalized hypertree decompositions, it can be seen that the subproblem completeness property is satisfied by such a pair of functions,
too.

\medskip \noindent \textbf{Submodular Width.}
For the sake of completeness, note that the only known decomposition technique that does not fit the above framework is the one based on the
\emph{submodular width}~\cite{M10}. This method is in fact not ``purely'' structural. Indeed, according to this technique, a number of view
schemas are computed in fixed-parameter polynomial time (hence not polynomial-time, in general) by looking at the database $\DB$ of the given
instance, too (while $\lDM$ functions depend on the query only). Moreover, their associated database relations are not necessarily
subproblem-complete.

\subsection{Decomposition Methods and Consistency Properties}

By using Theorem~\ref{thm:correttezzaViews}, it is possible to characterize the power of local-consistency based algorithms in structural
decomposition methods, as stated in the following result.\footnote{For completeness, we observe that a similar result has been proved
in~\cite{GS10} in the setting of constraint satisfaction problems, by precisely exploiting Theorem~\ref{thm:correttezzaViews}.} In fact, this
result is not a trivial consequence of Corollary~\ref{cor:correttezzaViews}, as it is evident by contrasting their statements: here, the
database $\DB''$ for the views is {\em computed} from a database over the vocabulary of the query $Q$ only, according to the specific function
$\rDM$ characterizing the method {\tt DM}, while it is an arbitrary (legal) database in Corollary~\ref{cor:correttezzaViews}.

\begin{theorem}\label{prop:bis}
Let {\tt DM} be a decomposition method, let $Q$ be a conjunctive query,  and let $\V=\lDM(Q)$. The following are equivalent:

\vspace{-1mm}
\begin{enumerate}
\item[(1)] For every database $\DB$ (over the vocabulary of $Q$) and for every view $w\in\V$ with $O\subseteq \vars(w)$,
    $w^\onDBp[O]=Q^\onDB[O]$,  where $\DB'=\red(\V,\DB'')$ and $\DB''=\rDM(Q,\DB)$.

\item[(2)] A set of variables $O\subseteq \vars(Q)$ is $\tpCovered$ in $\V$.
\end{enumerate}
\end{theorem}

\begin{proof}
The fact that $(2) \Rightarrow (1)$ immediately follows from Corollary~\ref{cor:correttezzaViews}. We have to show that $(1) \Rightarrow (2)$
holds as well. Observe that if $O$ is not $\tpCovered$ in $Q$ w.r.t.~$\V=\lDM(Q)$, by Theorem~\ref{thm:correttezzaViews} we conclude the
existence of a locally consistent (legal) database $\DB=\DB_Q\cup \DB_\V$, with $\DB_Q$ being over the vocabulary of $Q$ and with $\DB_\V$
being over the vocabulary of $\V$, respectively, and  the existence of a  view $\bar w\in\V$ such that $\bar w^\onDB[O]\supset Q^\onDB[O]$,
with $O\subseteq \vars(\bar w)$. Let $\DB''=\rDM(Q,\DB_Q)$ be the database comprising the relations for the views in $\V$ built according to
method {\tt DM}, and let $\DB'=\red(\V,\DB'')$ be its reduct, obtained by enforcing local consistency.

We first claim that {\it the database $\DB_\V$ is included in $\DB''$}, formally, for any  $w\in\V$, $w^\onDB= w^{\onDB_\V} \subseteq
w^\onDBpp$. Consider such a view $w \in \V$, having variables $S = \vars(w)$, and the subproblem $(Q_{|S},\DB_{|S})$ of $(Q,\DB)$ induced by
$S$. By construction, only variables from $S$ occur in $Q_{|S}$ and thus, for each atom $a\in \atoms(Q_{|S})$, $\vars(a)\subseteq S$. It
trivially follows that the pair of hypergraphs $(\HG_{Q_{|S}},\HG_{\V})$ has a tree projection.
Let $\V^+=\V \cup \views(Q_{|S})$ be the set of views obtained by adding to $\V$ the query views  associated with the induced subproblem, and
let $\DB^+$ be the database obtained by adding to $\DB$ the relations in $\DB_{|S}$, as well as their copies on the relation symbols of the
query views $\views(Q_{|S})$. Clearly, $\DB^+$ is a legal database (w.r.t.~$Q^+$ and $\V^+$)
and $\V^+$ is a view system for $Q^+$. Moreover, since we just added new views to $\V$, the pair $(\HG_{Q_{|S}},\HG_{\V^+})$ has a tree
projection, too. In particular, from Fact~\ref{fact:tpcoveredAndTP}, $S$ is $\tpCovered$ in $Q_{|S}$ w.r.t.~$\V^+$.  Moreover, observe that the
database relations for the new views in $\V^+$ are just projections of the relations of the original
query views, which already belong to $\V$. Therefore, their presence has no impact on the local consistency property, and $\lc(\V^+,\DB^+)$
holds. By Theorem~\ref{thm:correttezzaViews}, for every $O'\subseteq S$, we get $w^\onDB[O']= w^\onDBpiu[O'] \subseteq
Q_{|S}^{\onDB_{|S}}[O']$. That is, $w^\onDB$ contains \emph{only} solutions of the subproblem induced by $w$.
On the other hand, the subproblem completeness condition entails that $w^\onDBpp \supseteq Q_{|S}^{\onDB_{|S}}$ . Hence the claim follows, as
for any chosen $w\in\V$ with variables $S = \vars(w)$, $w^\onDB\subseteq Q_{|S}^{\onDB_{|S}} \subseteq w^\onDBpp$.

To conclude, recall that $\lc(\V,\DB_\V)$ holds, so that $\DB_\V$ is a locally consistent database included in $\DB''$, and thus all its tuples
will survive after enforcing local consistency on $\DB''$, that is, all of them belongs to the reduct $\DB'=\red(\V,\DB'')$. Therefore,
$w^\onDB\subseteq w^\onDBp$,  $\forall w\in\V$. In particular, for the view $\bar w$ and the set of variables $O\subseteq \vars(\bar w)$, we
get  $Q^\onDB[O] \subset \bar w^\onDB[O]\subseteq \bar w^\onDBp[O]$, hence we get wrong solutions (over $O$) using the view $\bar w$ with the
database $\DB'$.~\hfill~$\Box$
\end{proof}

For the decision problem ($O=\emptyset$), we get the following special case.

\begin{corollary}\label{cor:DecisionDM}
Let {\tt DM} be a decomposition method, let $Q$ be a conjunctive query,
  and let $\V=\lDM(Q)$. The following are equivalent:

\vspace{-1mm}
\begin{enumerate}
\item[\em (1)] For every database $\DB$ (over the vocabulary of $Q$), $\red(\V,\rDM(Q,\DB))\neq\emptyset$ entails \emph{$Q^\onDB \neq
    \emptyset$}.

\item[\em (2)] There is a subquery $Q' \homEquiv Q$ for which $(\HGUnoSet,\HGDue)$ has a tree projection.

\item[\em (3)] There is a core $Q''$ of $Q$ for which $(\HG_{Q''},\HGDue)$ has a tree projection.
\end{enumerate}
\end{corollary}

If we consider decision problem instances ($O=\emptyset$) and the treewidth method ($\V={\it v\mbox{-}tw_k(Q)}$), from
Corollary~\ref{cor:DecisionDM}, we (re-)obtain the nice characterization of~\cite{ABD07} about the relationship between $k$-local consistency
and the treewidth of the core of $Q$.\footnote{As already observed, for treewidth and (generalized) hypertree-width isomorphic substructures
behave in the same way, so that all cores have equivalent properties. Thus, for these methods one may simply say ``the core'' $Q'$ (instead of
some core).}

If we consider the generalized hypertree-width ($\V={\it v\mbox{-}hw_k(Q)}$), we next provide the answer to the corresponding open question for
the unbounded arity case.
Recall that in~\cite{CD05} it was shown that if the core of $Q$ has generalized hypertree-width at most $k$, then the procedure enforcing
$k$-union (of constraints/atoms) consistency is always correct, i.e., the reduct of the database is not empty if, and only if, the query has
some answer. We next show that this sufficient condition is necessary, too.

In fact, observe that the following result does not follow immediately from Corollary~\ref{cor:DecisionDM}. Indeed, any core $Q'$ of $Q$ may be
much smaller than $Q$, and thus the set of views ${\it v\mbox{-}hw_k(Q')}$ available using $Q'$ is in general (possibly much) smaller than the
set of views ${\it v\mbox{-}hw_k(Q)}$ available when the whole query $Q$ is considered.
For an extreme example, think of the undirected grid (see again Figure~\ref{fig:grid}), where any edge is a core: in this case, the set of
available views for computing a hypertree decomposition of the core is precisely this one edge (for any $k$), while considering the whole
query, the available views comprise all unions of $k$ edges.

This subtle issue is irrelevant for the treewidth method, because such a technique considers all possible combinations of at most $k$
variables, and clearly only those variables occurring in the core are useful for computing any of its tree decompositions. Instead, when
generalized hypertree decomposition is considered, in principle using some particular combination of variables occurring in some atom outside
any core $Q'$ may be necessary for getting a width-$k$ generalized hypertree decomposition of $Q'$.

\begin{theorem}\label{thm:DecisionHW}
Let $Q$ be a conjunctive query, and let $\V={\it v\mbox{-}hw_k(Q)}$. The following are equivalent:

\vspace{-1mm}
\begin{enumerate}
\item[\em (1)] For every database $\DB$ (over the vocabulary of $Q$), $\red(\V,{\it d\mbox{-}hw_k(Q,\DB)})\neq\emptyset$ entails
    \emph{$Q^\onDB \neq \emptyset$}.

\item[\em (2)] There is a subquery $Q' \homEquiv Q$ having generalized hypertree-width at most $k$.

\item[\em (3)] There is a core $Q'$ of $Q$ having generalized hypertree-width at most $k$.
\end{enumerate}
\end{theorem}
\begin{proof}
It suffices to show that $(3)$ is equivalent to $(3')$ below. Then, the theorem follows from Corollary~\ref{cor:DecisionDM}.

$(3')$ \ There is a core $Q'$ of $Q$ for which $(\HGUnoSet,\HGDue)$ has a tree projection, with $\V={\it v\mbox{-}hw_k(Q)}$.

\medskip

Let $\V'={\it v\mbox{-}hw_k(Q')}$. Note that $(3)$ is equivalent to say that  $(\HGUnoSet,\HG_{\V'})$ has a tree projection, which entails
$(3')$, because $Q'$ is a subquery of $Q$ and thus $\HG_{\V'}\leq \HGDue$.

It remains to show that $(3') \Rightarrow (3)$. Assume by contradiction that this is not the case, hence there is a core $Q'$ of $Q$ for which
$(\HGUnoSet,\HGDue)$ has a tree projection $\HG_a$, but every core of $Q$ has generalized hypertree width greater than $k$. In particular, this
must hold for $Q'$, too. It follows that there exists some hyperedge $h$ that belongs to $\HG_a$ and thus is covered by some hyperedge of
$\HGDue$, but it is not covered by any hyperedge of $\HG_{\V'}$, where $\V'={\it v\mbox{-}hw_k(Q')}$. That is, there is no view $w$ in $\V'$
such that $h\subseteq \vars(w)$.
Recall that, by definition of function ${\it v\mbox{-}hw_k}$, views in $\V$ (resp., $\V'$) contain the union of variables from all possible
sets of at most $k$ atoms occurring in $Q$ (resp., $Q'$). It follows that there is some atom $a\in\atoms(Q)$ with $\bar X=\vars(a)\cap
h\neq\emptyset$ which does not belong to $Q'$ and whose role in $w$ cannot be played by any other atom in $Q'$. Formally, there is no atom
$a'\in\atoms(Q')$ such that $\bar X' \subseteq \vars(a')$, where $\bar X'= \bar X \cap \vars(Q')$. In fact, note that $\bar X'$ are the only
possible crucial variables: further variables of $w$ not occurring in $Q'$ are never necessary in any tree projection of $\HGUnoSet$ (w.r.t.
any hypergraph), as it is known and easy to see that, if a tree projection exists, there always exists one that uses only nodes from
$\HGUnoSet$~\cite{GS08}.

However, $Q'$ is a core of $Q$, and thus it is a retract, which means that there must exist a homomorphism $f$ from $Q$ to $Q'$ where $f(X)=X$,
for each term $X$ occurring in $Q'$. Therefore, the atom $a$ should be mapped to some atom $a'\in\atoms(Q')$ that contains all variables $f(X)$
for each $X\in\vars(a)$. In particular, this entails that all variables in $\bar X'$ occur in $a'$, because $f$ is the identity mapping over
them. Contradiction.~\hfill~$\Box$
\end{proof}

For the special case of $k=1$, the above result provides the precise relationship between local consistency and acyclic queries, extending the
classical result given in~\cite{BFMY83} for simple queries (in fact, for acyclic schemas). Recall that, for the acyclic method, the set of
views ${\it v\mbox{-}acyc(Q)}$ is just the set of query views $\views(Q)$, and their database relations in ${\it d\mbox{-}acyc(Q,\DB)}$ are
just the copies of their corresponding query atoms.

\begin{theorem}\label{thm:DecisionHW2}
For any conjunctive query $Q$, the following are equivalent:

\vspace{-1mm}
\begin{enumerate}

\item[\em (1)] For every database $\DB$ (over the vocabulary of $Q$), $\red({\it v\mbox{-}acyc(Q)},{\it
    d\mbox{-}acyc(Q,\DB)})\neq\emptyset$ entails \emph{$Q^\onDB \neq \emptyset$}.

\item[\em (2)] There is an acyclic subquery $Q' \homEquiv Q$.

\item[\em (3)] $Q$ has an acyclic core.
\end{enumerate}
\end{theorem}

\section{Larger Islands of Tractability}\label{SECGAMES}

In this section, we investigate a tractable variant of the notion of tree projections that allows us to identify new islands of tractability
for query answering, constraint satisfaction problems, and further problems that are easy on tree-like structures.
Indeed we argue that, in practical database applications, ``blind'' local-consistency enforcing procedures are hardly used, because the number
of semijoin operations to be performed depends on the database size and may be very high. On the other hand, if one is able to compute a tree
projection, then the views to be processed will be only those involved in the tree projection, and the number of semijoin operations to be
performed will be at most the number of these views (hence, independent of the database).

The new notion is based on the game characterization of tree projections proposed in~\cite{GS08}.
To formalize our results, we need to introduce some additional definitions and notations, which will be intensively used in the following.

Assume that a hypergraph $\HG$ is given. Let $V$, $W$, and $\{X,Y\}$ be sets of nodes. Then, $X$ is said \adj{V}\ (in $\HG$) to $Y$ if there
exists a hyperedge $h\in \edges(\HG)$ such that $\{X,Y\}\subseteq (h -V)$. A \spath{V}\ from $X$ to $Y$ is a sequence $X=X_0,\ldots,X_\ell=Y$
of nodes such that $X_{i}$ is \adj{V}\ to $X_{i+1}$, for each $i\in [0...\ell\mbox{-}1]$. We say that $X$ \touches{V} $Y$ if $X$ is
\adj{\emptyset} to $Z\in \nodes(\HG)$, and there is a \spath{V} from $Z$ to $Y$; similarly, $X$ \touches{V} the set $W$ if $X$ \touches{V} some
node $Y\in W$.
We say that $W$ is \connected{V}\ if $\forall X,Y\in W$ there is a \spath{V}\ from $X$ to $Y$. A \component{V} (of $\HG$) is a maximal
\connected{V}\ non-empty set of nodes $W\subseteq (\nodes(\HG)-V)$. For any \component{V}\ $C$, let $\edges(C) = \{ h\in \edges(\HG)\;|\; h\cap
C\neq\emptyset\}$, and for a set of hyperedges $H\subseteq \edges(\HG)$, let $\nodes(H)$ denote the set of nodes occurring in $H$, that is
$\nodes(H)=\bigcup_{h\in H} h$. For any component $C$ of $\HG$, we denote by $\F(C,\HG)$  the \emph{frontier} of $C$ (in $\HG$), i.e., the set
$\nodes(\edges(C))$.\footnote{The choice of the term ``frontier'' to name the union of a component with its outer border is due to the role
that this notion plays in the hypergraph game described in the subsequent section.}
Moreover, $\partial (C,\HG)$ denote the {\em border} of $C$ (in $\HG$), i.e., the set $\F(C,\HG)\setminus C$. Note that $C_1\subseteq C_2$
entails $\F(C_1,\HG)\subseteq \F(C_2,\HG)$.

In the following sections, given any pair of hypergraphs $(\HG_1,\HG_2)$ and a set of nodes $C\subseteq \HG_1$, we write for short $\F(C)$ and
$\partial C$ to denote $\F(C,\HG_1)$ and  $\partial (C,\HG_1)$, respectively.

\subsection{Game-Theoretic Characterization}\label{GTC}

The \emph{Robber and Captain} game is played on a pair of hypergraphs $(\HG_1,\HG_2)$  by a Robber and a Captain controlling some squads of
cops, in charge of the surveillance of a number of strategic targets. The Robber stands on a node and can run at great speed along the edges of
$\HG_1$. However, (s)he is not permitted to run trough a node that is controlled by a cop. Each move of the Captain involves one squad of cops,
which is encoded as a hyperedge $h\in \edges(\HG_2)$. The Captain may ask some cops in the squad $h$ to run in action, as long as they occupy
nodes that are currently reachable by the Robber, thereby blocking an escape path for the Robber. Thus, ``second-lines'' cops cannot be
activated by the Captain. Note that the Robber is fast and may see cops that are entering in action. Therefore, while cops move, the Robber may
run trough those positions that are left by cops or not yet occupied. The goal of the Captain is to place a cop on the node occupied by the
Robber, while the Robber tries to avoid her/his capture.

\begin{definition}  Let $\HG_1$ and $\HG_2$ be two hypergraphs.
The Robber and Captain game on $(\HG_1,\HG_2)$ is formalized as follows.
A \emph{position} for the Captain is a pair $(h,M)$ where $h$ is a hyperedge of $\HG_2$ and $M\subseteq h$.
A \emph{configuration} is a triple $(h,M,C)$, where $(h,M)$ is a position for the Captain, and $C$ is the \component{M} where the Robber
stands.\footnote{It is easy to see that in such games, being the robber arbitrarily fast, what matters is not the precise node where the robber
stands, but just the \component{M} where (s)he is free to move.}
 The initial configuration is $(\emptyset,\emptyset,\nodes(\HG_1))$.

A \emph{strategy} $\sigma$ is a function that encodes the \emph{moves} of the Captain. Its domain includes the initial configuration. For each
configuration $v_p=(h_p,M_p,C_p)$ in the domain of~$\sigma$, $\sigma(v_p)=(h_{r},M_{r})$, with $M_{r}\subseteq h_r\cap\F(C_p)$, is the novel
position for the Captain. After this move, the Robber can select any {\em \iecomponent{v_p,M_{r}}}, i.e., any \component{M_{r}} $C_{r}$ such
that $C_{p}\cup C_{r}$ is \connected{M_{p}\cap M_{r}}. If there is no \ecomponent{v_p,M_{r}}, then $(h_{r},M_{r},\emptyset)$ is said a
\emph{capture} configuration induced by $\sigma$. The move of the Captain is \emph{monotone} if, for each \ecomponent{v_p,M_{r}} $C_{r}$,
$C_{r}\subseteq C_p$. The domain of $\sigma$ includes the configuration $(h_{r},M_{r},C_{r})$, for each \ecomponent{v_p,M_{r}} $C_{r}$. No
other configuration is in the domain of $\sigma$.
The strategy $\sigma$ is monotone if it encodes only monotone moves over the configurations in its domain.

A strategy $\sigma$ can be represented as a directed graph $G(\sigma)=(N,A)$, called {\em strategy graph}, as follows. The set $N$ of nodes is
the set of all configurations in the domain of $\sigma$ plus all capture configurations induced by $\sigma$. If $v_p=(h_p,M_p,C_p)$ is a
configuration and $\sigma(v_p)=(h_{r},M_{r})$, then $A$ contains an arc from $v_p$ to $(h_{r},M_{r},C_{r})$ for each \ecomponent{v_p,M_{r}}
$C_{r}$, and to $(h_{r},M_{r},\emptyset)$ if there is no \ecomponent{v_p,M_{r}}.
We say that $\sigma$ is a winning strategy (for the Captain) if $G(\sigma)$ is acyclic. Otherwise, i.e., if $G(\sigma)$ contains a cycle, then
the Robber can avoid her/his capture forever.~\hfill~$\Box$
\end{definition}

\begin{figure}[t]
  \centering
  \includegraphics[width=0.98\textwidth]{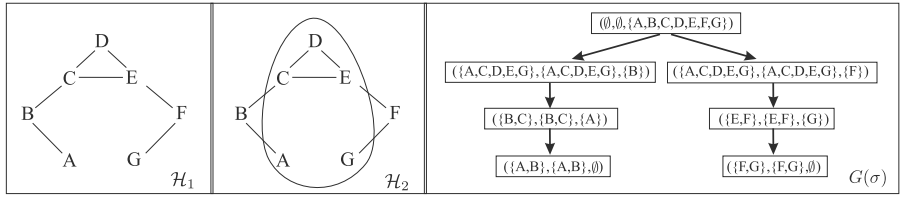}
  \caption{The hypergraphs $\HG_1$ and $\HG_2$, plus the graph $G(\sigma)$ in Example~\ref{ex:construction}.}\label{fig:greedy-bis}
\end{figure}

\begin{example}\label{ex:construction}\em
Consider the two hypergraphs $\HG_1$ and $\HG_2$ reported in Figure~\ref{fig:greedy-bis}, together with the strategy graph $G(\sigma)$. The
graph encodes a winning strategy $\sigma$ for the Captain. From the initial configuration $(\emptyset,\emptyset,\nodes(\HG_1))$, the Captain
activates all the cops in the hyperedge $\{A,C,D,E,G\}$, so that the Robber has two available options, i.e., $\{B\}$ and $\{F\}$. In the former
(resp., latter) case, the Captain activates all the cops in the hyperedge $\{B,C\}$ (resp., $\{E,F\}$), so that the Robber has necessarily to
occupy the node $A$ (resp., $G$). Finally, the Captain activates the cops in $\{A,B\}$ (resp., $\{F,G\}$) and captures the Robber.
Note that the strategy $\sigma$ is non-monotone, because the Robber is allowed to return on $A$ and $G$, after that these nodes have been
previously occupied by the Captain in the first move. \hfill $\lhd$
\end{example}

In the above example, the hyperedge $\{A,C,D,E,G\}$ of $\HG_2$ ``absorbs'' the cycle in $\HG_1$, so that it is easily seen that there is a tree
projection $\HG_a$ of $\HG_1$ w.r.t.~$\HG_2$ (see Figure~\ref{fig:strategy-new-bis}). The fact that on this pair the Captain has a winning
strategy is not by chance.

\begin{figure}[t]
  \centering
  \includegraphics[width=0.99\textwidth]{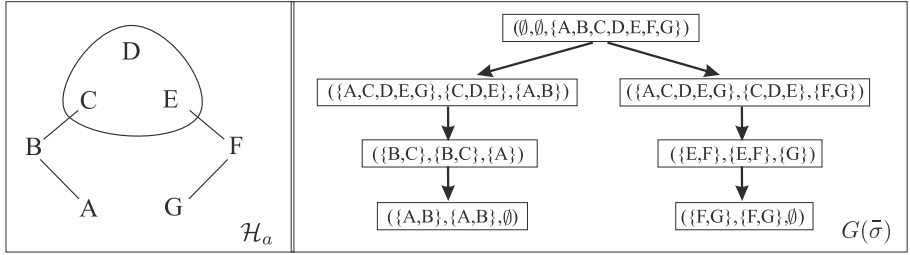}
  \caption{A tree projection $\HG_a$ for the pair in Example~\ref{ex:construction}, plus the graph $G(\bar \sigma)$.}\label{fig:strategy-new-bis}
\end{figure}

\begin{theorem}[\cite{GS08}]
There is a tree projection of $\HG_1$ w.r.t.~$\HG_2$ if, and and only if, there is a winning strategy in the Captain and Robber game played on
$(\HG_1,\HG_2)$.
\end{theorem}

Recall that the winning strategy in Example~\ref{ex:construction} is not monotone. However, an important property of this game is that there is
no incentive for the Captain to play a strategy that is not monotone.

\begin{theorem}[cf. \cite{GS08}]\label{thm:costruzione}
In the Captain and Robber game played on the pair $(\HG_1,\HG_2)$, a winning strategy exists if, and only if, a monotone winning strategy
exists.

Moreover, from any monotone winning strategy, a tree projection of $\HG_1$ w.r.t.~$\HG_2$ can be computed in polynomial time.
\end{theorem}

\begin{example}\label{ex:construction-bis}\em
Consider again the setting of Example~\ref{ex:construction}, and the strategy graph $G(\bar \sigma)$ shown in
Figure~\ref{fig:strategy-new-bis}. Note that the strategy $\bar \sigma$ is monotone, and in fact the moves of the Captain one-to-one correspond
with the hyperedges in the tree projection $\HG_a$. \hfill $\lhd$
\end{example}

The crucial properties to establish Theorem~\ref{thm:costruzione} are next recalled, as they will be useful in our subsequent analysis too. Let
$\sigma$ be a strategy, and let $v_p=(h_p,M_p,C_p)$ and $v_r=(h_r,M_r,C_r)$ be two configurations in its domain such that
$\sigma(v_p)=(h_r,M_r)$ and $C_r$ is a \ecomponent{v_p,M_{r}}. Let $\sigma(v_r)=(h_s,M_s)$ and define $\E((M_r,C_r),M_s)=M_r\cap \F(C_r)
\setminus M_{s}$ (which is equivalent to $\partial C_r\setminus M_{s}$ because $C_r$ is an \component{M_r}) as the \emph{escape-door} of the
Robber in $v_r$ when attacked with $M_s$. From~\cite{GS08}, a move is monotone if, and only if, such an escape door is empty; in particular,
$\sigma(v_r)$ is non-monotone if (and only if) $\E((M_r,C_r),M_s)\neq\emptyset$.

Let $M_r'=M_r\setminus \E((M_r,C_r),M_s)$, let $C_r'$ be the \component{M_r'} with $C_r\cup \E((M_r,C_r),M_s)\subseteq C_r'$, which exists
since $\E((M_r,C_r),M_s)\subseteq \F(C_r)$ and $M_r'\subseteq M_r$, and let $v_r'=(h_r,M_r',C_r')$. Finally, consider the following strategy
$\sigma'$:
\begin{equation}
\label{def:transformation}
\sigma'(h,M,C)= \left\{
\begin{array}{ll}

(h_r,M_r') & \mbox{ if }(h,M,C)=(h_p,M_p,C_p)\\
\sigma(h,M,C) & \mbox{ otherwise.}\\
\end{array}
\right.
\end{equation}

For such a state of the game, a number of technical properties have been proved in~\cite{GS08}. We summarize them in the following lemma.

\begin{lemma}[\cite{GS08}]\label{lem:construction}
The following properties hold:

\begin{enumerate}
\item[\em (1)] $\E((M'_r,C'_r),M_s)=\emptyset$.
\item[\em (2)] For each \iecomponent{v_p,M_r} $C$, either $C\subseteq C_r'$ or $C$ is a \iecomponent{v_p,M_r'}.
\item[\em (3)] For each \iecomponent{v_p,M_r'} $C'\neq C_r'$, $C'$ is a \iecomponent{v_p,M_r}.
\item[\em (4)] A set $C$ is a \iecomponent{v_r,M_s} if, and only if, it is a \iecomponent{v_r',M_s}.
\item[\em (5)] If $\sigma$ is a winning strategy, then $\sigma'$ is a winning strategy too.
\end{enumerate}
\end{lemma}

\subsection{Greedy Strategies}

Since winning strategies correspond to tree projections, there is no efficient algorithm for their computation. Indeed, just recall that
deciding the existence of a tree projection is not feasible in polynomial time, unless $\Pol=\NP$~\cite{GMS07}. Our goal is then to focus on
certain ``greedy'' strategies that are easy to compute. Intuitively, in greedy strategies it is required that \emph{all} cops available at the
current squad $h_p$ and reachable by the Robber enter in action. If all of them are in action, then a new squad $h_r$ is selected, again
requiring that all the active cops, i.e., those in the frontier, enter in action.

\begin{definition}\label{def:greedy}
On the Captain and Robber game played on $(\HG_1,\HG_2)$, a strategy $\sigma$ is {\em greedy} if, for any configuration $v_p=(h_p,M_p,C_p)$ in
the domain of $\sigma$, the next position $\sigma(v_p)=(h_{r},M_{r})$ is such that $M_{r}=h_{r}\cap \F(C_p)$, where $h_{r}=h_p$ if $h_p\cap
C_p\neq\emptyset$, and $h_{r}$ is any squad in $\edges(\HG_2)$ if $h_p\cap C_p=\emptyset$. \hfill $\Box$
\end{definition}

Given such a greedy way to select cops at each step, observe that the former case ($h_p\cap C_p\neq\emptyset$) may only occur if the Robber is
able to come back to some position previously controlled by the Captain. Greedy winning strategies are indeed non-monotone in general, and for
some pair of hypergraphs it is possible that there is no monotone winning greedy strategy, although  monotone winning strategies (non-greedy)
exist.

\begin{example}\label{ex:greedy-strat}\em
Consider again the hypergraphs $\HG_1$ and $\HG_2$ shown in Figure~\ref{fig:greedy-bis}, and recall that the strategy graph of a monotone
winning strategy $\bar \sigma$ is depicted in Figure~\ref{fig:strategy-new-bis}. However, there is no monotone greedy strategy in this case.
Indeed, if at the beginning of the game the Captain asks the squad $\{A,C,D,E,G\}$ to enter in action and the Robber goes on $B$, then in the
next move the Robber is forced to lose the control on $A$ in order to move on $\{C,B\}$ and eventually win via $\{B,A\}$---see again
Figure~\ref{fig:greedy-bis}. On the other hand, if the attack of the Captain starts on either side, say on the left branch, the Captain has
then to attack the component that includes the triangle and the other branch. At this point, the only available greedy choice is use the big
squad and hence to employ cops $\{C,D,E,G\}$. However, as in the previous case, $G$ will be later (necessarily) left free to the Robber, in
order to win the game. \hfill $\lhd$
\end{example}

We now show that, differently from arbitrary strategies, the existence of greedy winning strategies can be decided in polynomial time. To
establish the result, a useful technical property is that greedy strategies can only involve a polynomial number of configurations. Let us
denote by $\texttt{MaxGreedyStrat}(\HG_1,\HG_2)$ the maximum domain cardinality over any greedy strategy in the Robber and Captain game on a
pair $(\HG_1,\HG_2)$.

\begin{lemma}\label{lem:maxconf}
Let $(\HG_1,\HG_2)$ be a pair of hypergraphs. Then, $\texttt{MaxGreedyStrat}(\HG_1,\HG_2)$ is at most $|\edges(\HG_2)|\times|\nodes(\HG_1)|
(|\edges(\HG_2)|\times|\nodes(\HG_1)|+1) +1$.
\end{lemma}

\begin{proof}
Let $\sigma$ be a greedy strategy, and let $v_p=(h_p,M_p,C_p)$ be a configuration in its domain. Note that the only configuration where
$h_p=M_p=\emptyset$ is the starting configuration $(\emptyset,\emptyset,\nodes(\HG_1))$, which is taken into account by the final ``$+1$'' in
the statement. Therefore, we next assume $M_p\neq\emptyset$.

Consider the case where $h_p\cap C_p=\emptyset$. In this case, a new squad $h_{r}\in\edges(\HG_2)$ is chosen by the Captain according to
$\sigma$. Since $C_p$ is an \component{M_p} and thus $\partial C_p\subseteq M_p\subseteq h_p$, we get that this case occurs only if $C_p$ is
actually an \component{h_p}, too. Such a component is uniquely identified by any pair of the form $(h_p,X_p)$ such that $X_p\in \nodes(\HG_1)$
is a representative of the component (e.g., the node in $C_p$ having the smallest position according to any fixed ordering over the nodes).
It follows that the new set of cops $M_{r}=h_{r}\cap \F(C_i)$ is uniquely determined by $h_{r}$ and $C_p$ and thus may be identified through a
triple $(h_{r},h_p,X_p)$. Thus, the maximum number of such sets $M_{r}$ of cops is $|\edges(\HG_2)|^2\times |\nodes(\HG_1)|$. Moreover, the
possible configurations $(h_{r},M_{r},C_{r})$ following $(h_p,M_p,C_p)$ in the game where the Captain plays according to $\sigma$ are
identified by quadruples of the form $(h_{r},h_p,X_p,X_{r})$, where $h_{r}$ is used both to identify itself and to determine the set $M_{r}$
together with $h_p$ and $X_p$, and where $X_{r}$ is a representative of the \component{M_{r}}.
In fact, if there is no \ecomponent{v_p,M_p}, then $X_{r}$ is a distinguished element not in $\nodes(\HG_1)$ (or some element in $M_p$ occupied
by some cop) meaning that the only configuration following $(h_p,M_p,C_p)$ is $(h_{r},M_{r},\emptyset)$ where the Robber is captured.
Overall, the maximum number of such configurations is $|\edges(\HG_2)|^2 \times |\nodes(\HG_1)|^2$.

Finally, consider the case where $h_p\cap C_p\neq\emptyset$. In this case, $M_{r}=h_p\cap \F(C_p)$. Since $C_p$ is an \component{M_p},
$\partial C_p\subseteq M_p\subseteq h_p$. It follows that the new nodes from $\F(C_p)$ to be included in $M_{r}$ belong to $C_p$, that is, we
may also write $M_{r}=M_p\cup (h_p\cap C_p)$. Note that no configuration of the game following this one can be of this type. Indeed, every
\component{M_{r}} $C_{r}$ where the Robber may go from $C_p$ will be a subset of $C_p$ (because $\partial C_p\subseteq M_p\subseteq
M_{r}\subseteq h_p$), and will have intersections with $h_p$. As a further consequence, such a $C_{r}$ must be an \component{h_p}. By
contradiction, if there is some node $X_p\in C_{r}\subseteq C_p$ that is \connected{M_{r}} to some $X_{r}$ in $h_p\setminus M_{r}$, then $X_p$
is also \connected{M_p} to $X_{r}$. However, this is impossible because $X_p$ is also in $C_p$ and hence $X_{r}$ would be in $C_p$, too, and
hence in $h_p\cap C_p$ and in $M_{r}$, by construction. Therefore, the possible configurations $(h_p,M_{r},C_{r})$ following $(h_p,M_p,C_p)$ in
the game where the Captain plays according to $\sigma$ are identified by pairs of the form $(h_p,X_p)$, where $X_p\in \nodes(\HG_1)$ is the
representative of the \component{h_p} $C_{r}$ (and where $M_{r}$ is computed from them).
As above, if there is no \ecomponent{v_p,M_p}, then $X_p$ is a distinguished element witnessing that the configuration is a capture
configuration of the form $(h_p,M_{r},\emptyset)$.
Overall, the maximum number of such configurations is $|\edges(\HG_2)|\times |\nodes(\HG_1)|$.~\hfill~$\Box$
\end{proof}

\begin{figure}[t]
  \vspace{2mm}
  \centering
  \fbox{\parbox{0.9\textwidth}{\small
  \vspace{-1mm}\begin{itemize}
  \item[]\hspace{-3mm}\emph{Boolean function} \textsc{GreedyWinningStrategy}$(h_p,M_p,C_p,i)$;
  \vspace{0mm}\item[] 
  /$\ast$  $(h_p,M_p,C_p)$ is an extended configuration over $(\HG_1,\HG_2)$,\\ \ \ \ \ \ $\hspace{10mm} i\geq 0$ is a natural number $\ast$/
  \
  \\ \vspace{-2mm}\hrule\ \\

\vspace{-4mm}
  \item[1)] \textbf{if} $i>\texttt{MaxGreedyStrat}(\HG_1,\HG_2)$, then \textbf{return} \textsc{False};
  \item[2)] \textbf{if} $h_p \cap C_p \neq \emptyset$, then  \textbf{let} $h_{r}=h_p$;
  \item[] \textbf{else}  \textbf{guess} a hyperedge $h_{r}\in\edges(\HG_2)$;
  \item[3)] \textbf{let} $M_{r}= h_{r} \cap \F(C_p) $;
  \item[4)] \textbf{for each} \ecomponent{(h_p,M_p,C_p),M_{r}} $C_{r}$ \textbf{do}
  \item[] \quad \ \textbf{if} \emph{not} \textsc{GreedyWinningStrategy}$(h_{r},M_{r},C_{r},i+1)$, then \textbf{return} \textsc{False};
  \item[5)] \textbf{return} \textsc{True};

 \vspace{-1mm}\end{itemize}}}
 \caption{\small \textsc{GreedyWinningStrategy}.}\label{fig:Algoritmo1}
\end{figure}

To see that the existence of a winning greedy strategy is decidable in polynomial time, consider the \textsc{GreedyWinningStrategy} algorithm
illustrated in Figure~\ref{fig:Algoritmo1}, which receives as input a configuration $(h_p,M_p,C_p)$ for the Robber and Captain game, plus a
``level'' $i$. Note that this algorithm is a high-level specification of an alternating Turing machine, say $\mathcal{M}_G$~\cite{john-90}.
After the first step, where we check that the number of recursive calls has not exceeded the number of all distinct configurations, the
algorithm suddenly evidences its \emph{non-deterministic} nature. Indeed, it guesses a hyperedge $h_{r}$ corresponding to the next move of the
Captain (existential step of $\mathcal{M}_G$). Eventually, it returns \textsc{True} if, and only if, the recursive calls
\textsc{GreedyWinningStrategy}$(h_{r},M_{r},C_{r},i+1)$ with $M_{r}= h_{r} \cap \F(C_p)$ succeed on each \ecomponent{(h_p,M_p,C_p),M_{r}}
(universal step of $\mathcal{M}_G$).

\begin{theorem}\label{thm:GreedyExistence}
Deciding the existence of a greedy winning strategy in the Robber and Captain game is feasible in polynomial time.
\end{theorem}

\begin{proof}
Let $(\HG_1,\HG_2)$ be a pair of hypergraphs, and consider the execution of the Boolean function $\textsc{GreedyWinningStrategy}$ on input the
starting configuration $(\emptyset, \emptyset, \nodes(\HG_1),0)$. Due to its non-deterministic nature, it is easily seen that, by getting rid
of step {(1)}, it returns \textsc{True} if, and only if, the Captain has a greedy winning strategy in the game played on $(\HG_1,\HG_2)$ (which
we assume to be ``visible'' by the function at a every call, to avoid a longer signature).
Moreover, we claim that the check performed at step~{(1)} cannot lead to a wrong \textsc{False} output. Indeed, just observe that the number of
recursive calls is bounded by the number of all distinct configurations, which is $\texttt{MaxGreedyStrat}(\HG_1,\HG_2)$ at most, by
Lemma~\ref{lem:maxconf}. Therefore, if the recursion level $i$ exceeds this threshold, then we can safely answer \textsc{False}.

Let us now focus on the running time. We have already observed that $\textsc{GreedyWinningStrategy}$ may be implemented on an alternating
Turing machine $\mathcal{M}_G$, whose existential steps correspond to the guess statements at step~2, while universal steps are used for
checking that the conditions at step~4 are satisfied by all the relevant components. In addition, by indexing the various data structures and
by referring each component via one point contained in it (selected through any fixed criterium), the machine can be implemented to use
logarithmic many bits on its worktape. For instance, recall from the proof of Lemma~\ref{lem:maxconf} that every configuration is identified by
at most four elements of the form $(h_p,h_{r},X_p,X_{r})$ with $h_p,h_{r}\in\edges(\HG_2)$ and $X_p,X_{r}\in\nodes(\HG_1)$. Therefore, any
configuration may be encoded by (at most) four indexes whose maximum size is $\log \max \{|\edges(\HG_2)|,|\nodes(\HG_1)|\}$. Moreover, the
check at step~{(1)} ensures that the length of each branch of the computation tree of $\mathcal{M}_G$ is finite, and actually bounded by a
polynomial in the size of the input. For the sake of completeness, observe that all subtasks in the function, such as computing connected
components and the like, are easily implementable in nondeterministic logspace, so that such tasks just correspond to further
(polynomially-bounded) branches of the computation tree of $\mathcal{M}_G$. Thus, $\textsc{GreedyWinningStrategy}$ may be implemented in a
\emph{log-space} alternating Turing machine, which  immediately entails the result, because Alternating Logspace is equal to Polynomial
Time~\cite{chan-etal-81}.~\hfill~$\Box$
\end{proof}

It is well known that an alternating Turing machine $\mathcal{M}_G$ can be simulated by a standard machine in polynomial time. First, compute
the polynomially-many possible instant descriptions (IDs) of the machine, and build a graph representing the possible connections between any
pair of IDs, according to its transition relation. Then, evaluate this graph along some topological ordering as follows. Mark all IDs without
outcoming arcs associated with final accepting states; then mark all IDs associated with  existential states having a marked successor, or
associated with universal states, and whose successors are all marked. Then, the machine $\mathcal{M}_G$ accepts its input if, and only if, the
starting ID is marked. Moreover, the subgraph induced by the marked nodes encodes its accepting computations.

Moreover, from such a marked graph it is straightforward to compute the strategy graph of a greedy winning strategy, because IDs associated
with (children of) existential states encode the possible choices of the Captain.\footnote{For the sake of completeness note that, by using
these ideas, one might also provide  a direct dynamic programming algorithm to compute a strategy graph by using a bipartite graph representing
all possible configurations and positions of the Robber and Captain game. However, we find the non-deterministic function
$\textsc{GreedyWinningStrategy}$ more elegant and easy to present.} Just visit the graph starting from the initial configuration, but for each
ID associated with an existential state, select one child to be visited arbitrarily (all choices are marked and hence accepting).

\begin{corollary}\label{cor:computeGreedyStrat}
The strategy graph of a greedy winning strategy (if any) in the Robber and Captain game is computable in polynomial time.
\end{corollary}

\subsection{Greedy Tree Projections and Larger Islands of Tractability}

From the previous sections (see Theorem~\ref{thm:costruzione} and Example~\ref{ex:greedy-strat}), we know that monotone winning strategies for
the Captain in the game over $(\HG_1,\HG_2)$ are associated with tree projections of $\HG_1$ w.r.t.~$\HG_2$, and that in some cases it is
possible that there is no monotone winning greedy strategy, although  monotone winning strategies (non-greedy) exist. In this section, we show
that from {\em any} (possibly non-monotone) greedy winning strategy a tree projection can be still computed in polynomial time. The key fact
here is that any non-monotone greedy strategy can be converted into a monotone one, though not a greedy one in general.

To show the result, it is useful to consider a special form of strategies that we call {\em nice} (for they remind the notion of nice tree
decompositions of graphs), where at every configuration the Captain first removes those cops that are no longer in the frontier.

Formally, $\sigma$ is a {\em nice} strategy if $\sigma(h_p,M_p,C_p)= (h_p,\partial C_p)$, whenever $\partial C_p \subset M_p$. Because such
inactive cops play no role in the Robber and Captain game, a winning nice strategy exists if (and only if) there exists a winning strategy, and
the same holds for greedy strategies. Just note that restricting the cops to the border of $C_p$ is a legal choice in greedy strategies (it
corresponds to the selection of the same squad $h_{r}=h_p$ before attacking the robber in the component $C_p$ with some further squad).
Clearly enough, such a nice strategy can be computed in polynomial time from any given strategy. Also, if desired, the above polynomial time
algorithm for computing a greedy strategy may be easily adapted to compute directly a winning nice greedy strategy (if any).

\begin{figure}[t]
  \centering
  \includegraphics[width=0.6\textwidth]{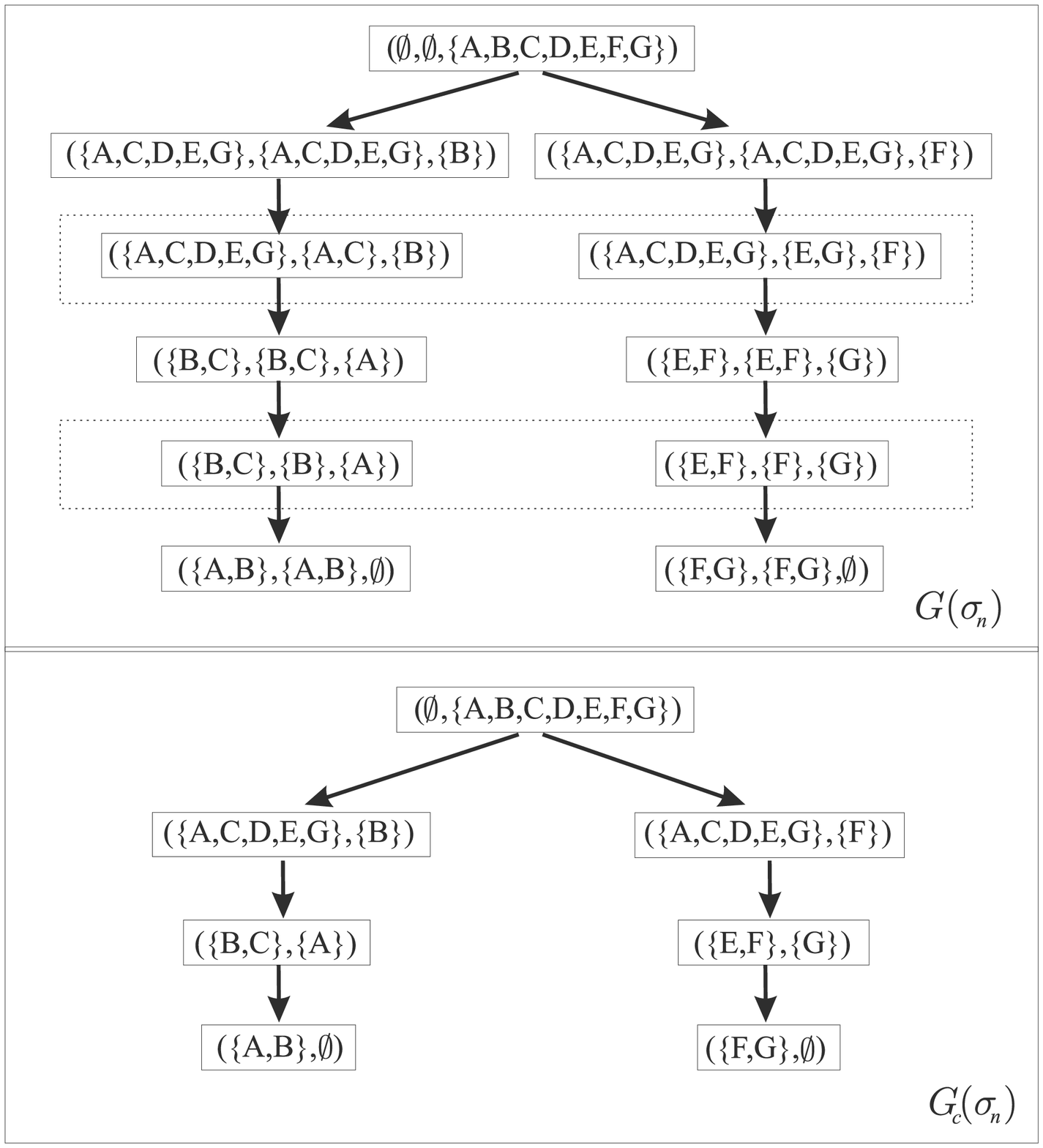}
  \caption{The strategy and component graphs for the nice strategy $\sigma_n$ in Example~\ref{ex:nice}.}\label{fig:nice}
\end{figure}

\begin{example}\label{ex:nice}\em
Consider again the setting discussed in Example~\ref{ex:construction} and illustrated in Figure~\ref{fig:greedy-bis}. Note that the strategy
$\sigma$ is not nice. Indeed, Figure~\ref{fig:nice} reports the strategy graph associated with a strategy $\sigma_n$ that is nice and that is
obtained from $\sigma$ by just explicitly adding the configurations where the Captain has to remove the cops that are no longer in the
frontier. \hfill $\lhd$
\end{example}

The reason for introducing these nice strategies is that they admit a more compact representation. First, given any configuration
$(h_p,M_p,C_p)$ and a Captain's choice $M_{r}$, the \ecomponents{(h_p,M_p,C_p),M_{r}} for the Robber are actually determined by $C_p$ and
$M_{r}$ only, because $\partial C_p$ is computable from $C_p$. Therefore, we use hereafter the simplified notation \ecomponent{C_p,M_{r}} to
refer to this set of \components{M_{r}}.
Moreover, in place of the strategy graph, we can use a \emph{component graph}, defined as follows.

\begin{definition}
Let $(\HG_1,\HG_2)$ be a pair of hypergraphs.
Let $G=(N,A)$ be a directed graph whose nodes are pairs of the form $(h_p,C_p)$, where $h_p\in\edges(\HG_2)$, and $C_p$ is either the emptyset
or a \component{\partial C_p} of $\HG_1$ such that $\partial C_p \subseteq h_p$. Then, we say that $G$ is a {\em component graph} if it meets
the following conditions:
\begin{enumerate}
\item[ (1)] There is a \emph{root} node $(\emptyset,\nodes(\HG_1))\in N$ that is the only node without incoming arcs.

\item[ (2)] Each node $(h_p,C_p)\in N$, with $C_p\neq \emptyset$, has outgoing arcs to $m\geq 0$ nodes $(h_{r},\bar C_1),\dots ,(h_{r},\bar
    C_m)$ such that, if $M_{r}$ is the set $\bigcup_{j=1}^m \partial \bar C_j \cup (C_p\setminus \bigcup_{j=1}^m \bar C_j)$, it holds that
    $M_{r}\subseteq h_{r}$ and the \ecomponents{C_p,M_{r}} are the components $\bar C_1,...,\bar C_m$.

\item[ (3)] Each node $(h_p,C_p)\in N$ has an outgoing arc to $(h_{r},\emptyset)$ if $C_p\subseteq h_{r}$. \hfill $\Box$
\end{enumerate}
\end{definition}

Note that every nice strategy $\sigma$ is encoded by the component graph $G_c(\sigma)=(N,A)$ defined as follows. There is a node $(h_p,C_p)$
(resp., $(h_p,\emptyset)$) in $N$ if there is a configuration $(h_p,\partial C_p,C_p)$ in the domain of $\sigma$ (resp., a capture
configuration $(h_p,\partial C_p,\emptyset$) induced by $\sigma$).
There is an arc in $A$ from a node $(h_p,C_p)$ to a node $(h_r,C_r)$ if there is an arc from $(h_p,M_p,C_p)$ to $(h_r,M_r,C_r)$ in the strategy
graph $G(\sigma)$. No more nodes and arcs occur in $N$ and $A$, respectively.
For instance, the graph depicted on the bottom part of Figure~\ref{fig:greedy-bis} is the component graph associated with the nice strategy
$\sigma_n$ of Example~\ref{ex:nice}.

Conversely, any component graph $G$ encodes a nice strategy $\sigma_{G}$, via the following procedure. Associate the root
$(\emptyset,\nodes(\HG_1))$ with the initial configuration $(\emptyset,\emptyset,\nodes(\HG_1))$. Inductively, assume that a node $(h_p,C_p)$
is associated with a configuration $(h_p,M_p,C_p)$, and that $(h_{r},\bar C_1),...,(h_{r},\bar C_m)$ are the labels of the nodes having an
incoming arc from $(h_p,M_p,C_p)$. Let $M_{r}= \bigcup_{j=1}^m \partial \bar C_j \cup (C_p\setminus \bigcup_{j=1}^m \bar C_j)$, with
$M_{r}\subseteq h_{r}$. Then, define $\sigma_{G}(h_p,M_p,C_p)=(h_{r},M_{r})$, and define $\sigma_G(h_{r},M_{r},\bar C_j)= (h_{r},\partial \bar
C_j)$, with $j\in \{1,...,m\}$, in the case where $\partial \bar C_j \subset M_{r}$.

\begin{theorem}\label{thm:greedy}
A tree projection of $\HG_1$ w.r.t.~$\HG_2$ can be computed in polynomial time if the Captain has a greedy winning strategy on $(\HG_1,\HG_2)$.
\end{theorem}
\begin{proof}
By Theorem~\ref{thm:GreedyExistence}, we can decide in polynomial time whether a winning greedy strategy for the Captain in the game played on
$(\HG_1,\HG_2)$ exists or not. In the negative case, we are done. Otherwise, compute in polynomial time a winning nice greedy strategy $\sigma$
(or turn a given strategy into a nice one), and compute its component graph $G_c(\sigma)$.
Make a copy $G'=(N',A')$ of $G_c(\sigma)$, and note that $G'$ is a directed acyclic graph, because it encodes a winning strategy.

Let $\overrightarrow N =v_1,\dots,v_{|N'|}$ be the topologically ordered sequence of the nodes of $G'$, where the  nodes without outgoing arcs,
called leaves, are in the first positions, and the node without incoming arcs, its root, is at the last position. Note that leaves correspond
to capture configurations for the robber, while the root $v_{|N'|}=(\emptyset, \nodes(\HG_1))$ is associated with the starting configuration
$(\emptyset, \emptyset, \nodes(\HG_1))$ of the game. Moreover, if $(v,v')\in A'$, the node $v$ is said to be a parent of $v'$, while $v'$ is
said to be a child of $v$. Then, modify the graph $G'$, by navigating the sequence $\overrightarrow N$ using an index $j$.

Starting with $j=1$,  while $j< |N'|$, consider the current node $v_j$ in the sequence, associated with a configuration $(h_j,M_j,C_j)$
(initially, the first leaf) in the domain of $\sigma_{G'}$. If every child of $v_j$ is labeled by some $(h'',C'')$ with $C''\subseteq C_j$,
then let index $j := j+1$ and continue the ``while'' loop, or stop and output the current graph $G'$ if $v_j$ is the root. Otherwise, let $v_s$
be a child of $v_j$ labeled by $(h_s,C_s)\in N'$ such that $C_s\not\subseteq C_j$, and associated with the configuration $(h_s,M_s,C_s)$. That
is, $\sigma_{G'}(h_j,M_j,C_j)=(h_s,M_s)$ is a non-monotone move. Then, take any parent $v_p$ of $v_j$, and let $(h_p,M_p,C_p)$ the
configuration associated with $v_p$ (whose label is thus $(h_p,C_p)$). Modify the graph so that  $\sigma_{G'}(h_p,M_p,C_p)=(h_j,M_j')$, where
$M_j'=M_j\setminus \E(v_j,M_s)$. In particular, let $C'_j$ be the \component{M'_j} that properly includes $C_j$, and for which thus $C_p\cup
C'_j$ is \connected{M_p\cap M'_j}. Then, the modified component graph will also encode the choice $\sigma_{G'}(h_j,M_j',C_j')=(h_j,\partial
C_j')$ if $\partial C_j' \subset M_j'$,  and $\sigma_{G'}(h_j,\partial C_j',C_j')=(h_s,M_s)$. The transformation of the graph is as follows:

\begin{itemize}
\item[(i)] Add a node $v_j'$ labeled by $(h_j,C'_j)$ to $N'$ and to the sequence $\overrightarrow N$ in the position before $v_j$, and add
    to $A'$ an arc from $v_j'$  to each child of $v_j$, i.e., to nodes labeled by $(h_s,C'')$, for each  \ecomponent{C_j',M_s} $C''$.

\item[(ii)] Remove from $A'$ all outgoing arcs of $v_p$ to nodes whose labels do not contain \ecomponents{C_p,M_j'} (in particular, the arc
    towards $v_j$ is removed).

\item[(iii)] Add to $A'$ an arc from $v_p$ to $v_j'$.

\item[(iv)] Remove from $N'$ any node different from the root which is left without incoming arcs, and continue the ``while'' loop
    considering again node $v_j$, or the next available node in $\overrightarrow N$ if $v_j$ has been removed by $N'$.
\end{itemize}

\begin{figure}[t]
  \centering
  \includegraphics[width=0.9\textwidth]{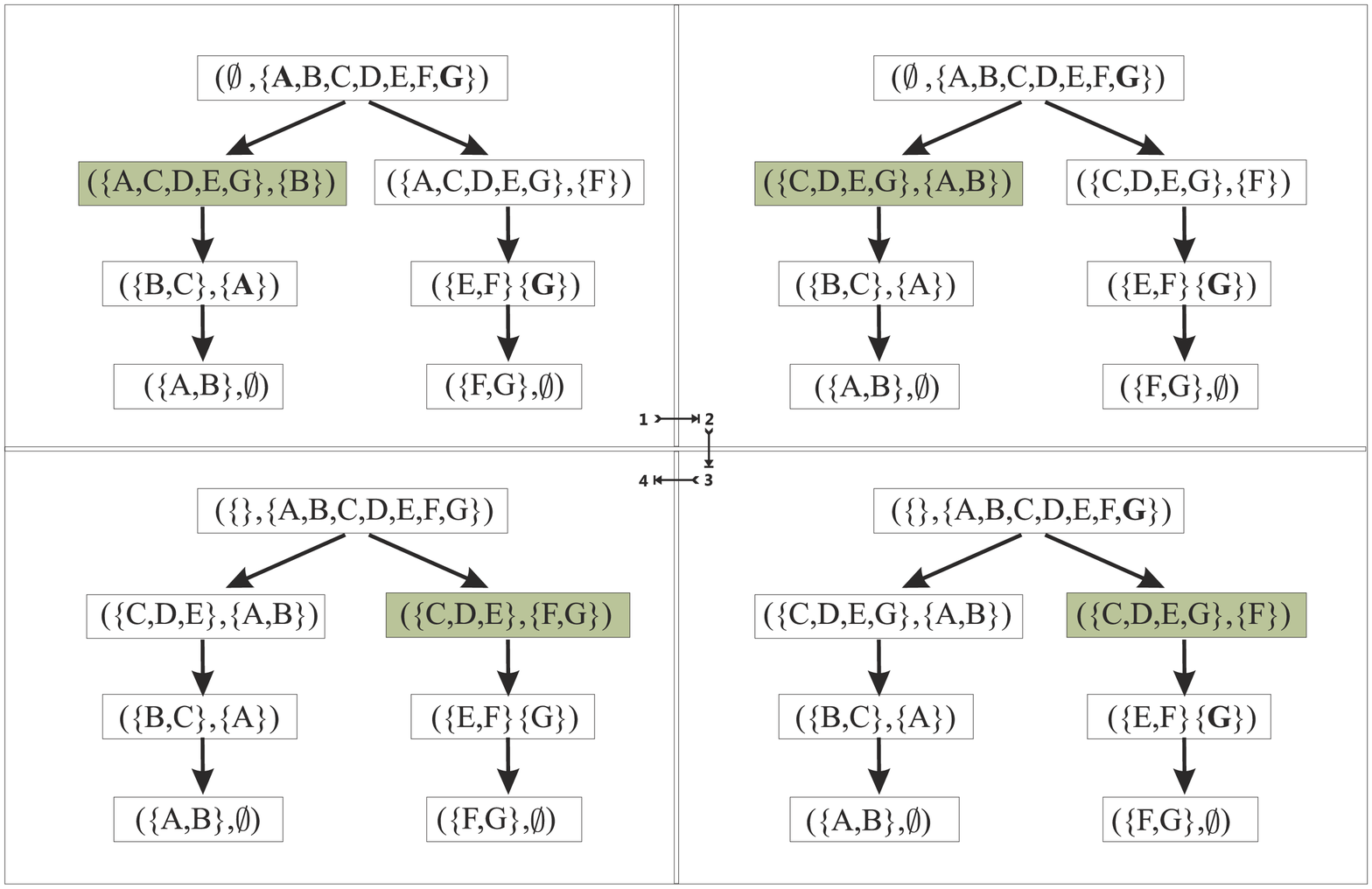}
  \caption{Illustration of the algorithm in the proof of Theorem~\ref{thm:greedy}.}\label{fig:greedybis}
\end{figure}

\begin{example}\em
The application of the above procedure to the nice strategy $\sigma_n$ discussed in Example~\ref{ex:nice} is illustrated in
Figure~\ref{fig:greedybis}. Note that two non-monotone moves are removed in total. Note that, at the end of the transformation, we get a
component graph encoding precisely the monotone strategy $\bar \sigma$, whose strategy graph has been illustrated in
Figure~\ref{fig:strategy-new-bis}. \hfill $\lhd$
\end{example}

First observe that every iteration of the loop at step~1 above, precisely implements on the graph $G'$ the transformation (of the non-monotone
strategy encoded by $G'$) described by Expression~\eqref{def:transformation}, and whose properties are described by
Lemma~\ref{lem:construction}. In more detail, with these properties in mind, by executing steps~(i)--(iii) we replace the Captain's choice
$(h_j,M_j)$ at $(h_p,M_p,C_p)$ by the new choice $(h_j,M_j')$, and we get the following situation: (a) Because of the new choice $M_j'$, only
one new $\ecomponent{C_p,M_j'}$ is available to the robber, that is, the $\component{M_j'}$ $C_j'$ properly including the $\component{M_j}$
$C_j$. As a consequence, at step~(i) the one node $v_j'$ corresponding to this component is added to $N'$.  (b) The $\ecomponents{C_j',M_s}$
are the same as the $\ecomponents{C_j,M_s}$, so that  the outgoing arcs of $v_j'$ will be the same as the node $v_j$. That is, we keep the same
winning strategy as before, as the Robber's options after the Captain's choice $M_s$  are the same as before (and hence the Captain knows how
to successfully attack them). (c) The set of $\ecomponents{C_p,M_j'}$, with the exception of the new $C_j'$, are a subset of the
$\ecomponents{C_p,M_j}$. In fact, some components may collapse after the new choice of the Captain. Then, at step~(iv), we remove the nodes
associated with $\ecomponents{C_p,M_j}$ that are now left without incoming arcs. For instance, it is possible that we delete $v_j$ if $v_p$ was
its only parent, or it is possible that we delete some nodes associated with collapsed components. Note that the new graph $G'$ obtained from
these steps is still a component graph, hence it encodes a (new) nice strategy $\sigma_{G'}$.

Therefore, Lemma~\ref{lem:construction} entails that, after each iteration and thus after the entire procedure, the strategy $\sigma_{G'}$ is a
winning strategy. We claim that it is actually a monotone winning strategy, by a simple inductive argument: if $v_j$ is the current node, after
the execution of steps (i)--(iv), $\sigma_{G'}$ is a monotone winning strategy for the game starting at the configuration $v_j$. Then, the
claim follows because, for $j=|N'|$, it means that  $\sigma_{G'}$ is a monotone winning strategy for the whole game starting at the root. The
base case is when the algorithm starts at $j=1$, and hence the statement holds because the first position in $\overrightarrow N$ is occupied by
some leaf, which is a capture configuration of the winning strategy.
Now assume that the statement holds for $j-1$, and consider the execution of the above procedure on node $v_j$. Note that the proposed
transformation deals with just one (possibly new) component $C_j'$ instead of the strictly smaller $C_j$; everything else in the strategy does
not change, in particular no node preceding $v_j$ in the topological order is affected by the transformation.
Then, the monotonicity of the strategy on the game starting at $v_j$ immediately follows from the induction hypothesis and from
Lemma~\ref{lem:construction}.(1), which says that $\E(v_j',M_s)=\emptyset$ and hence that this move is monotone, so that $C''\subseteq C'_j$,
for each $\ecomponent{C_j',M_{i+1}}$ $C''$.

Because each iteration in feasible in polynomial time, it just remains to show that the whole procedure requires at most polynomially many
iterations. To this end, note that whenever some node $v_j$ encodes a non-monotone move, one node $v_j'$ is added to $N'$ for each parent $v_p$
of $v_j$. Indeed, the node $v_j$ is considered again after the first iteration where it was evaluated, if it still has incoming arcs (see
step~(iv)). However, after steps (i)--(iv), $\sigma_{G'}$ is a monotone winning strategy for the game starting at the new configuration $v_j'$.
Therefore, no new node will be subject to further transformations in subsequent iterations along the given topological ordering of $N'$. It
follows that the number of iterations of the described procedure is bounded by $\nodes(G_c(\sigma))\times {\it MaxIn}$, where ${\it MaxIn}$ is
the largest in-degree over the nodes of $G_c(\sigma)$. Thus, the number of iterations is bounded by a polynomial in the size of the strategy
graph of the greedy winning strategy, which is in its turn polynomial in the size of $(\HG_1,\HG_2)$.

Finally, from the monotone winning strategy  $\sigma_{G'}$ encoded by the output $G'$ of the above procedure, a tree projection $\HG_a$ of
$(\HG_1,\HG_2)$ is immediately available. Just define $\nodes(\HG_a)=\nodes(\HG_1)$ and $\edges(\HG_a)= \{ M \mid \sigma_{G'}(v)=(h,M) \textit{
for some configuration } v \textit{ in the domain of } \sigma_{G'}\}$. See~\cite{GS08}, for more detail about such a relationship between
monotone strategies and tree projections.~\hfill~$\Box$
\end{proof}

With the above result in place, let $\C_{gtp}$ denote the class of all pairs $(Q,\V)$ such that there exists a greedy winning strategy $\sigma$
for the Captain in the game $\CR(\HG_Q,\HG_\V)$. As shown in the proof of Theorem~\ref{thm:greedy}, based on $\sigma$ a tree projection of
$\HG_Q$ w.r.t.~$\HG_\V$, which we call \emph{greedy tree projection}, can be computed in polynomial time. Therefore, the following is
immediately established.

\begin{corollary}\label{cor:greedyTP}
$\C_{gtp}$ is an island of tractability.
\end{corollary}

\subsection{Captain vs Marshal}

A related class of tractable pairs has been defined in~\cite{adler08} in terms of the \emph{Robber and Marshal game} played by one Marshal and
the Robber on the hypergraphs $(\HG_1,\HG_2)$. This game has been originally defined on a single hypergraph to characterize hypertree
decompositions~\cite{gott-etal-03}, and its natural extension to pairs of hypergraphs has been defined and studied in~\cite{adler08}.
The game is as follows. The Marshal may control one hyperedge of $\HG_2$, at each step.
The Robber stands on a node and can run at great speed along hyperedges of $\HG_1$; however, (s)he is not permitted to run through a node that
is controlled by the Marshal. Thus, a \emph{configuration} is a pair $(h,C)$, where $h$ is the hyperedge controlled by the Marshal, and $C$ is
an \component{h} where the Robber stands.
Let $(h_p,C_p)$ be a configuration. This is a capture configuration, where the Marshal wins, if $C_p\subseteq h_p$. Otherwise, the Marshal
moves to another hyperedge $h_{r}\in \edges(\HG_2)$; while (s)he moves, the Robber may run through those nodes that are left by the Marshal or
not yet occupied. Thus, the Robber selects an \component{h_{r}} $C_{r}$ such that $C_{r}\cup C_p$ is \connected{h_p\cap h_{r}}.
We say that the Marshal has a \emph{winning strategy} if, starting from the initial configuration $(\emptyset,\node)$, (s)he may end up the
game in a capture position, no matter of the Robber's moves. A winning strategy is \emph{monotone} if the Marshal may monotonically shrink the
set of nodes where the Robber stands.

Because only nodes in the frontier are actually used at each step in the monotone Robber and Marshal game, the monotone variants of the above
two games clearly define the same hypergraph properties.
\begin{fact}\label{prop:greedy-monotone}
The following are equivalent:
\begin{itemize}
  \item[(1)] There is a monotone winning strategy for the Marshal in the Robber and Marshal game on $(\HG_1,\HG_2)$.

  \item[(2)] There is a monotone winning greedy-strategy for the Captain in the Robber and Captain game on $(\HG_1,\HG_2)$.
\end{itemize}
\end{fact}

Let $\C_{rm}$ denote the class of all pairs $(Q,\V)$ such that there exists a  monotone winning strategy for the Marshal on $(\HG_Q,\HG_\V)$.
From the results in~\cite{adler08,adler-thesis}, $\C_{rm}$ is an island of tractability as well. However, the set of tractable instances
identified by greedy winning strategies in the Robber and Captain game properly includes this class. The reason is that greedy winning
strategies are allowed to be non-monotone.

\begin{theorem}\label{thm:comparazione}
$\C_{rm} \subset \C_{gtp}$.
\end{theorem}
\begin{proof}
Because greedy strategies are not required to be monotone, $\C_{rm} \subseteq \C_{gtp}$ follows from Fact~\ref{prop:greedy-monotone}.
For the proper inclusion, just consider again Example~\ref{ex:greedy-strat}. The pair of hypergraphs shown in Figure~\ref{fig:greedy-bis} is
such that the Marshal has no monotone winning strategy, while the Captain has a (non-monotone) winning greedy strategy.\footnote{This example
is in fact inspired by a similar simpler pair of hypergraphs where no monotone strategy for the Marshal exists, described
in~\cite{adler08}.}~\hfill~$\Box$
\end{proof}

For completeness, recall that the non-monotone variant of the Marshal and Robber game is instead too powerful to be useful. Indeed, there are
pairs of hypergraphs where the Marshal has a non-monotone winning strategy but no tree projection exists. We refer the interested reader
to~\cite{adler08} for more detail about the monotonicity gap in the Robber and Marshal game, and to~\cite{GS10} for a measure of distance
between non-monotone strategies in the Robber and Marshal game and tree projections.

\subsection{Greedy Decomposition Methods}

The tractability result about the general case of greedy tree projections can be immediately applied to every structural decomposition method,
in order to get new tractable variants of these methods.

Recall from Definition~\ref{def:dm} that a structural decomposition method {\tt DM} is a pair of polynomial-time computable functions $\lDM$
and $\rDM$ that, given a conjunctive query $Q$ and a database $\DB'$, compute a view system  $\V=\lDM(Q)$ and a database $\DB''=\rDM(Q,\DB')$
over the vocabulary of $\V$ that may be used to answer $Q$ on $\DB'$. In particular, the decompositions of $Q$ according to {\tt DM} are tree
projections of $\HG_Q$ w.r.t.~$\HG_\V$. Then, it is natural to consider the greedy variant of any structural decomposition method {\tt DM},
denoted by $\textit{greedy-}{\tt DM}$, whose associated decompositions are the \emph{greedy} tree projections of $\HG_Q$ w.r.t.~$\HG_\V$.

From Corollary~\ref{cor:greedyTP}, every decomposition method, possibly an intractable one such as the generalized hypertree decomposition
method, defines an island of tractability by means of its greedy variant.

\begin{fact}
Let {\tt DM} be a structural decomposition method and let $\textit{greedy-}{\tt DM}$ be its greedy variant. Then, the class of all queries
having a $\textit{greedy-}{\tt DM}$ decomposition is recognizable in polynomial time, and every query in the class may be evaluated in
polynomial time over any given database.
\end{fact}

For a notable example, consider the method based on generalized hypertree decompositions. Let $k\geq 1$. Recall that the width-$k$ generalized
hypertree decompositions of a query $Q$ are the tree projections of $(\HG_Q,\HG_Q^{k})$, as the view set $\it v\mbox{-}hw_k(Q)$ contains one
distinct view over each set of variables that can be covered by at most $k$ query-atoms. Then, the width-$k$ {\em greedy
hypertree-decompositions} (we omit ``generalized'', for short) of $Q$ are the greedy tree projections of $(\HG_Q,\HG_Q^{k})$. Accordingly, the
{\em greedy (generalized) hypertree-width} of $Q$, denoted by $\textit{gr-hw}$, is the smallest $k$ such that $Q$ has a greedy hypertree
decomposition.
In fact, this greedy variant provides a new tractable approximation of the (intractable) notion of generalized hypertree decomposition, which
is better than (standard) hypertree decompositions.

\begin{figure}[t]
  \centering
  \includegraphics[width=0.98\textwidth]{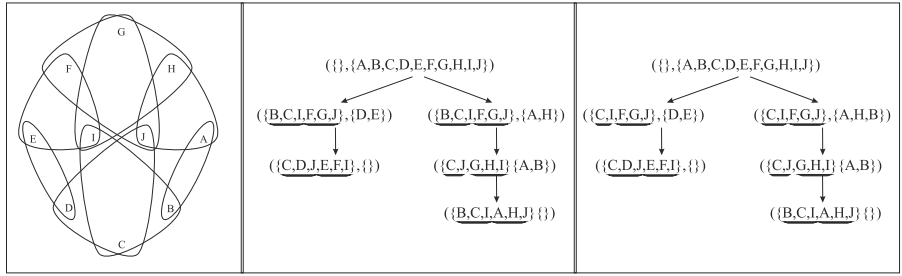}
  \caption{Examples in the proof of Fact~\ref{greedy-hw}.}
  \label{fig:GMS}
\end{figure}

\begin{fact}\label{greedy-hw}
For any query $Q$, $\textit{ghw}(Q)\leq \textit{gr-hw}(Q) \leq \textit{hw}(Q)$ holds. Moreover, there are queries $Q$ for which
$\textit{gr-hw}(Q) < \textit{hw}(Q)$, even for $\textit{gr-hw}(Q)=2$.
\end{fact}
\begin{proof}
The first relationship is immediate: in the first inequality we use the fact that greedy hypertree decompositions are a special case of
generalized hypertree decompositions, while the second inequality holds because the notion of hypertree decomposition is characterized by the
monotone Robber and Marshals game, played on $\HG_Q$ by a Robber and $k$ Marshals~\cite{gott-etal-03}. This game is equivalent to play the
monotone game with one Marshal on the pair of hypergraphs $(\HG_Q,\HG_Q^{k})$, which is the same as playing the monotone Robber and Captain
game.

For the strict upper bound  $\textit{gr-hw}(Q) < \textit{hw}(Q)$, consider the query $Q_0$, taken from~\cite{CJG08,GMS07}, whose hypergraph
$\HG_{Q_0}$  is depicted in the left part of Figure~\ref{fig:GMS}. For this query, it is shown in~\cite{GMS07} that $\textit{hw}(Q_0)=3$ and
$\textit{ghw}(Q_0)=2$. However, $\textit{gr-hw}(Q_0)=2$ holds. Indeed, there is a winning greedy strategy for the Captain in the game played on
$(\HG_{Q_0},\HG_{Q_0}^{2})$, as shown in the central part of Figure~\ref{fig:GMS}, and thus there exists a greedy tree projection of
$\HG_{Q_0}$ w.r.t.~$\HG_{Q_0}^{2}$.
In the figure, the set of selected cops at each step is underlined in such a way that the reader may identify the original pair of hyperedges
from $\HG_{Q_0}$ that forms the chosen squad in $\HG_{Q_0}^{2}$. Note that the strategy is non-monotone, as it is witnessed by the right branch
where the Robber can return on the node $B$. However, by using the construction in Theorem~\ref{thm:greedy}, it can be turned into a monotone
(while not greedy) one, by removing the escape door $B$ in the first move of the Captain (see the right part of the figure). From the monotone
strategy, we immediately get the desired tree projection.~\hfill~$\Box$
\end{proof}

More general examples are given by the {\em subedge-based} decomposition methods, defined in~\cite{GMS07}. Recall that a subedge-method ${\tt
DM}$ is based on a function $f$ associating with each integer $k \geq 1$ and each hypergraph $\HG_Q=(V,E)$ of some query $Q$ a set $f (\HG_Q,
k)$ of subedges of $\HG_Q$, that is, a set of subsets of hyperedges in $E$. Moreover, the set of width-$k$ ${\tt DM}$-decompositions of $Q$ can
be obtained as follows: (1) obtain a hypertree decomposition $\HD$ of $\HG_f = (V, E\cup f (\HG, k))$, and (2) convert $\HD$ into a generalized
hypertree decomposition of $\HG_Q$ by replacing each subedge $h\in f (\HG_Q, k)\setminus E$ occurring in $\HD$ by some hyperedge $h'\in E$ such
that$h\subseteq h'$ (which exists because $h$ is a subedge).

Because such a method is based on width-$k$ hypertree decompositions, in the tree projection framework it can be recast as follows. A width-$k$
${\tt DM}$-decomposition is any tree decomposition of $\HG_Q$ w.r.t.~$\HG_f^k$ associated with some monotone winning strategy of the Robber and
Marshal game on this pair of hypergraphs. On the other hand, according to its greedy variant $\textit{greedy-}{\tt DM}$, the width-$k$
decompositions are the greedy tree projections of $\HG_Q$ w.r.t.~$\HG_f^k$. It follows that the greedy variant of this method is more powerful,
in general.

\begin{fact}\label{fact:subedge}
Let ${\tt DM}$ be any {subedge-based} decomposition method. Let $k\geq 1$ and let $Q$ be a query. Then, a width-$k$ ${\tt DM}$-decomposition of
$Q$ exists only if a width-$k$ $\textit{greedy-}{\tt DM}$-decomposition of $Q$ exists. The converse does not hold, in general.
\end{fact}
\begin{proof}
The first entailment follows from Theorem~\ref{thm:comparazione}. The fact that the converse does not hold in general, follows from
Fact~\ref{greedy-hw}, because the hypertree decomposition method is a subedge-based method (based on the function
$f(\HG_Q,k)=\emptyset$).~\hfill~$\Box$
\end{proof}

This is a remarkable result, as in~\cite{GMS07} some examples of subedge-based decomposition methods, such as the {\em component hypertree
decompositions}, are shown to generalize most previous proposals of tractable structural decomposition methods, such as hypertree and
spread-cut decompositions (in fact, all of them, but the approximation of fractional hypertree decomposition, later introduced in~\cite{M09}).
From Fact~\ref{fact:subedge}, their greedy variants are even more powerful.

\section{Tractability of Tree Projections over Small Arity Structures}\label{smallArities}

In this (light) section, we consider the case of relational structures having small arity, which is a relevant special case in real-world
applications.

In fact, observe that any variable that is not involved in any join operation in a conjunctive query (that is, any variable that occurs in one
atom only) is irrelevant and may be projected out in a preprocessing phase. It follows that the {\em effective arity} to be considered in our
structural techniques is actually determined by the largest number of variables that any atom has in common with other atoms (i.e., those
variables involved in join operations), independently of the arity of the relations in the original database schema. This number is often
small, in practice.\footnote{In fact, it is easy to further generalize this line of reasoning, by considering as ``effective arity'' the
maximum cardinality over the hyperedges in the GYO-reduct of $\HG_Q$. (Recall that the GYO reduct  of a hypergraph is obtained by iteratively
removing nodes that occur in one hyperedge only and hyperedges included in other hyperedges, until no further removal is possible---see,
e.g.,~\cite{ullm-89}.)}

Therefore, it is interesting to investigate whether the general problem of computing a tree projection of a pair of hypergraphs is any easier
in the case of small arity structures (for the sake of presentation, we just consider here the standard structure arity, leaving to the
interested reader the straightforward extension to the above mentioned ``effective arity''). We next show that the problem is indeed in
polynomial-time for bounded-arity structures, and it is moreover {\em fixed-parameter tractable (FPT)}, if the arity is used as a parameter of
the problem. This is not difficult to prove, but it was never stated before (as far as we know), and we believe it is important to pinpoint
this tractability result.

Recall that a problem is FPT if there is an algorithm that solves the problem in {\em fixed-parameter polynomial-time}, that is, with a cost
$f(k) O(n^{O(1)})$, for some computable function $f$ that is applied to the parameter $k$ only. In other words, this algorithm not only runs in
polynomial time if $k$ is bounded by a fixed number, but it also exhibits a ``nice'' dependency on the parameter, because $k$ is not in the
exponent of the input size $n$. Let \textit{p-TP} be the problem of computing a tree projection of $\HG_Q$ w.r.t.~$\HG_{\V}$, for a given pair
$(Q,\V)$, parameterized by the maximum arity of the relations occurring in $(Q,\V)$.

\begin{theorem}\label{theo:fp}
The problem $\textit{p-TP}$ is fixed-parameter tractable.
\end{theorem}
\begin{proof}
Let $(Q,\V)$ be an input pair for $\textit{p-TP}$, let $(\HG_Q,\HG_{\V})$ be the pair of associated hypergraphs, and let $k$ be the parameter.

Compute the simplicial version $\HG_s$ of the hypergraph $\HG_{\V}$, that is, the hypergraph having the same set of nodes as $\HG_{\V}$, and
where $\edges(\HG_s)= \{ h'\neq\emptyset \mid h'\subseteq h, h\in \edges(\HG_{\V})\}$. Therefore, $\edges(\HG_s)$ contains all subsets of every
hyperedge of $\HG_{\V}$.
Clearly, $\HG_s$ can be computed in time $O(2^k \times |\edges(\HG_{\V})|)$, and the tree projections of $(\HG_Q,\HG_{\V})$ are the same as the
tree projections of $(\HG_Q,\HG_s)$. To conclude, observe that any tree projection of the latter pair can be computed in polynomial-time by
Theorem~\ref{thm:greedy} and the fact that, having a squad for every possible set of cops in any squad/hyperedge of $\HG_{\V}$, the greedy
strategies in the game $\CR(\HG_Q,\HG_s)$ are precisely the (unrestricted) strategies in the game $\CR(\HG_Q,\HG_{\V})$, which characterize the
tree projections of $(\HG_Q,\HG_{\V})$.\footnote{Note that the same relationship holds for the monotone strategies and, hence, for the
Marshal's strategies in the Robber and Marshal game over the pair $(\HG_Q,\HG_s)$, as observed by Adler~\cite{adler-thesis}.}~\hfill~$\Box$
\end{proof}

The above tractability result is smoothly inherited by all structural decomposition methods {\tt DM} such that the arity of the views in $\lDM$
is $O(f(k))$ for some computable function $f$ that does not depend on the size of the input. For instance, this is the case for the methods
based on bounded (generalized hyper)tree decompositions, but not for fractional hypertree decompositions. In particular, if $w$ is the fixed
maximum width for a class of queries having bounded generalized hypertree width,  the maximum arity of the computed views is $w\times k$. Thus,
if $\textit{p-ghw}_w$ denotes the problem of computing a width-$w$ generalized hypertree decomposition of a query, parameterized by the maximum
arity of the query atoms, we immediately get the following result.

\begin{corollary}
The problem $\textit{p-ghw}_w$ is fixed-parameter tractable.
\end{corollary}

We believe that this is a useful result. Indeed, even if for queries $Q$ having maximum arity $k$ we have $ghw(Q)\leq tw(Q) \leq k \times
ghw(Q)$, we know that the problem of evaluating queries is not fixed-parameter tractable, with respect to the (generalized hyper)treewidth
parameter. It follows that, under usual fixed-parameter complexity assumptions, an exponential dependency on such width parameters is
unavoidable, hence evaluating such queries has a cost of the form $O(n^{f(w)})$, where $w$ is the treewidth (or the hypertree width) and $n$ is
the combined size of the database and the query (which is typically largely dominated by the size of the database). We thus argue that
employing generalized hypertree width instead of treewidth provides an exponential saving in the query-evaluation time, in general, and it is
convenient even for small arity instances. Moreover, recall that the computation of the decomposition depends on the hypergraph only (and not
on the database) and, unlike other fixed-parameter algorithms, the algorithm described in Theorem~\ref{theo:fp} is ``practical,'' as there are
no huge constants and the dependence on the arity parameter is single-exponential.

\section{Conclusion}\label{sec:conclusion}

In this paper, we have fully characterized the power of algorithms for evaluating conjunctive queries (and constraint satisfaction problems)
based on enforcing local consistency. We studied both the general framework where consistency is enforced over arbitrary views and  the more
specific cases where views are computed according to structural decomposition methods.
These results have already found application to the problems of enumerating query answers~\cite{GS10} and computing optimal
solutions~\cite{GS11}.

In addition to the questions mentioned in the Introduction, it is worthwhile recalling another open question that eventually finds an answer
with these results. The question was raised in~\cite{GS84}, where the \emph{tree projection theorem} was proved.
Roughly, a query program $P$ is a finite sequence of steps involving project, select and join operations. The relation computed in the final
step is the result of $P$. The {tree projection theorem} states that a query program $P$ solves a query $Q$ (i.e., the result of $P$ always
coincides with the answers of $Q$ over its set of output variables) if, and only if, there is a tree projection of $Q$ w.r.t. the hypergraph
associated with the various relations/views determined by $P$. A crucial point here is that $P$ is a fixed program, so that the number of its
operations does not depend on the database size. The natural question in~\cite{GS84} was therefore to ask what happens if $P$ is allowed to
contain a ``semijoin loop,'' that is, a loop that is to be executed until nothing changes in the involved relations/views. Is it the case that
the tree projection theorem still holds for such programs, where the number of steps is data-dependent?
The results in the paper provide a positive answer to this question for the setting of simple queries (implicitly) considered in~\cite{GS84}
and, in fact, also a complete answer covering the general case where queries may contain more atoms over the same relation symbol.

Finally, by exploiting a recent hypergraph-game characterization of tree projections, we also identified new islands of (structural)
tractability, and we pinpointed the fixed-parameter tractability of tree projections and of (most) structural decomposition methods when small
arity structures are considered. We believe that such results may be very useful in practical applications, and we are currently working on
direct implementations of the proposed techniques in real-world database management systems.

There are still a number of interesting questions to be answered about structural decomposition methods. For instance, even for the bounded
arity case, the frontier of tractability for the problem of enumerating with polynomial delay the answers of a conjunctive query $Q$ over a
given arbitrary set of output variables is not known (see~\cite{G07,BDGM09}). Moreover, in the general unbounded-arity case, the frontier of
tractability is not known even for Boolean conjunctive queries. In fact, in the unbounded arity case, the notion of submodular width~\cite{M10}
allows us to identify the class of conjunctive queries that are fixed-parameter tractable (where the parameter is the size of the query),
assuming the exponential-time hypothesis. As a consequence, we now have an interesting gap to be explored between the polynomial-time
tractability of instances having bounded fractional hypertree width~\cite{GM06,M09} and the fixed-parameter tractability of instances having
bounded submodular width.


\begin{thebibliography}{10}

\bibitem{abit-etal-95} S.~Abiteboul, R.~Hull, and V.~Vianu.
\newblock {\em Foundations of Databases}.
\newblock Addison-Wesley, 1995.

\bibitem{adler04} I. Adler.
\newblock Marshals, monotone marshals, and hypertree-width.
\newblock {\em Journal of Graph Theory}, 47(4), pp. 275--296, 2004.

\bibitem{adler-thesis} I. Adler.
\newblock Width Functions for Hypertree Decompositions.
\newblock {\em PhD Thesis}, University of Freiburg, 2006.

\bibitem{adler08} I. Adler.
\newblock Tree-Related Widths of Graphs and Hypergraphs.
\newblock {\em SIAM Journal Discrete Mathematics}, 22(1), pp. 102--123, 2008.

\bibitem{AGG07} I. Adler, G. Gottlob, and M. Grohe.
\newblock Hypertree-Width and Related Hypergraph Invariants.
\newblock{ \em European Journal of Combinatorics}, 28, pp. 2167--2181, 2007.

\bibitem{ABD07} A. Atserias, A. Bulatov, and V. Dalmau.
\newblock On the Power of k-Consistency,
\newblock In {\em Proc. of ICALP'07}, pp. 279--290, 2007.

\bibitem{BFMY83} C. Beeri, R. Fagin, D. Maier, and M. Yannakakis.
\newblock On the Desirability of Acyclic Database Schemes.
\newblock {\em Journal of the ACM}, 30(3), pp. 479--513, 1983.

\bibitem{bern-good-81} P.A.~Bernstein and N.~Goodman.
\newblock The power of natural semijoins.
\newblock {\em SIAM Journal on Computing}, 10(4), pp. 751--771, 1981.

\bibitem{bodl-96} H.L.~Bodlaender and F.V.~Fomin.
\newblock A Linear-Time Algorithm for Finding Tree-Decompositions of Small Treewidth.
\newblock {SIAM Journal on Computing}, 25(6), pp. 1305-1317, 1996.

\bibitem{BDGM09} A. Bulatov, V. Dalmau, M. Grohe, and D. Marx.
\newblock Enumerating Homomorphisms.
\newblock {\em Journal of Computer and System Sciences}, 78(2), pp. 638-650, 2012.

\bibitem{chan-etal-81} A.K. Chandra, D.C. Kozen, and L.J. Stockmeyer. Alternation. {\em Journal of the ACM}, 26:114--133, 1981.

\bibitem{CD05} H. Chen and V. Dalmau.
\newblock Beyond Hypertree Width: Decomposition Methods Without Decompositions.
\newblock In {\em Proc. of CP'05}, pp. 167--181, 2005.

\bibitem{CJ06} D. A. Cohen and P. Jeavons.
\newblock The Complexity of Constraint Languages.
\newblock In {\em Handbook of Constraint Programming},
F. Rossi, P. van Beek, and T. Walsh, Eds., Elsevier, 2006.

\bibitem{CJG08} D. A. Cohen, P. Jeavons, and M. Gyssens.
\newblock A unified theory of structural tractability for constraint satisfaction problems.
\newblock  {\em Journal of Computer and System Sciences}, 74(5), pp. 721-743, 2008.

\bibitem{DKV02} V. Dalmau, Ph.G. Kolaitis, and M.Y. Vardi.
\newblock Constraint Satisfaction, Bounded Treewidth, and Finite-Variable Logics.
\newblock In {\em Proc. of CP'02}, pp. 310--326, 2002.

\bibitem{dech-03} R.~Dechter.
\newblock Constraint Processing.
\newblock Morgan Kaufmann, 2003.

\bibitem{DP89} R. Dechter and J. Pearl.
\newblock Tree clustering for constraint networks.
\newblock {\em Artificial Intelligence}, pp. 353--366, 1989.

\bibitem{fagi-83} R. Fagin.
\newblock Degrees of acyclicity for hypergraphs and relational database schemes.
\newblock {\em Journal of the ACM}, 30(3):514--550, 1983.

\bibitem{FFG02} J. Flum, M. Frick, and M. Grohe.
\newblock Query evaluation via tree-decompositions.
\newblock {\em Journal of the ACM}, 49(6):716--752, 2002.

\bibitem{FN06} P.~Fraigniaud and N.~Nisse.
\newblock Connected Treewidth and Connected Graph Searching.
\newblock In {\em Proc. of LATIN'06}, pp. 479--490, 2006.

\bibitem{Fre90} E.C. Freuder.
\newblock Complexity of K-tree structured constraint satisfaction problems.
\newblock In {\em Proc. of the 8th National Conference
on Artificial Intelligence}, pp. 4--9, 1990.

\bibitem{gott-etal-00} G.~Gottlob, N.~Leone, and F.~Scarcello.
\newblock A Comparison of Structural CSP Decomposition Methods.
\newblock \emph{Artificial Intelligence}, 124(2), 243--282, 2000.

\bibitem{gott-etal-01} G.~Gottlob, N.~Leone, and F.~Scarcello.
\newblock The complexity of acyclic conjunctive queries.
\newblock {\em Journal of the ACM}, 48(3), pp. 431--498, 2001.

\bibitem{gott-etal-99} G.~Gottlob, N.~Leone, and F.~Scarcello.
\newblock Hypertree decompositions and tractable queries.
\newblock {\em Journal of Computer and System Sciences}, 64(3), pp. 579--627, 2002.

\bibitem{gott-etal-03} G.~Gottlob, N.~Leone, and F.~Scarcello.
\newblock Robbers, marshals, and guards: game theoretic and logical characterizations of hypertree width.
\newblock {\em Journal of Computer and System Sciences}, 66(4), pp. 775--808, 2003.

\bibitem{GMS07} G. Gottlob, Z. Mikl\'os, and T. Schwentick.
\newblock Generalized hypertree decompositions: $\NP$-hardness and tractable variants.
\newblock {\em Journal of the ACM}, 56(6), 2009.

\bibitem{GN08} G. Gottlob and A. Nash.
\newblock Efficient Core Computation in Data Exchange.
\newblock {\em Journal of the ACM}, 55(2), Article 9, 2008.

\bibitem{goodman83synatctic} N.~Goodman and O.~Shmueli.
\newblock Syntactic characterization of tree database schemas.
\newblock {\em Journal of the ACM}, 30(4):767--786, 1983.

\bibitem{GS84} N. Goodman and O. Shmueli.
\newblock The tree projection theorem and relational query processing.
\newblock {\em Journal of Computer and System Sciences}, 29(3), pp. 767--786, 1984.

\bibitem{GS08} G. Greco and F. Scarcello. Tree Projections: Hypergraph Games and Minimality.
\newblock In {\em Proc. of ICALP'08}, pp. 736--747, 2008. Full version available as CoRR technical report at http://arxiv.org/abs/1212.2314.


\bibitem{GS10b} G. Greco and F. Scarcello.
\newblock The Power of Tree Projections: Local Consistency, Greedy Algorithms, and Larger Islands of Tractability.
\newblock In {\em Proc. of PODS'10}, pp. 327--338, 2010.

\bibitem{GS10} G. Greco and F. Scarcello.
\newblock Structural Tractability of Enumerating CSP Solutions.
\newblock In {\em Proc. of CP'10}, pp. 236--251, 2010.

\bibitem{GS11} G. Greco and F. Scarcello.
\newblock Structural Tractability of Constraint Optimization.
\newblock In {\em Proc. of CP'11}, pp. 340--355, 2011.

\bibitem{G07} M.~Grohe.
\newblock The complexity of homomorphism and constraint satisfaction problems seen from the other side.
\newblock {\em Journal of the ACM}, 54(1), 2007.

\bibitem{GM06} M.~Grohe and D.~Marx.
\newblock Constraint solving via fractional edge covers.
\newblock In {\em Proc. of SODA '06}, pp. 289--298, 2006.


\bibitem{JK84} D. Johnson and A. Klug.
\newblock Testing containment of Conjunctive Queries Under Functional and Inclusion Dependencies.
\newblock {\em Journal of Computer and System Sciences}, 28(1), pp. 167--189, 1984.

\bibitem{john-90} D.S. Johnson, A Catalog of Complexity Classes, \emph{Handbook of Theoretical Computer Science, Volume A: Algorithms and
    Complexity}, pp. 67-161, 1990.

\bibitem{K03} Ph.G. Kolaitis.
\newblock Constraint Satisfaction, Databases, and Logic.
\newblock In {\em Proc. of IJCAI'03}, pp. 1587--1595, 2003.

\bibitem{LS99} A.~Lustig and O.~Shmueli.
\newblock Acyclic hypergraph projections.
\newblock {\em Journal of Algorithms}, 30(2):400--422, 1999.

\bibitem{M09} D. Marx.
\newblock Approximating fractional hypertree width.
\newblock {\em ACM Transactions on Algorithms}, 6(2), 2010.

\bibitem{M10} D. Marx. Tractable Hypergraph Properties for Constraint Satisfaction and Conjunctive Queries. In \emph{Proc. of
STOC'10}, pp. 735--744, 2010.

\bibitem{RS84} N. Robertson and P.D. Seymour.
\newblock Graph minors III: Planar tree-width.
\newblock {\em Journal of Combinatorial Theory, Series B}, 36, pp. 49--64, 1984.

\bibitem{R06} R. Rosati.
\newblock On the finite controllability of conjunctive query answering in databases under open-world assumption.
\newblock {\em Journal of Computer and System Sciences}, 77(3):572--594, 2011.

\bibitem{SS93} Y. Sagiv and O. Shmueli.
\newblock Solving Queries by Tree Projections.
\newblock {\em ACM Transaction on Database Systems}, 18(3), pp. 487--511, 1993.

\bibitem{rein-04} O.~Reingold.
\newblock Undirected ST-connectivity in log-space,
\newblock {\em Journal of the ACM}, 55(4), 2008.

\bibitem{ruzz-80} W.L. Ruzzo. Tree-size bounded alternation. {\em Journal of Cumputer and System Sciences}, 21, pp. 218-235, 1980.

\bibitem{ST93} P.D. Seymour and R. Thomas.
\newblock Graph searching and a min-max theorem for tree-width.
\newblock {\em Journal of Combinatorial Theory, Series B}, 58, pp. 22--33, 1993.

\bibitem{SGG08} F. Scarcello, G. Gottlob, and G. Greco.
\newblock Uniform Constraint Satisfaction Problems and Database Theory.
\newblock In {\em Complexity of Constraints}, LNCS 5250, pp. 156--195, Springer-Verlag, 2008.

\bibitem{SH07} S.~Subbarayan and H.~Reif Andersen.
\newblock Backtracking Procedures for Hypertree, HyperSpread and Connected Hypertree Decomposition of CSPs.
\newblock In {\em Proc. of IJCAI'07}, pp. 180--185, 2007.

\bibitem{tarj-yann-84} R.E. Tarjan, and M. Yannakakis.
\newblock Simple linear-time algorithms to test chordality of graphs, test acyclicity
of hypergraphs, and selectively reduce acyclic hypergraphs.
\newblock {\em SIAM Journal on Computing}, 13(3):566-579, 1984.

\bibitem{ullm-89} J.~D. Ullman.
\newblock {\em Principles of Database and Knowledge Base Systems}.
\newblock Computer Science Press, 1989.

\bibitem{yann-81} M.~Yannakakis.
\newblock Algorithms for acyclic database schemes.
\newblock In {\em Proc. of VLDB'81}, pp. 82--94, 1981.


\end{thebibliography}
\end{document}